\documentclass[acmtog,authorversion,nonacm]{acmart} 
\renewcommand\footnotetextcopyrightpermission[1]{} 
\newcommand{\notforarxiv}[1]{}

\newcommand{\blueit}[1]{#1}

\newcommand{\bluevar}[1]{#1}
\newcommand{\move}{}
\newcommand{\remove}{}
\newcommand{\redit}[1]{#1}
\newcommand{\redr}[1]{#1}
\newcommand{\redvar}[1]{#1}
\newcommand{\newremove}{}

\usepackage[utf8]{inputenc}

\usepackage[utf8]{inputenc}
\usepackage[ruled,vlined,linesnumbered]{algorithm2e}   
\usepackage{listings}                           
\usepackage{newtxmath}                          
\usepackage{mathdots}
\usepackage{mathtools}

\usepackage {enumitem}
\usepackage{circuitikz}
\usepackage{tikz}
\usepackage{calc}
\newlength{\unit}
\newlength{\gap}
\newlength{\labelHeight}
\usepackage{overpic}

\newtheorem{definition}{Definition}[section]            
\newtheorem{corollary}{Corollary}[section]            
\definecolor{clrAG}{HTML}{7B3F00}                      

\usepackage{soul}
\definecolor{clrFigBox1}{RGB}{43,57,144}
\definecolor{clrFigBox2}{RGB}{0,165,171}
\definecolor{clrFigBox3}{RGB}{69,181,57}
\newcommand{\figColBoxA}[1]{{\color{clrFigBox1} #1}}
\newcommand{\figColBoxB}[1]{{\color{clrFigBox2} #1}}
\newcommand{\figColBoxC}[1]{{\color{clrFigBox3} #1}}

\newcommand{\cc}{}
\newcommand{\gr}{}
\newcommand{\oh}{
\begingroup\color{lightgray}0\endgroup}

\newcommand\scalemath[2]{\scalebox{#1}{\mbox{\ensuremath{\displaystyle #2}}}}

\author{Abdalla G. M. Ahmed}
\orcid{0000-0002-2348-6897}
\affiliation{%
  \institution{KAUST}
  \country{KSA}}
\email{abdalla_gafar@hotmail.com}

\author{Mikhail Skopenkov}
\orcid{0000-0003-2453-0009}
\affiliation{%
  \institution{KAUST}
  \country{KSA}}
\email{mikhail.skopenkov@gmail.com}

\author{Markus Hadwiger}
\orcid{0000-0003-1239-4871}
\affiliation{%
  \institution{KAUST}
  \country{KSA}}
\email{markus.hadwiger@kaust.edu.sa}

\author{Peter Wonka}
\orcid{0000-0003-0627-9746}
\affiliation{%
  \institution{KAUST}
  \country{KSA}}
\email{pwonka@gmail.com}

\title{Analysis and Synthesis of Digital Dyadic Sequences}

\date{February 2021}

\setcopyright{rightsretained}
\acmJournal{TOG}
\acmYear{2023} \acmVolume{1} \acmNumber{1} \acmArticle{1} \acmMonth{1} \acmPrice{}\acmDOI{10.1145/3618308}

\begin{abstract}
We explore the space of matrix-generated $(0, m, 2)$-nets and $(0, 2)$-sequences in base 2, also known as digital dyadic nets and sequences.
\blueit{In computer graphics, they are arguably leading the competition for use in rendering.}
We provide a complete characterization of the design space and count the possible number of constructions with and without considering possible reorderings of the point set. 
Based on this analysis, we then show that every digital dyadic net can be reordered into a sequence, together with a corresponding algorithm.
Finally, we present a novel family of self-similar digital dyadic sequences, to be named $\xi$-sequences, that spans a subspace with fewer degrees of freedom.
Those $\xi$-sequences are extremely efficient to sample and compute, and we demonstrate their advantages over the classic Sobol $(0, 2)$-sequence.
\end{abstract}



\begin{teaserfigure}
    {\centering\scriptsize
    \begin{tabular*}{1\textwidth}{@{}c@{\extracolsep{\fill}}c@{\extracolsep{\fill}}c@{\extracolsep{\fill}}c@{\extracolsep{\fill}}c@{\extracolsep{\fill}}c@{}}
        \includegraphics[width=1\unit]{images/teaser2/Sobol-matrices.pdf}&%
        \includegraphics[width=1\unit]{images/teaser2/Hammersley-matrices.pdf}&%
        \includegraphics[width=1\unit]{images/teaser2/LP-matrices.pdf}&%
        \includegraphics[width=1\unit]{images/teaser2/Gray-matrices.pdf}&%
        \includegraphics[width=1\unit]{images/teaser2/xi-matrices.pdf}&%
        \includegraphics[width=1\unit]{images/teaser2/random-matrices.pdf}\\
        \includegraphics[width=1\unit]{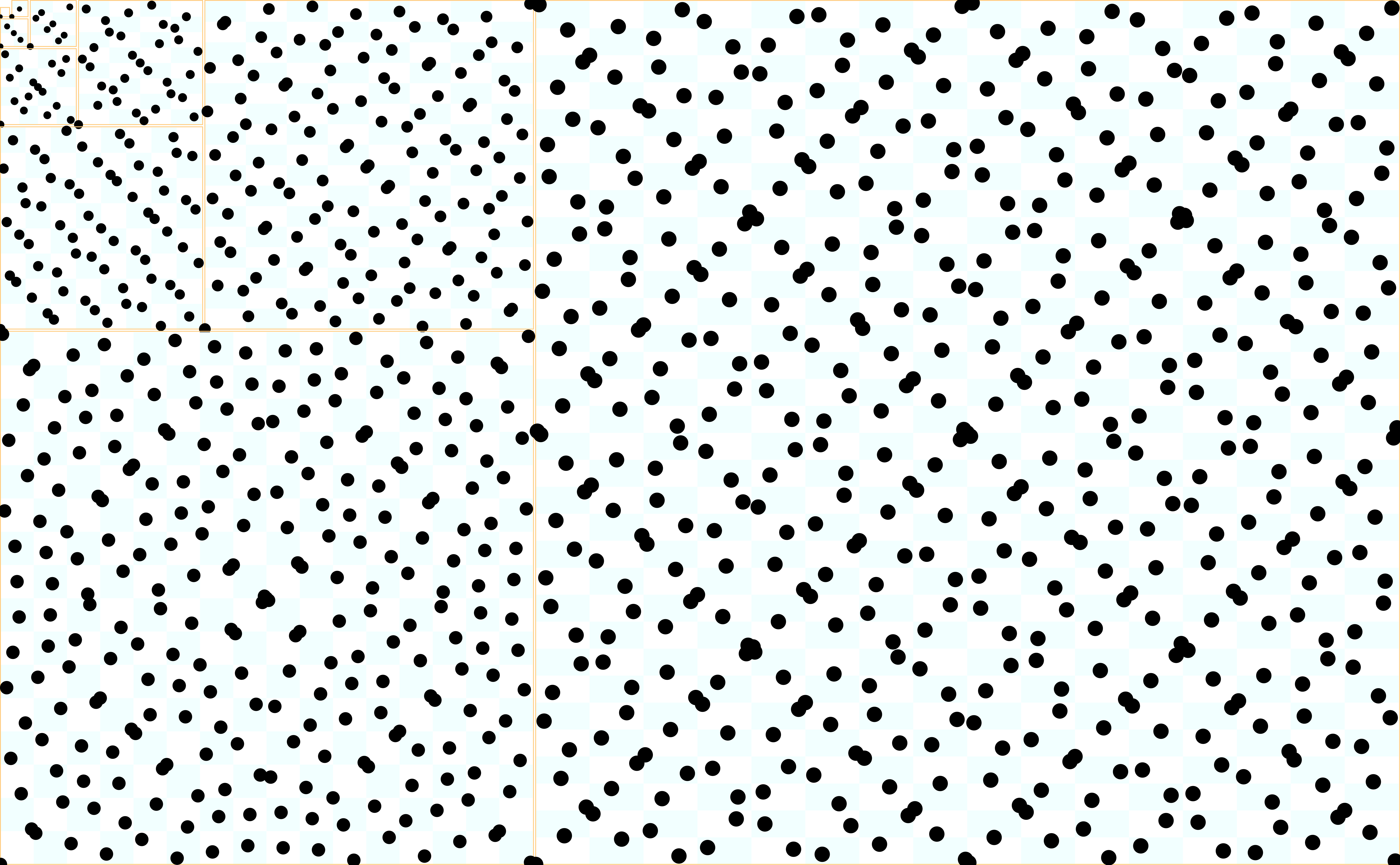}&%
        \includegraphics[width=1\unit]{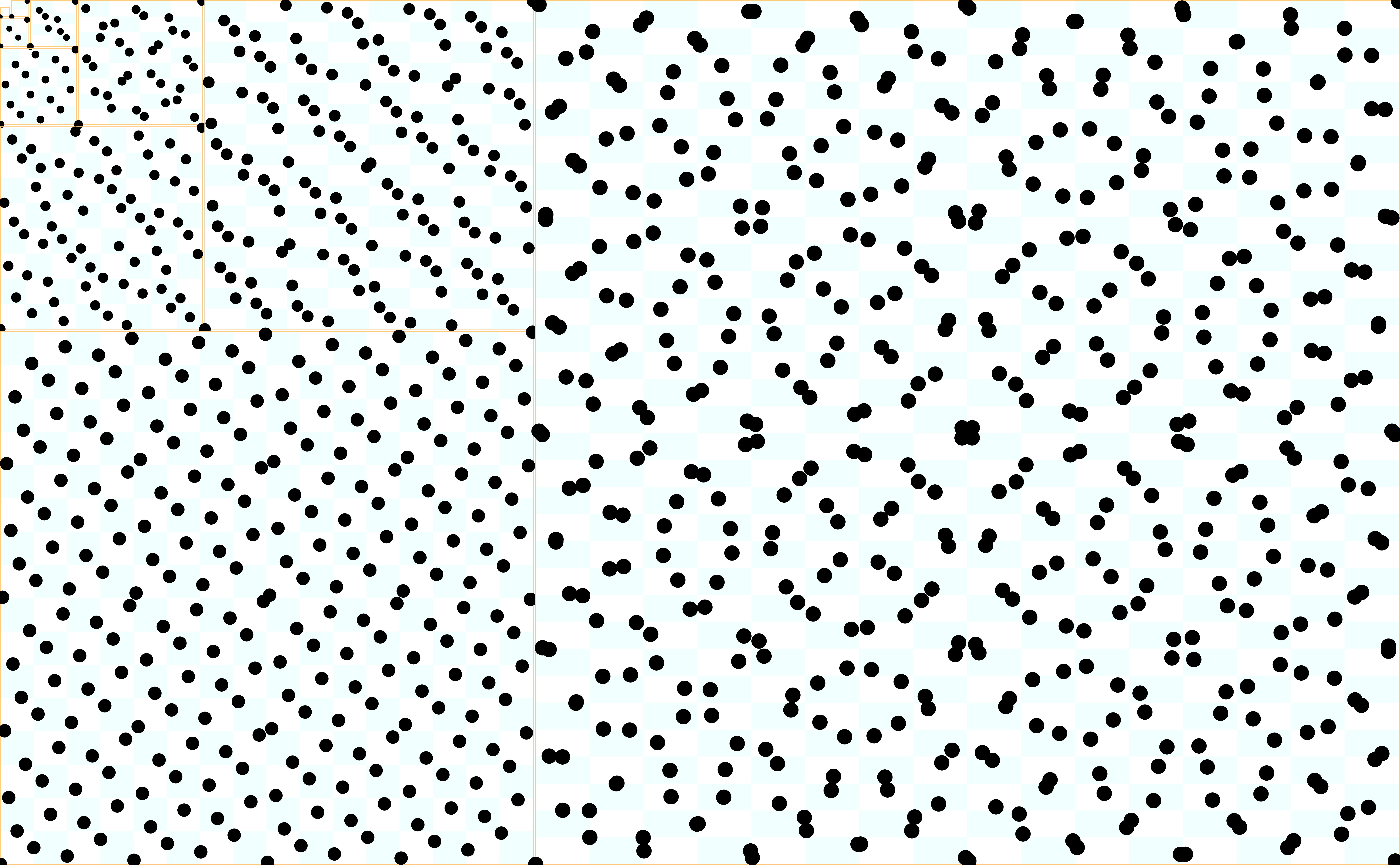}&%
        \includegraphics[width=1\unit]{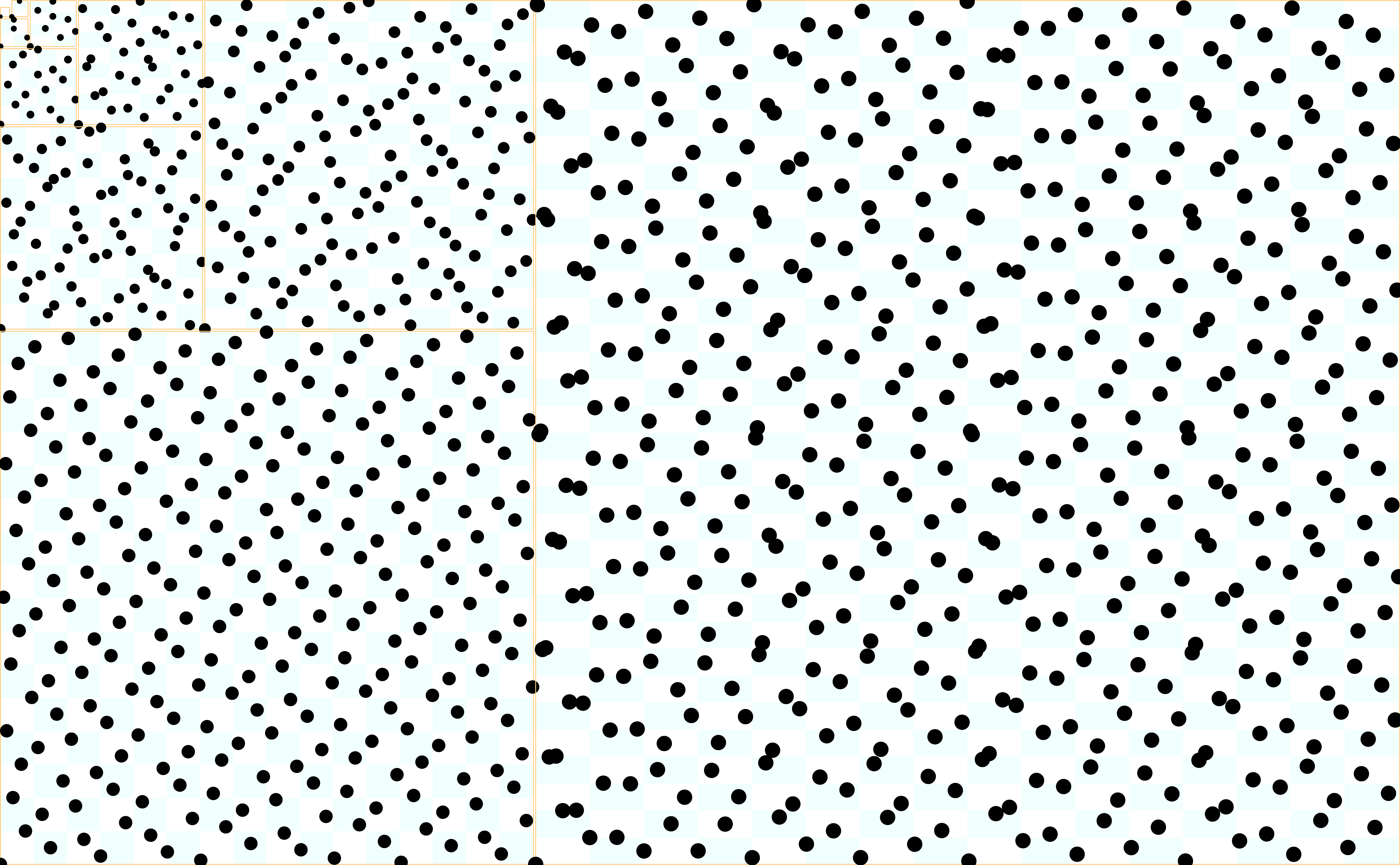}&%
        \includegraphics[width=1\unit]{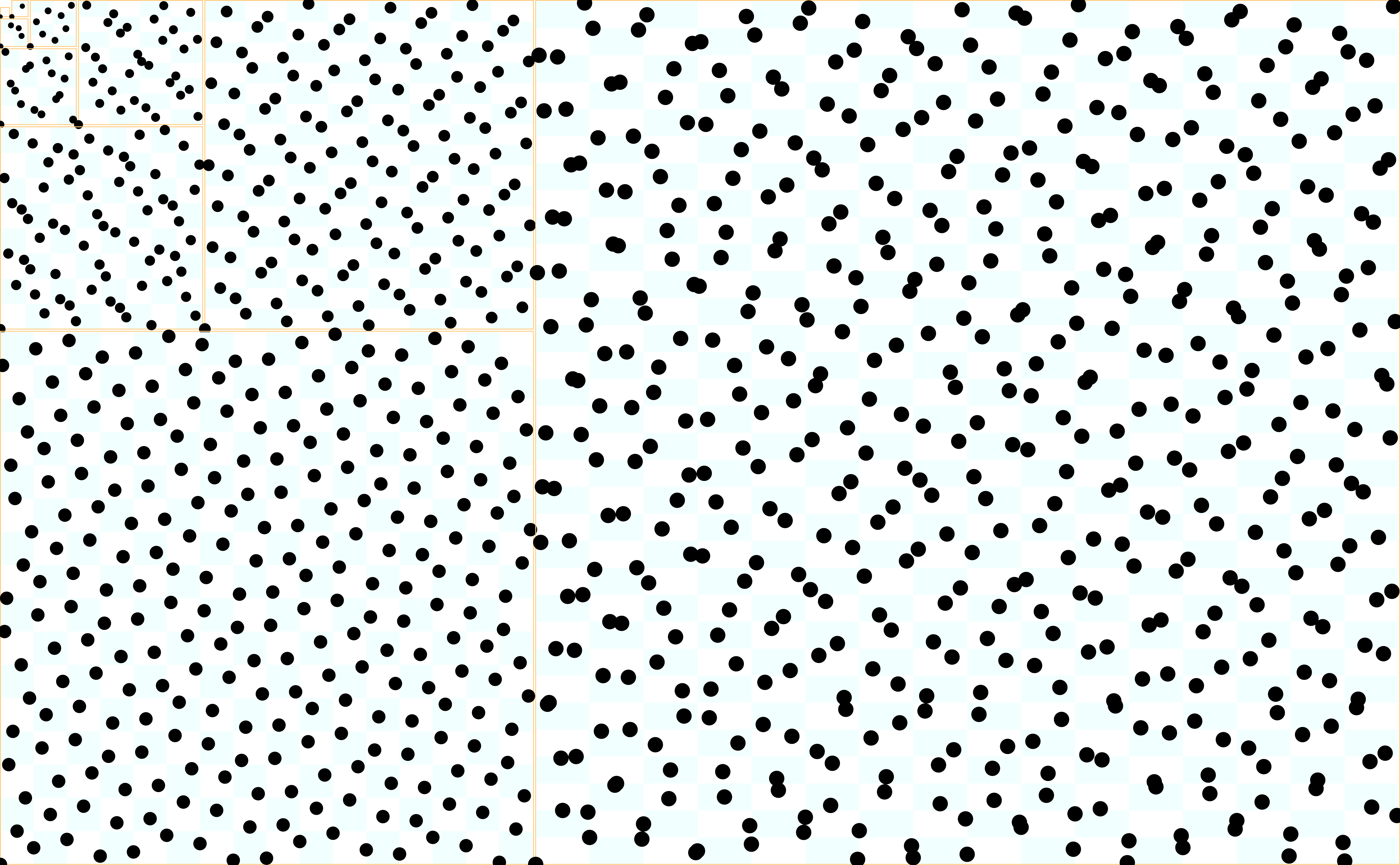}&%
        \includegraphics[width=1\unit]{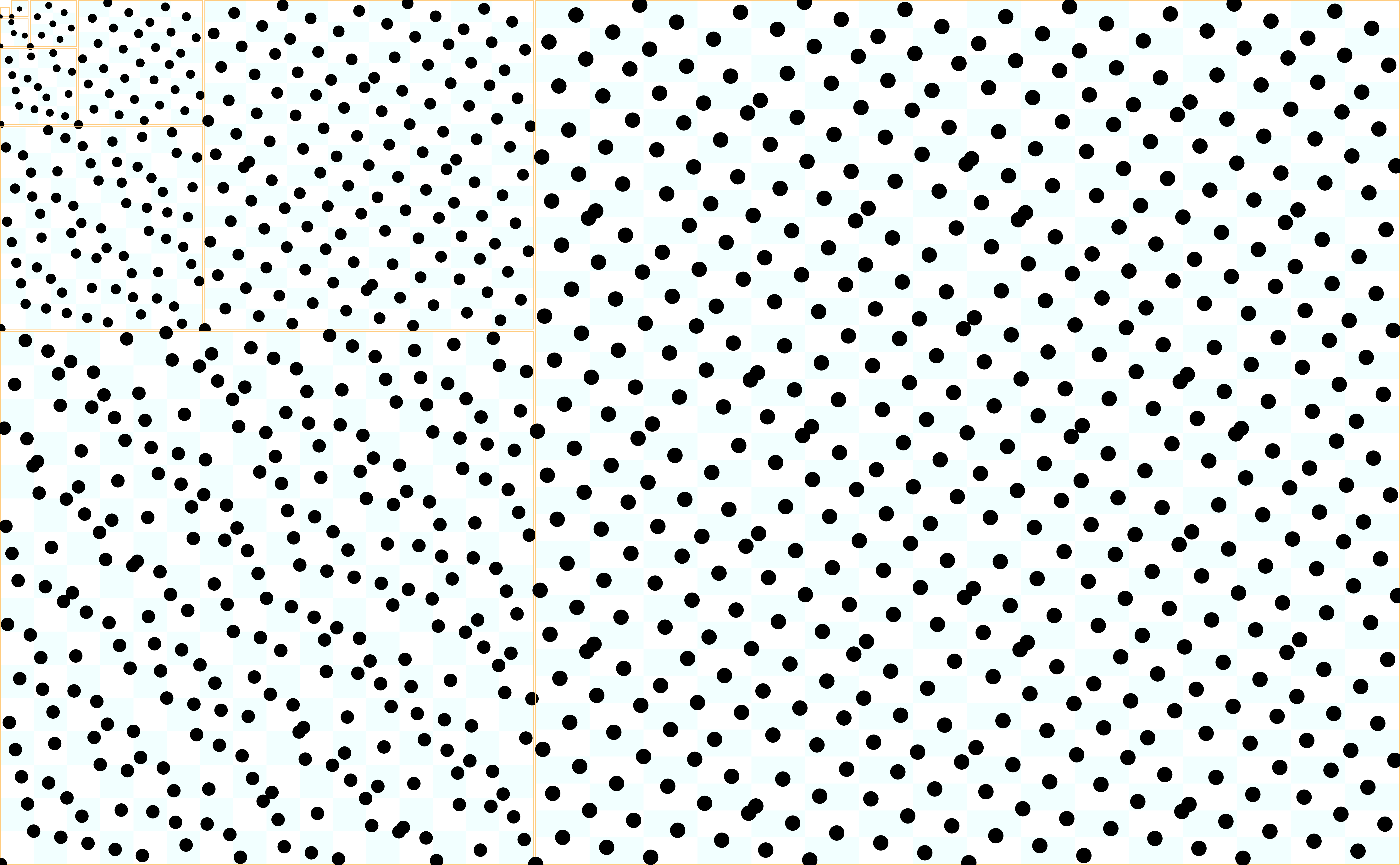}&%
        \includegraphics[width=1\unit]{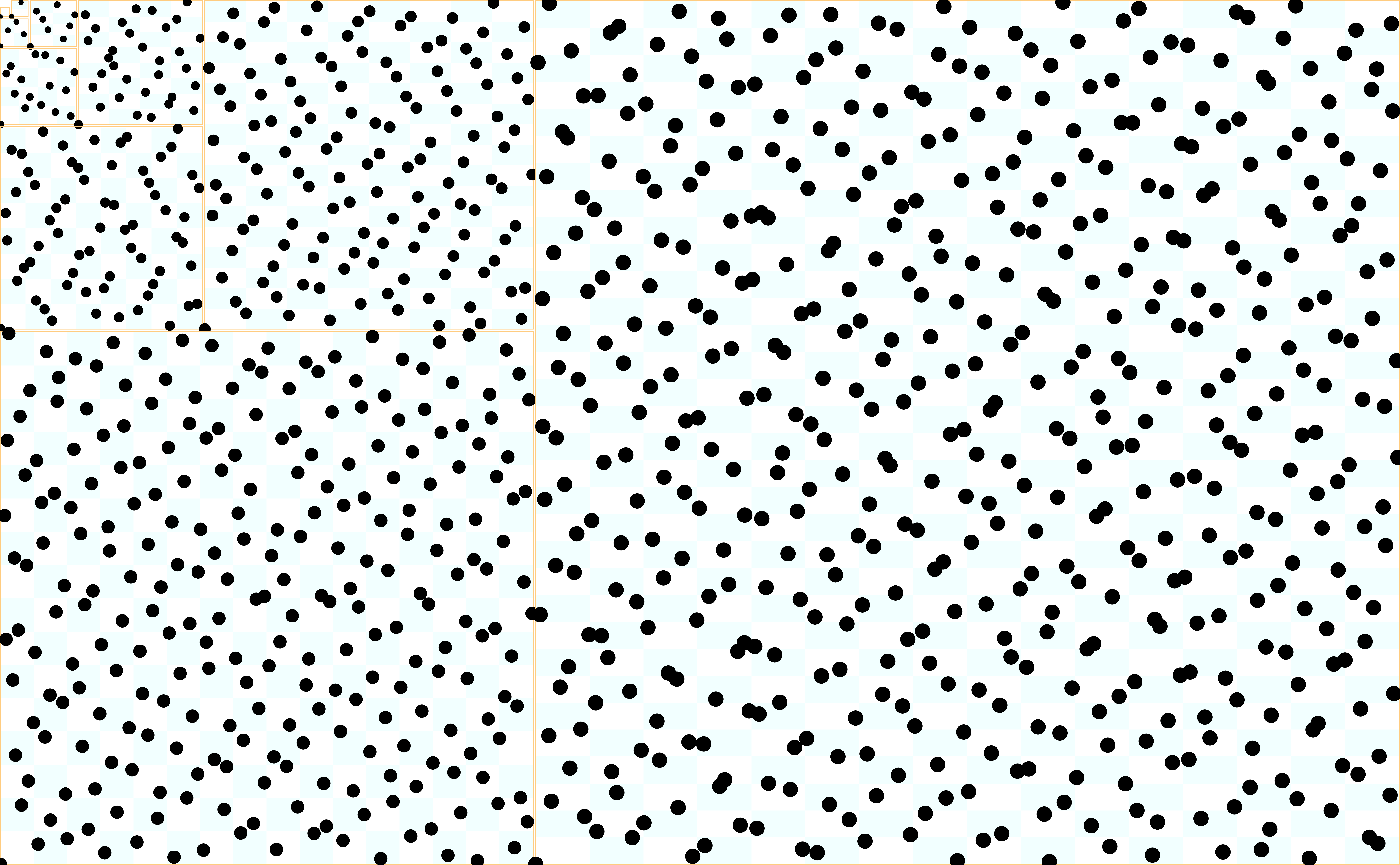}\\
        \includegraphics[width=1\unit]{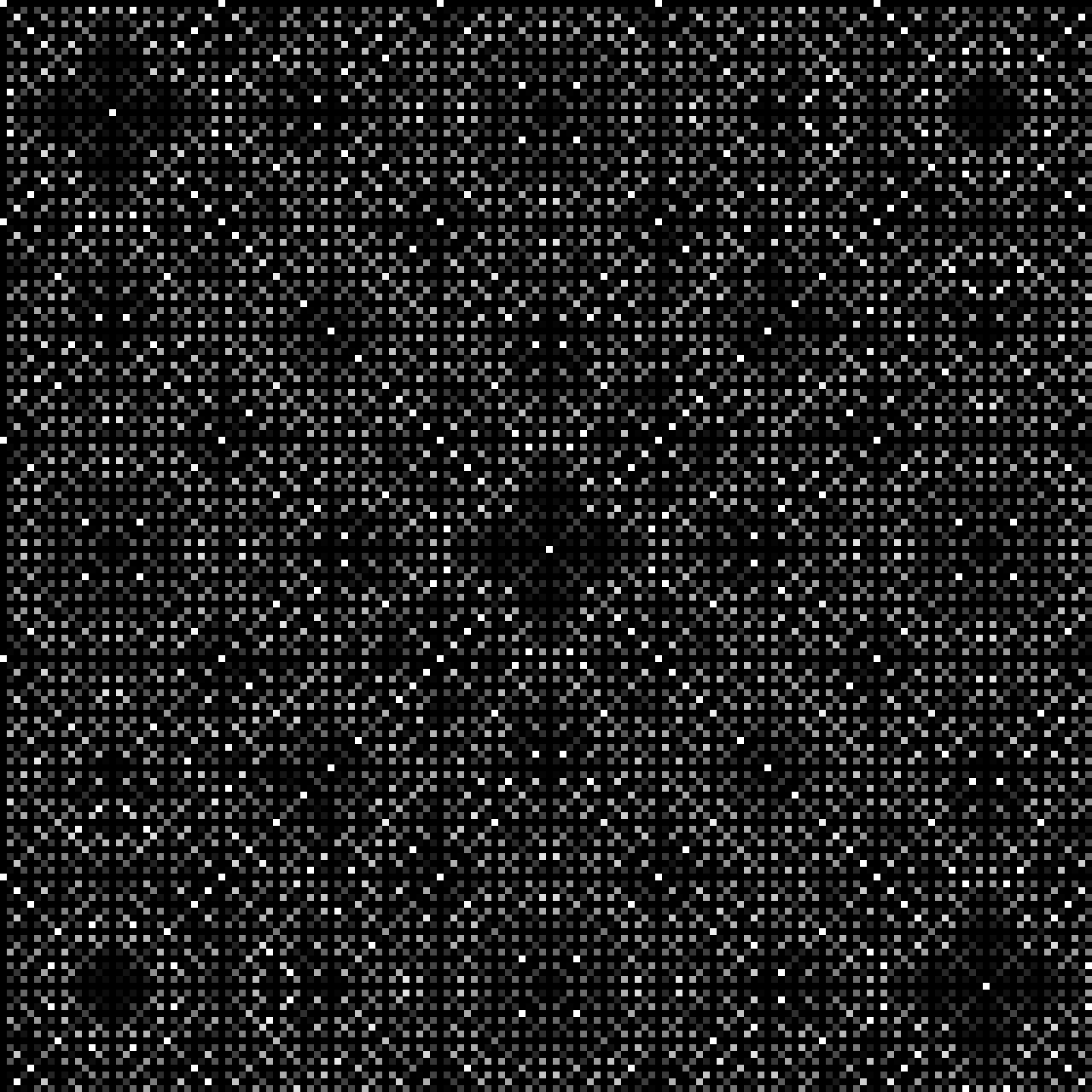}&%
        \includegraphics[width=1\unit]{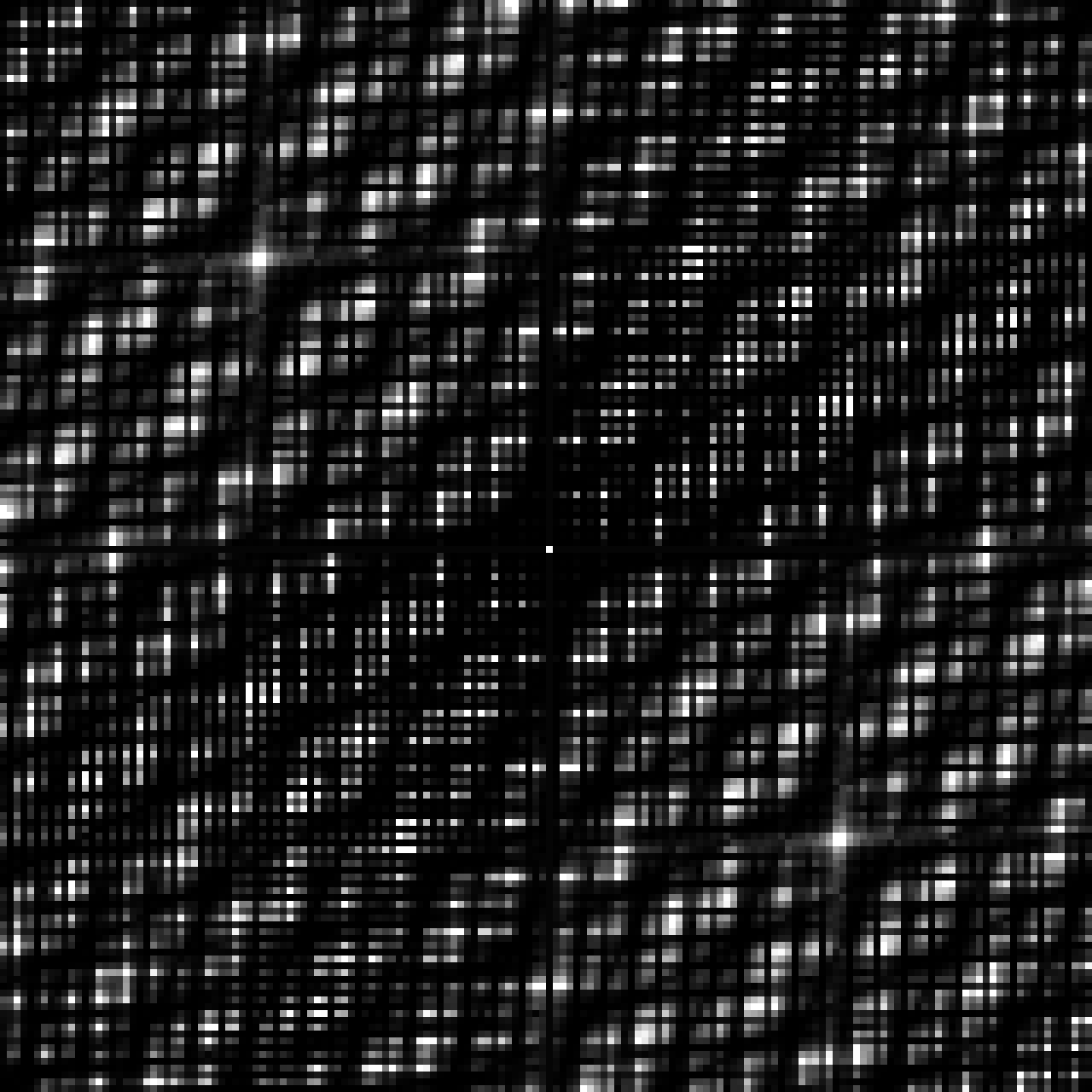}&%
        \includegraphics[width=1\unit]{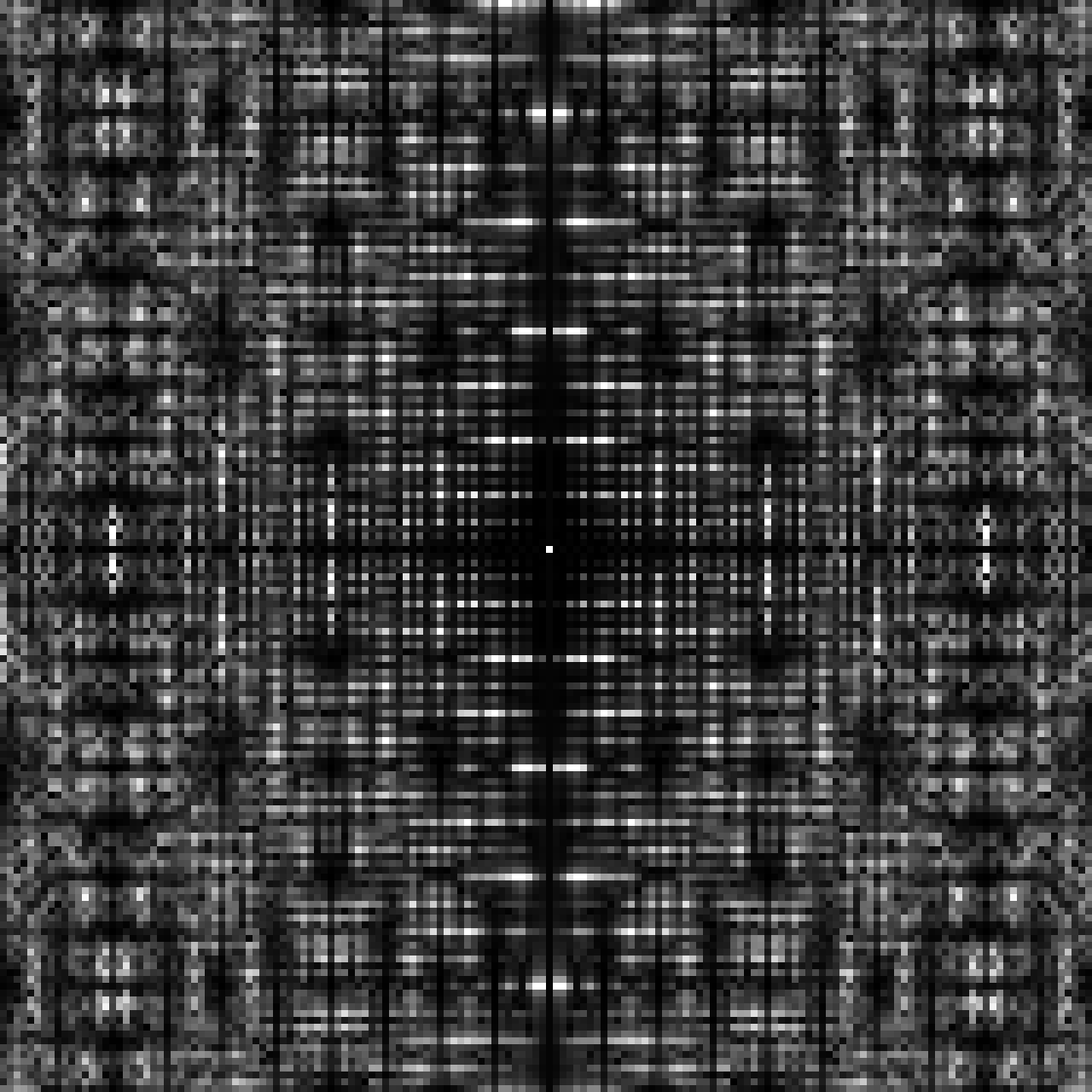}&%
        \includegraphics[width=1\unit]{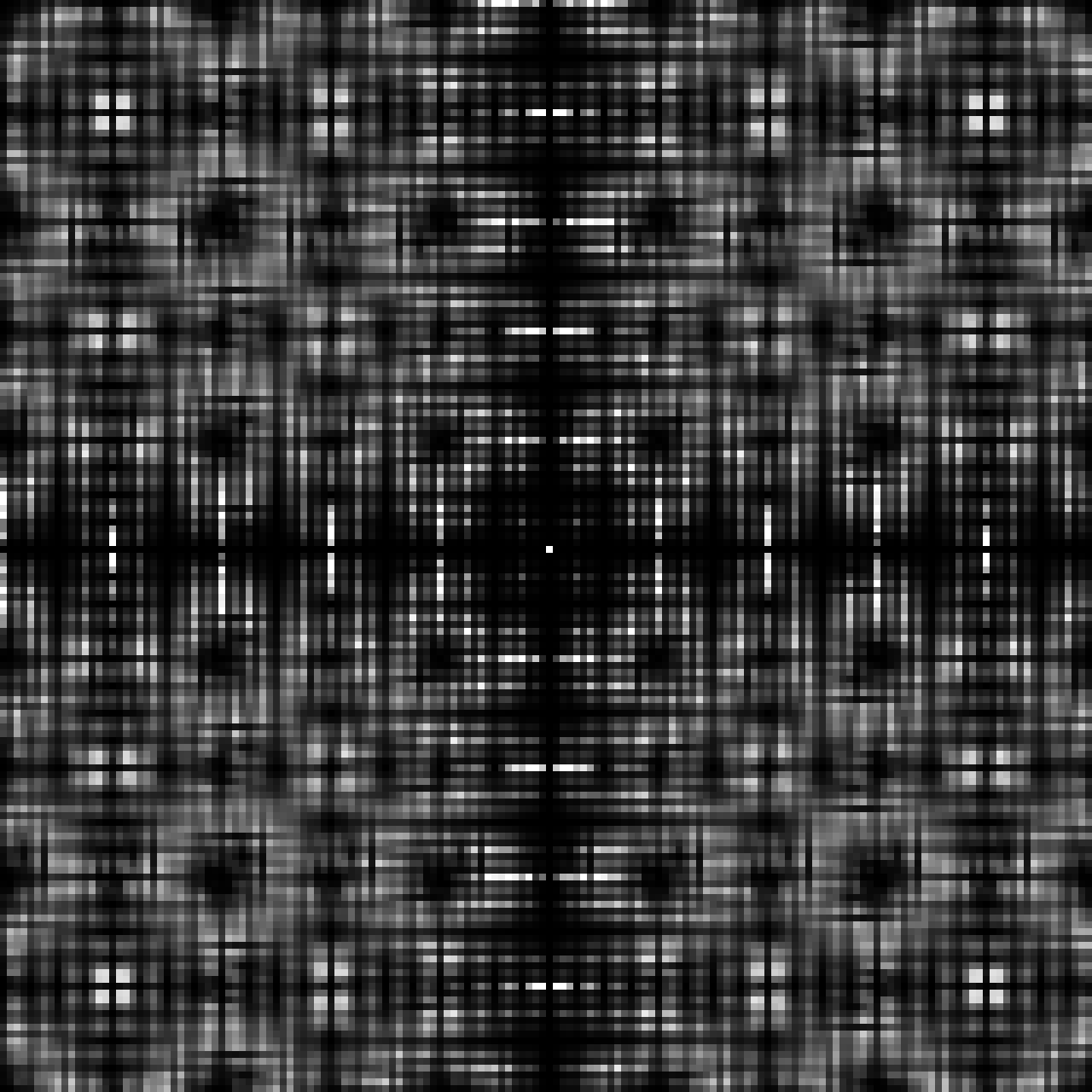}&%
        \includegraphics[width=1\unit]{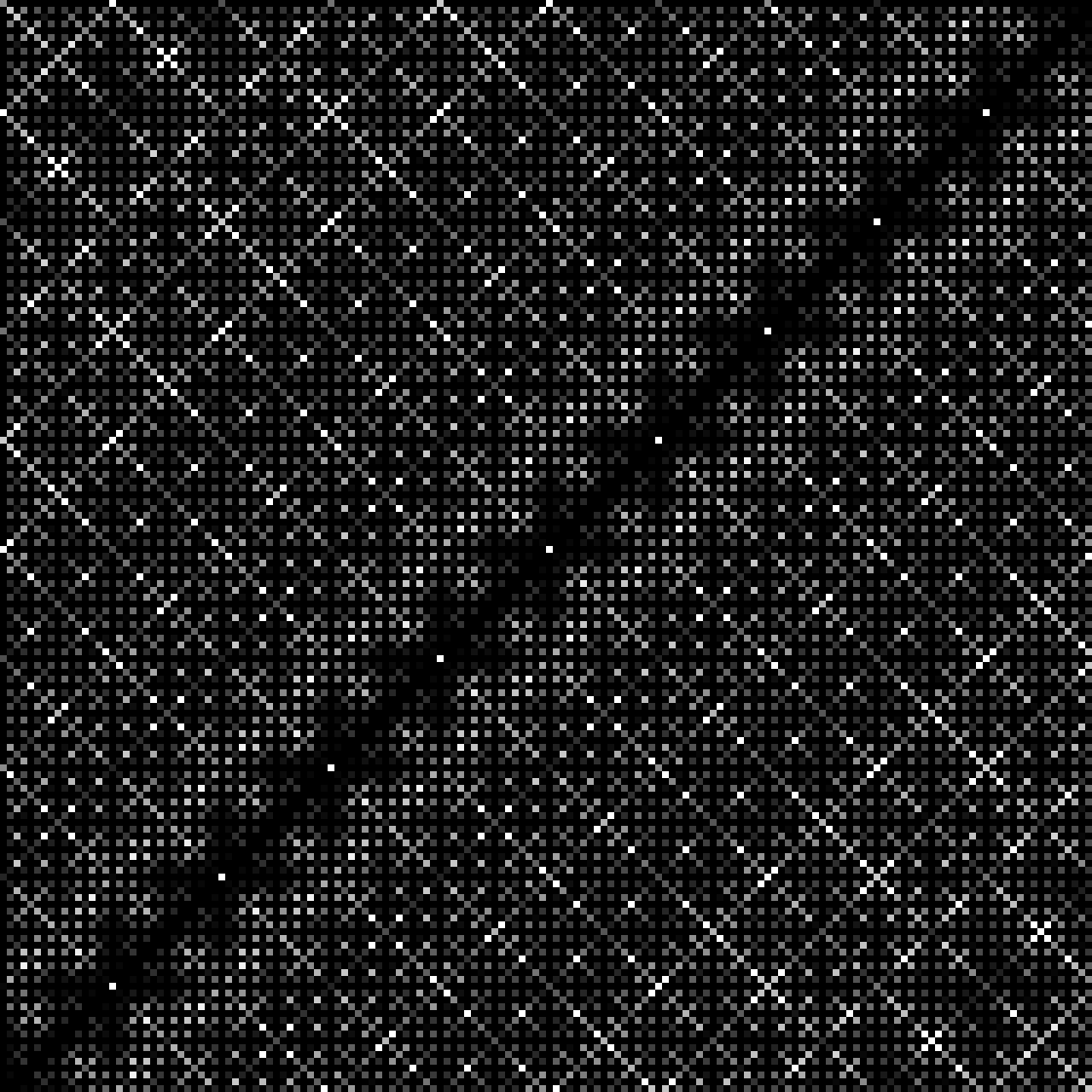}&%
        \includegraphics[width=1\unit]{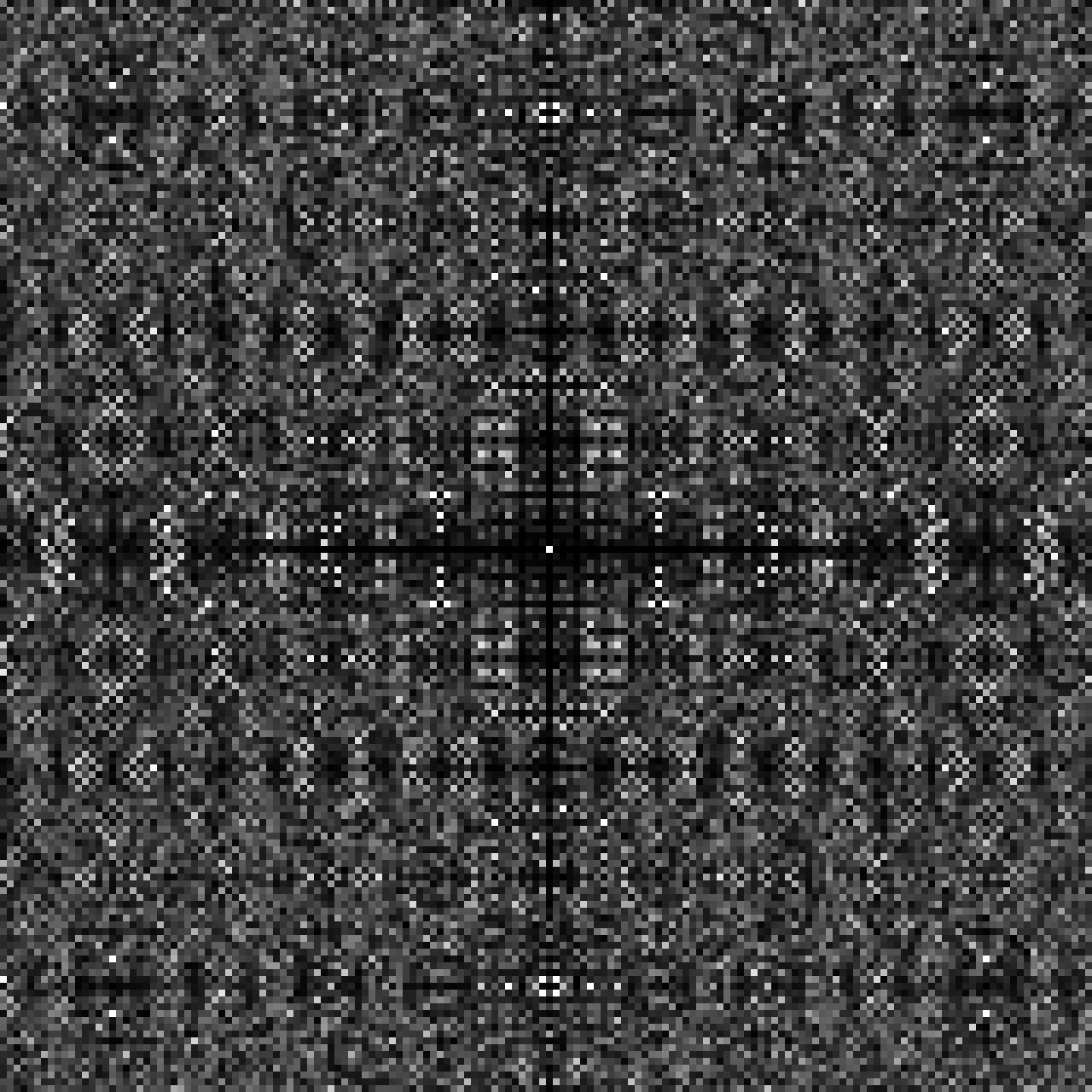}\\
        (a) Sobol & (b) Hammersley & (c) LP & (d) Extended Gray & (e) Self-similar & (f) Random
    \end{tabular*}}
    \caption{\label{fig:teaser}%
      Generating matrices, points, and periodograms of various digital dyadic sequences. The first $2^m$ points of the sequences are shown for $m=0,\dots,9$ in separate squares of increasing sizes. The digital dyadic sequences in (b) and (c) are obtained by reordering the known digital dyadic nets with the same names; such reordering is one of the main discoveries of the paper. The sequence in (d) is obtained from a $256$-point Gray net by first reordering it into a sequence and then extending the sequence in one of the possible ways.  
      The sequence in (e) represents a new class of self-similar sequences introduced in the paper. 
    }
\end{teaserfigure}

\begin{document}

\maketitle


\section{Introduction\label{sec:introduction}}

Sampling is a fundamental process in computer graphics, underlying Monte Carlo integration in rendering, halftoning, stippling, generative modeling, and object distributions.
A wide range of sampling strategies have been developed, and none is conclusively considered the best so far. Low-discrepancy patterns, however, introduced by Shirley \shortcite{Shirley91Discrepancy} to graphics and popularized mostly by Keller and collaborators, are arguably leading the competition for use in rendering.
Of special interest are so-called dyadic nets and sequences, mainly the Hammersley net and the Sobol sequence. These combine exceptionally high production rates with excellent convergence behavior when used in Monte Carlo integration, for which they are specifically designed. They are very versatile, as described in the discussion by Keller~\shortcite{Keller13Nutshell}. Furthermore, they are easy to understand thanks to their modular geometric ``multi-stratified'' structure \cite{Pharr16PBRT}.
In a nutshell, a dyadic net is concerned with a ``fair'' 2D distribution of \mbox{$N=2^m$} points,
so that specifically constructed rectangles of equal area 
contain exactly one point. These rectangles are obtained by partitioning the unit square $[0,1)^2$ 
into rectangular cells of sizes $\frac{1}{2^m}\times 1$, $\frac{1}{2^{m-1}}\times \frac{1}{2}$, $\ldots$, $1\times \frac{1}{2^m}$
(see Fig.~\ref{fig:example 4-net}).
A dyadic sequence is a sequence of sample points such that all leading sub-sequences of points of size $2^m$, for all~$m$, are dyadic nets as well as all subsequent sub-sequences of size $2^m$ (see Fig.~\ref{fig:sequence tree}). We would like to remark that dyadic nets and sequences fulfill the weaker $n$-rooks and Latin hypercube conditions so that they are proper subsets.
In computer graphics, the research effort in the area of dyadic nets and sequences was spearheaded by Keller and collaborators, who early on used and analyzed dyadic distributions in rendering \cite{Kollig02Efficient,Keller06Myths,Gruenschlos08Nets,Gruenschlos09Nets}.

In this paper, we set out to analyze the design space of matrix-generated, or \emph{digital}, dyadic nets and sequences. One possible reason why this design space is still not fully explored is that the focus of the Monte Carlo community might have been on proving discrepancy bounds and not the actual construction of dyadic sequences and dyadic nets.
Therefore, important contributions in the area of sample construction were made even recently. Examples include a compact parametrization of the whole dyadic space by Ahmed and Wonka \shortcite{Ahmed2021Optimizing}, a compact parametrization of a sub-space comprising dyadic sequences, along with a very efficient generation algorithm, by Helmer et al. \shortcite{Helmer2021Stochastice}, and an all-new construction of a high-dimensional distribution that has pair-wise two-dimensional projections by Paulin et al. \shortcite{Paulin2021Cascaded}.

Our paper makes the following contributions:
\begin{itemize}
\item a comprehensive analysis of 
\blueit{digital} dyadic nets and sequences. We introduce the missing ingredient (Thm.~\ref{th-digital dyadic sequence}) that enables us to map the analysis of matrices generating \blueit{digital} dyadic sequences to the language of linear algebra.

\item 
\redit{the computation of the exact dimension} of the design space for both digital dyadic nets and sequences. 

\item an algorithm that reorders any digital dyadic net into a digital dyadic sequence. For example, we can reorder the Hammersley net or the Larcher-Pillichshammer net into digital dyadic sequences.

\item the discovery of self-similar dyadic sequences, also called $\xi$-sequences. They are more efficient to construct than the Sobol sequence, and they are easily invertible. The subspace of $\xi$-sequences is large enough to contain examples with low discrepancy and a large minimal distance between two sample points.

\end{itemize}

\begin{figure}[tb]
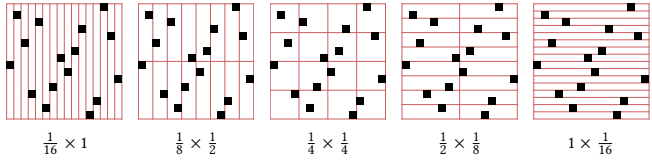

\setlength{\unit}{0.18\columnwidth}
{\centering\scriptsize
  \begin{tabular*}{1\columnwidth}{@{}c@{\extracolsep{\fill}}c@{\extracolsep{\fill}}c@{\extracolsep{\fill}}c@{\extracolsep{\fill}}c@{}}
    \includegraphics[width=1\unit]{images/4-net/0.pdf}&%
    \includegraphics[width=1\unit]{images/4-net/1.pdf}&%
    \includegraphics[width=1\unit]{images/4-net/2.pdf}&%
    \includegraphics[width=1\unit]{images/4-net/3.pdf}&%
    \includegraphics[width=1\unit]{images/4-net/4.pdf}\\[1mm]
    \bluevar{$\frac{1}{16}\times 1$} & 
    \bluevar{$\frac{1}{8}\times \frac{1}{2}$} &
    \bluevar{$\frac{1}{4}\times \frac{1}{4}$} &
    \bluevar{$\frac{1}{2}\times \frac{1}{8}$} &
    \bluevar{$1\times \frac{1}{16}$}
  \end{tabular*}}
  \vspace{-3mm}
    \caption{\label{fig:example 4-net}
        An example 4-bit 16-point dyadic net (i.e., $m=4$, $N=16$) showing all $m + 1 = 5$ possible stratifications: 
        \bluevar{$\frac{1}{2^4}\times 1$, $\frac{1}{2^3}\times \frac{1}{2}$, $\frac{1}{2^2}\times \frac{1}{2^2}$, $\frac{1}{2}\times \frac{1}{2^3}$, $1\times \frac{1}{2^4}$.}
    }
  \vspace{-4mm}
\end{figure}

\begin{figure}[t]
\setlength{\unit}{0.18\columnwidth}
{\centering\scriptsize
    \includegraphics[width=0.85\columnwidth]{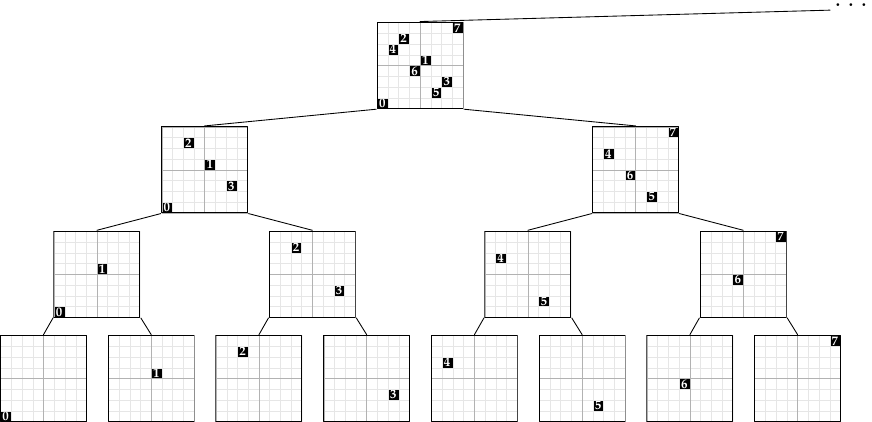}}
    \vspace{-2.5mm}
    \caption{\label{fig:sequence tree}
        The hierarchical structure of a dyadic sequence is visualized as a tree. 
        Note that the individual points also represent (0, 0, 2)-nets (i.e., 1-point dyadic nets). All sub-sequences corresponding to an internal (or leaf) node in the tree are required to be dyadic nets. \bluevar{Here the gray midlines are shown for better appearance.}
    }
    \vspace{-3.5mm}
\end{figure}

\blueit{The rest of the paper is organized as follows.
In Section~\ref{sec-related-work}, we highlight the most relevant literature.
In Section~\ref{sec:digital sequences}, we define and analyze digital dyadic nets and sequences, \redit{recall some known} and state our new theoretical results: 
compute the exact dimension of \redit{the design space for both} 
and give an algorithm that reorders any digital dyadic net into a digital dyadic sequence. In Section~\ref{sec:synthetic sequences}, we apply our algorithm for several examples: we reorder the Hammersley, the Larcher--Pillichshammer, and the Gray nets into digital dyadic sequences. In Section~\ref{sec:xi-sequences}, we introduce and study new self-similar dyadic sequences and give an explicit parametrization of their design space.}


\vspace{-2mm}
\section{Related Work}
\label{sec-related-work}

Our work has interesting connections to different fields of study in computer graphics and mathematics.
In this section, we try to highlight the most relevant literature, organized \bluevar{into} three categories.


\subsection{Low-Discrepancy Sampling\label{sec:LD sampling}}

The study of uniform sample distributions is much older than computer graphics.
An early 1D sequence was presented by van der Corput \shortcite{vanderCorput1935}, who established the digit-reversal principle that underlies almost all Low-Discrepancy (LD) constructions.
Roth~\shortcite{Roth1954OnIrregularities} and Hammersley~\shortcite{Hammersley1960MonteCarlo}, respectively, presented techniques for extending the van der Corput principle to two and higher dimensions.

In a parallel line of study, Sobol~\shortcite{sobol1967} pioneered the \emph{digital} construction of LD sets and sequences, generalized by Niederreiter \shortcite{Niederreiter87Point,Niederreiter92Random}, who coined the terms $(t, m, s)$-nets and $(t, s)$-sequences, and initiated a new field of study of these constructions.
We defer definitions to Section~\ref{sec:digital sequences}, but point out that nets are finite LD sets, while sequences are extensible sets that admit more points while maintaining a low discrepancy.
Most of these constructions are \emph{digital}, i.e. sample coordinates are derived from the sequence number via binary matrix vector multiplication and vector addition modulo~$2$. 
A thorough study of common construction methods is given by Dick and Pillichshammer~\shortcite{Dick10Digital}.
Not all nets and sequences, however, are necessarily digital, and there are known algorithms for constructing non-digital nets \cite{Gruenschlos09Nets,Ahmed2021Optimizing}.

Aside from direct construction techniques of nets and sequences, Owen~\shortcite{Owen95Randomly,Owen98Scrambling} presented a powerful technique to randomize $(t, m, s)$-nets and $(t, s)$-sequences while preserving their favorable properties; in fact improving their quality.
The original concept by Owen, while quite simple, is very memory intensive, consuming as many random bits---per dimension---as the number of samples.
Much work followed Owen's, mainly aimed at simplifications for more efficient implementation \cite{Friedel02Fast,Kollig02Efficient,Burley2020Scrambling}, at the cost of different compromises.
However, there are also variants that are conceptually different; most notably the techniques of Faure and Tezuka \shortcite{Faure02Another} aimed at shuffling the order of the samples rather than scrambling their locations.

Owen's scrambling works only on existing nets and sequences, so a method for generating the underlying sets is still required.
Ahmed and Wonka~\shortcite{Ahmed2021Optimizing} presented a technique for the direct construction of arbitrary 2D nets, but their method can only produce nets and consumes even more random bits than Owen's scrambling.


\subsection{Recursive Tiling Techniques}

Ostromoukhov \shortcite{Ostromoukhov07Polyominoes} is credited for importing the digit-reversal principle of LD constructions (Section~\ref{sec:LD sampling}) and using it to distribute blue-noise samples over a self-similar set of tiles.
The idea was followed up by Wachtel et al. \shortcite{Wachtel14Fast}, and matured in Ahmed et al. \shortcite{Ahmed17ART}, whose tile set, ART, has a low granularity and tuneable properties, presenting a comparable blue-noise competitor to LD sequences \cite{Christensen18Progressive,Ahmed2021Optimizing,Kopfetal2006}.


\subsection{Low-Discrepancy Blue Noise}

A recent trend of research tries to inject blue noise properties into LD sets and/or sequences. Early ideas along this line of research may be traced back to Keller~\shortcite{Keller06Myths}, who speculated on the superiority of LD sets that bear a large minimum distance between points, showcasing the very-low-discrepancy family of nets by Larcher and Pillichshammer~\shortcite{Larcher03Sums} that also attain a large minimum distance between points. Gr{\"u}nschlo{\ss} and colleagues subsequently presented examples of small-sized digital nets with maximized minimum distance between points obtained via a search algorithm \shortcite{Gruenschlos08Nets}, as well as non-digital constructions developed heuristically \shortcite{Gruenschlos09Nets}. The recent work of Christensen et al. \shortcite{Christensen18Progressive}, along with \cite{Pharr16PBRT,Helmer2021Stochastice}, may be seen as an extension of this search approach that targets sequences instead of nets.

A different line of research in this category was that of Ahmed et al. \shortcite{Ahmed16LDBN}, with an algorithm for proactively imposing spectral control over LD nets, at the cost of compromising the quality of the net. Perrier et al.~\shortcite{Perrier18Sequences} extended this idea to sequences. See a survey by Singh et al.~\shortcite{Singh2019} for details. Ahmed and Wonka~\shortcite{Ahmed2021Optimizing} demonstrated the possibility of incorporating spectral control into nets without compromising their quality but highlighted the difficulty of extending the idea to sequences.

\section{\bluevar{Analysis of Digital Dyadic Nets and Sequences}}
\label{sec:digital sequences}
\label{sec:background}

In this section, we define and analyze digital dyadic nets and sequences, and state our new theoretical results: 
computation of the exact dimension \redit{of the design space for both} 
and an algorithm that reorders any digital dyadic net into a digital dyadic sequence.

\subsection{Definitions}
\label{ssec:definitions}
There is a general notion of $(t, m, s)$-nets in base $b$. In this paper, we are focused on the particular case of 2D dyadic nets. They are known as $(0, m, 2)$-nets in base 2 (in short, $m$-nets, or $2^m$-point nets). Further, we specifically study dyadic sequences, also known as $(0, 2)$-sequences in base $2$.

\begin{definition}
\blueit{By a \emph{rectangle} $a\times b$ we mean the set of points $$\{\,(x,y)\in\mathbb{R}^2:x_0\le x<x_0+a,\, y_0\le y<y_0+b\,\} \quad\text{for some $x_0,y_0$.}$$}
Consider $m+1$ possible stratifications of \blueit{the unit square} $[0,1)^2$ into equal rectangles \blueit{$\frac{1}{2^m}\times 1$, $\frac{1}{2^{m-1}}\times \frac{1}{2}$, $\ldots$, $1\times \frac{1}{2^m}$
(\emph{strata}).}
See Fig.~\ref{fig:example 4-net}.
A $2^m$-point set in $[0,1)^2$ is a \emph{dyadic net} if it contains a single point per stratum in each of the $m+1$ stratifications.

A \emph{dyadic sequence} is a finite or infinite sequence of points in $[0,1)^2$ such that the first $2^m$ points, the next $2^m$ points, and all subsequent blocks of
$2^m$ points are dyadic nets for all $m = 0, 1, \ldots$.
\end{definition}

Thus, infinite dyadic sequences comprise sequences of dyadic nets of all sizes. Any dyadic sequence can be visualized as a binary tree, where each node corresponds to a dyadic net. 
See Fig.~\ref{fig:sequence tree}.




Now we turn to \emph{digital dyadic nets} \cite{Niederreiter87Point}. \move{} Hereafter all the computations use linear algebra over the Galois field $GF(2)$, also denoted by $\mathbb{Z}/2\mathbb{Z}$ or $\mathbb{Z}_2$: all vector and matrix entries $\in \{0,1\}$ and $1+1=0$.
Computation in $GF(2)$ is different from computing with binary matrices over the real numbers ($1+1 = 2$), or Boolean algebra ($1+1 =1$).
In what follows, a \emph{matrix} refers to an $m\times m$ matrix with the entries in $GF(2)$ unless otherwise indicated.

\begin{definition}
    \label{def-digital-dyadic-net}
    For two bit vectors $X=(x_1,\dots,x_m)$ and $Y=(y_1,\dots,y_m)$, denote $\rho(X,Y):=(0.x_1\dots x_m,0.y_1\dots y_m)\in [0,1)^2$.
    An ordered collection of $2^m$ distinct points $(X_1,Y_1)$, $(X_2,Y_2)$, \dots is called \emph{digital} if 
    \begin{equation}
    (X_i, Y_i) = \rho(C_xS_i, C_yS_i)\,,     \label{eq:XY S}
    \end{equation}
    where $(C_x, C_y)$ is a fixed pair of binary 
    $m\times m$ matrices, $S_1$, $S_2$, \dots enumerate all the binary $m$-bit column vectors in the order $(0,\dots,0)$, 
    $(1,0,\dots,0)$, $(0,1,\dots,0)$, \dots, $(1,\dots,1)$,
    and the products are modulo~$2$.
    A digital collection of points that is also a dyadic net (sequence) is called a \emph{digital dyadic net (sequence)}.
\end{definition}

Equivalently, based on the analysis from~\cite{Ahmed2021Optimizing}, a digital collection of points is a \emph{digital dyadic net}, if it contains a unique point $\rho(X,Y)$ with any given last $r$ bits of $X$ and last $m-r$ bits of $Y$, for each $r=0,1,\dots,m$. 
Likewise, a digital collection of points is a \emph{digital dyadic sequence}, if the first and all the subsequent blocks of $2^k$ points form a digital dyadic net after removing the last $m-k$ bits from both $X$ and $Y$, for each $k<m$.

\begin{table*}[t]
\footnotesize
\centering
\setlength{\aboverulesep}{1pt}
\setlength{\belowrulesep}{1pt}
\caption{\textbf{Overview of constructions} of different digital dyadic nets and sequences. $L_x, L_y$ are lower unitriangular matrices, $U_x, U_y$ are upper unitriangular matrices, $J$ is the anti-diagonal matrix \bluevar{with all ones on the anti-diagonal}, $P$ is the binary Pascal matrix, $M$ is any invertible matrix, $X_0, Y_0$ are binary bit vectors, and $S$ is the index of the generated point $(X,Y)$ in binary bit vector form. \bluevar{See Section~\ref{sec:synthetic sequences} for details.}}
\vspace{-2mm}
\begin{tabular}{@{}l|l|c|c}
    \toprule
     \textbf{Name of construction}&\textbf{Explicit construction}& \textbf{\# constructions (counting permutations)}   & \textbf{\# constructions (not counting permutations)}    \\
    \midrule
     \cc Digital dyadic nets       & \cc $(X,Y) = (M, L_y U_y J M) S$ & \cc $2^{ 3m (m - 1) / 2 } (2^1-1) (2^2-1) \cdots (2^m-1)$ & \cc $2^{ m (m - 1) }$ \\
     \cc GS-nets       & \cc $(X,Y) = (L_x M, L_y P M) S$ & \cc $2^{ 3m (m - 1) / 2 } (2^1-1) (2^2-1) \cdots (2^m-1)$ & \cc $2^{ m (m - 1) }$ \\
     \gr Digital dyadic sequences  & \gr $(X,Y) = (L_x U_x, L_y P U_x) S$ & \gr $2^{ 3m (m - 1) / 2 }$ & \gr $2^{ m (m - 1) }$ \\
    \cmidrule{1-4}
     \cc GFaure sequences          & \cc $(X,Y) = (L_x, L_y P) S$ & \cc $2^{ m (m - 1) }$   & \cc $2^{ m (m - 1) }$  \\
     \gr Hammersley-like nets       & \gr $(X,Y) = (J, L_yU_y) S$            & \gr $2^{ m (m - 1) }$   & \gr $2^{ m (m - 1) }$  \\
     \cmidrule{1-4}
     \cc Affine digital dyadic nets       & \cc $(X,Y) = (M, L_y U_y J M) S + ( X_0, Y_0 )$ 
     & \cc $2^{ m (3m + 1) / 2 } (2^1-1) (2^2-1) \cdots (2^m-1)$ & \cc $2^{ m^2 }$ \\
     \gr Affine digital dyadic sequences  & \gr $(X,Y) = (L_x U_x, L_y  {P} U_x) S + ( X_0, Y_0 )$ 
     & \gr $2^{ m (3m + 1) / 2 }$ & \gr $2^{ m^2 }$ \\
    \bottomrule
\end{tabular}
\vspace{-3mm}
\label{tab:methods_overview}
\end{table*}

%
%
\begin{table}[t]
\footnotesize
\centering
\setlength{\aboverulesep}{1pt}
\setlength{\belowrulesep}{1pt}
\caption{\textbf{Overview of constructions} of particular digital dyadic nets and sequences. $I$ and $J$ are the identity and anti-diagonal matrices, respectively. 
$ {P}$ is the binary Pascal matrix. \bluevar{$U_{\mathrm{LP}}$} and $L_{\bluevar{\mathrm{LP}}}$ are 
given by~\eqref{eq-C_LP} and~\eqref{eq-LP}. $S$ is the index of the generated point $(X,Y)$ in binary bit vector form. }
\vspace{-2mm}
\begin{tabular}{@{}l|l}
    \toprule
     \textbf{Name of construction}&\textbf{Explicit construction}  \\
    \midrule
     \cc Sobol (Faure) sequence            & \gr $(X,Y) = (I,  {P}) S$      \\
     \cc Hammersley net            & \cc $(X,Y) = (J, I) S$                \\
     \cc Hammersley sequence       & \cc $(X,Y) = (JPJ, PJ) S$                \\
     \cc LP net                    & \gr $(X,Y) = \bluevar{(J, U_{\bluevar{\mathrm{LP}}})} S$  \\
     \cc LP sequence                    & \gr $(X,Y) = (L_{\bluevar{\mathrm{LP}}}, PJ) S$  \\
    \bottomrule
\end{tabular}
\vspace{-3mm}
\label{tab:methods_overview_part2}
\end{table}

This notation enables us to express well-known sequences and nets in succinct form; the map ``$\rho$'' is usually omitted. Table~\ref{tab:methods_overview} summarizes explicit parametrizations of digital dyadic nets and sequences, and the analysis of how many different constructions are possible by each parametrization. The GFaure construction introduced by Tezuka \shortcite{Tezuka1994Generalization} describes a family of sequences obtained by setting $C_x = L_x$ and $C_y = L_y {P}$, with $L_x$ and $L_y$ arbitrary lower unitriangular matrices and ${P}$ the 
binary Pascal matrix. By a \emph{unitriangular} matrix we mean an upper or lower triangular matrix with all ones on the diagonal. The entries of the \emph{Pascal matrix} ${P}$ are $ {P}_{ij} = \binom{j-1}{i-1} \mod 2$, with $i$ row and $j$ column index, respectively:
\begin{equation}
     {P} =  \left(\begin{array}{ccccccccc}
         1      & 1      & 1      & 1      & 1      & 1      & 1      & 1      & \cdots \\
         \oh      & 1      & \oh      & 1      & \oh      & 1      & \oh      & 1      & \cdots \\
         \oh      & \oh      & 1      & 1      & \oh      & \oh      & 1      & 1      & \cdots \\
         \oh      & \oh      & \oh      & 1      & \oh      & \oh      & \oh      & 1      & \cdots \\
         \oh      & \oh      & \oh      & \oh      & 1      & 1      & 1      & 1      & \cdots \\
         \oh      & \oh      & \oh      & \oh      & \oh      & 1      & \oh      & 1      & \cdots \\
         \oh      & \oh      & \oh      & \oh      & \oh      & \oh      & 1      & 1      & \cdots \\
         \oh      & \oh      & \oh      & \oh      & \oh      & \oh      & \oh      & 1      & \cdots \\
         \vdots & \vdots & \vdots & \vdots & \vdots & \vdots & \vdots & \vdots & \ddots 
    \end{array}\right)\,.   \label{eq:pascal matrix}
\end{equation}
Hereafter a zero entry is \bluevar{in gray} for readability.

Table~\ref{tab:methods_overview_part2} summarizes our newly introduced \bluevar{(see Section~\ref{sec:synthetic sequences})} as well as known constructions for reference, e.g. the Hammersley net is obtained by $C_x =  {J}, C_y =  {I}$, with $ {J}$ \bluevar{the} anti-diagonal matrix \bluevar{with all ones on the anti-diagonal}, and the Larcher--Pillichshammer (LP) net uses 
\blueit{$C_x =  {J}, C_y =  U_{\bluevar{\mathrm{LP}}}$ with
\begin{equation}\label{eq-C_LP}
    U_{\bluevar{\mathrm{LP}}}=\left(\begin{array}{ccccc}
         1      & 1      & \cdots & 1      & 1     \\
         \oh      & 1      & \cdots & 1      & 1     \\
         \vdots & \vdots & \ddots & \vdots & \vdots\\
         \oh      & \oh      & \cdots & 1      & 1     \\
         \oh      & \oh      & \cdots & \oh      & 1      
    \end{array}\right)\,.
\end{equation}
Here $C_x =  {J}$ ensures that the points are ordered so that their $x$-coordinates are increasing.}
A clever construction of digital dyadic sequences has been suggested by \cite{sobol1967}; see a nice exposition by \cite{BratleyFox88}. The simplest and the most popular 2D Sobol sequence uses $C_x =  {I}, C_y =  {P}$. It was studied by \cite{Faure1982} and is often called simply the \emph{Sobol sequence}; see \cite{Pharr16PBRT}.

\remove{In the following, we first discuss digital dyadic nets and then digital dyadic sequences.} In order to determine if two pairs of matrices $(C_x, C_y)$ and $(C'_x, C'_y)$ generate the same set of points in a different order, we use the concept of the \emph{characteristic matrix}.
\begin{definition} For a pair of invertible matrices $(C_x, C_y)$,
the matrix  $ {C} = C_y C^{-1}_x$ is called the \emph{characteristic matrix} of the pair.
\end{definition}
If we want to check if two pairs of matrices generate the same set of points, we simply compare their characteristic matrices. The characteristic matrix directly associates the $x$-coordinate of a point with its $y$-coordinate, i.e $Y= {C}X$ when $X$ and $Y$ are written in appropriate binary form.
We also note that changing a pair of matrices $(C_x,C_y)$ to another pair $(C_xM,C_yM)$ with an invertible matrix $M$ preserves the characteristic matrix ${C}$ and enumerates all possible transformations that change the order of points but keep the final net the same. This transformation is equivalent to replacing $S$ by $MS$ in~\eqref{eq:XY S}, i.e., to a reordering of the bit vectors $S$.

In particular, if $M$ is upper unitriangular, then for each $k=1,\dots,m-1$ \bluevar{this} reordering preserves the decomposition of the sequence of $2^m$ possible bit vectors $S$ into blocks of $2^k$ consecutive vectors: only the vectors within each block and the whole blocks are reordered. Such reordering preserves the dyadic sequence property: a digital dyadic sequence is reordered to another digital dyadic sequence. Conversely, one can show that a matrix $M$ having this property must be upper unitriangular. 

We will also need the following definition for a type of matrix that is integral to understanding the digital construction.
\begin{definition}
    \label{def:progressive matrix}
    A \emph{progressive matrix} is a square matrix whose leading principal sub-matrices are invertible.
\end{definition}
In this context, we should also comment on an essential fact of linear algebra in $GF(2)$. A matrix is progressive if and only if it can be factored into the product of a lower unitriangular matrix~$L$ and an upper unitriangular matrix~$U$ \cite[Corollary 3.5.5]{Horn-Johnson-85}. The factorization is also unique.
So we only have $m(m-1)$ bits to determine a progressive matrix. This is $m$ bits less than required to determine an arbitrary $m\times m$ matrix. 
Finally, there are some interesting facts about the Pascal matrix that will be helpful: the Pascal matrix is upper unitriangular; $ {P}$ is its own inverse; mirroring both rows and columns by multiplying with $ {J}$ on each side gives $ {J} {P} {J}$, which is a lower unitriangular matrix; and $ {J} {P} {J}= {P} {J} {P}$ \redr{\cite[Lemma~2.5]{kajiura2018}}. See Appendix~\ref{app:proof of progressive pairs} for short proofs.

Also some comments on (numerical) linear algebra in $GF(2)$. Matrix inversion generally works by adapting the Gauss--Seidel algorithm. The LU-factorization algorithm can also be adapted. See Algorithm~\ref{alg:lu} in the appendix. The QR-factorization using Gram--Schmidt does not work, because vectors can be self-orthogonal. We make heavy use of determinant calculations, in particular, the Laplace expansion is employed in many proofs. As mathematics books treat linear algebra over general fields, we believe that a specific introduction to linear algebra over $GF(2)$ is more likely to be found in engineering textbooks, e.g.~\cite{Bard2009}. We also found packages for the numerical implementations in Python and C++, but ultimately reimplemented all algorithms from scratch.

\subsection{Digital Dyadic Nets}

We first would like to establish the conditions that need to be satisfied for a pair of matrices $(C_x, C_y)$ to form a digital dyadic net. We introduce the term \emph{dyadic pair of matrices} for the following discussion.
\begin{definition}
    A \emph{dyadic pair of matrices} $(C_x, C_y)$ is a pair of matrices such that the \emph{hybrid} matrices
\begin{equation}
    H_{r} = \left(\begin{array}{llll}
         a_{1, 1}   & a_{1, 2}   & \ldots & a_{1, m}   \\
         \vdots     & \vdots     & \ddots & \vdots        \\
         a_{m-r, 1} & a_{m-r, 2} & \ldots & a_{m-r, m} \\
         b_{1, 1}   & b_{1, 2}   & \ldots & b_{1, m}   \\
         \vdots     & \vdots     & \ddots & \vdots        \\
         b_{r, 1}   & b_{r, 2}   & \ldots & b_{r, m}   \\
    \end{array}\right)\,,   \label{eq:hybrid matrix}
\end{equation}
comprising the first $m-r$ rows from $C_x=(a_{ij})$ and the first $r$ rows from $C_y=(b_{ij})$, are invertible for all $r \le m$.
\end{definition}
It is well established \cite{Larcher01Walsh,Gruenschlos08Nets} that two binary matrices $(C_x, C_y)$ form a dyadic pair if and only if they produce a digital dyadic net with Eq.~\eqref{eq:XY S}. However, even though the conditions are known, they are not sufficient to derive an explicit construction or a complete description of the design space which we will \redit{now recall}.


\begin{theorem}\label{th-dyadic-pairs} \redvar{(See \cite[Lemma~2.6]{HoferSuzuki})}
A pair of matrices $(C_x,C_y)$ is dyadic if and only if $C_x$ is invertible, and $C_y=LU {J}C_x$ for some lower unitriangular matrix~$L$ and some upper unitriangular matrix~$U$. 
\end{theorem}
\remove{} 
This explicit parametrization enables the possibility to enumerate all possible constructions directly, and to sample from possible constructions without having to check that the condition for a dyadic pair is satisfied. While it is also possible to construct and sample from general invertible matrices $C_x$ over $GF(2)$, this introduces unwanted complexity that will be eliminated shortly. Luckily, for nets we do not need to distinguish between constructions that create the same points in a different order. Therefore, nets can be conveniently constructed in this canonical order, i.e. by ignoring $C_x$ or by using the parametrization of Hammersley-like nets.

We remark that the parametrization in Theorem~\ref{th-dyadic-pairs} is actually symmetric in $C_x$ and $C_y$: if $C_y=LUJC_x$, then $C_x=L'U'JC_y$ with $L'=JU^{-1}J$ and $U'=JL^{-1}J$. Here $L'$ and $U'$ are lower- and upper-unitriangular respectively because the left- and right-multiplication by $J$ reverses the order of rows and columns respectively.

Now we ask the question of \blueit{\emph{how many constructions are possible by the original construction}}. We can determine that there are
\begin{equation*}
2^{m(m-1)} \cdot 2^{m(m-1)/2} \cdot (2^1-1) \cdots (2^{m}-1)
\end{equation*}
possible constructions. The first factor of the product refers to the choices due to enumerating the matrices $L$ and $U$, and the rest to enumerating all invertible matrices $C_x$. 
\blueit{Such enumeration is given by the following known construction algorithm for invertible matrices over $GF(2)$: As the first column, pick any of the $2^{m}-1$ nonzero $m$-bit vectors. As the second column, pick any of the $2^{m}-2$ binary vectors not proportional to the first column. As the third column, pick any of the $2^{m}-2^2$ vectors not contained in the linear span of the first two columns, etc.} 

Next, we wish to answer the question of \bluevar{\emph{how many unique sets of points can be generated}}. This means we are interested in the number of constructions without considering (counting) permutations. 
We can determine that there are
\begin{equation}\label{eq-counting-characteristic-dyadic}
2^{m(m-1)}
\end{equation}
possible constructions by observing that the characteristic matrix has the form $ {C} = LU {J}$. The number of constructions directly follows from the number of bits we can choose for designing arbitrary unitriangular matrices $L$ and $U$. These results are included in Table~\ref{tab:methods_overview}. 
We would like to relegate the detailed proofs to \redr{Appendix~\ref{app:proof of dyadic pairs}}. While the proofs are not necessary to understand the main results in the paper, the ones for dyadic nets provide important insights for the case of dyadic sequences (the more difficult case and main focus of the paper).

\subsection{Digital Dyadic Sequences}
To discuss digital dyadic sequences we first need the following definition.

\begin{definition}
    \label{thm:progressive pair}
    A \emph{progressive pair} of matrices is a pair of matrices such that the $k\times k$ leading principal sub-matrices form dyadic pairs for all $k$. In other words, a pair $(C_x,C_y)=(a_{ij},b_{ij})$ is \emph{progressive}, if
    the \emph{hybrid} matrices
\begin{equation}
    H_{k,r} = \left(\begin{array}{llll}
         a_{1, 1}   & a_{1, 2}   & \ldots & a_{1, k}   \\
         \vdots     & \vdots     & \ddots & \vdots        \\
         a_{k-r, 1} & a_{k-r, 2} & \ldots & a_{k-r, k} \\
         b_{1, 1}   & b_{1, 2}   & \ldots & b_{1, k}   \\
         \vdots     & \vdots     & \ddots & \vdots        \\
         b_{r, 1}   & b_{r, 2}   & \ldots & b_{r, k}   \\
    \end{array}\right)\,,   \label{eq:hybrid matrix 2}
\end{equation}
are invertible for all $r\le k\le m$.
\end{definition}
Note that this definition implies that each of the two matrices is progressive. The concept of a progressive pair of matrices is the missing link that is required to fully map the analysis of digital dyadic sequences to the language of linear algebra. We establish the link as follows.

\begin{theorem}\label{th-digital dyadic sequence}
    A pair of matrices $(C_x, C_y)$ creates a digital dyadic sequence with~\eqref{eq:XY S} if and only if $(C_x, C_y)$ is a progressive pair. 
\end{theorem}

\begin{figure}
    \centering
    \includegraphics[width=0.9\columnwidth]{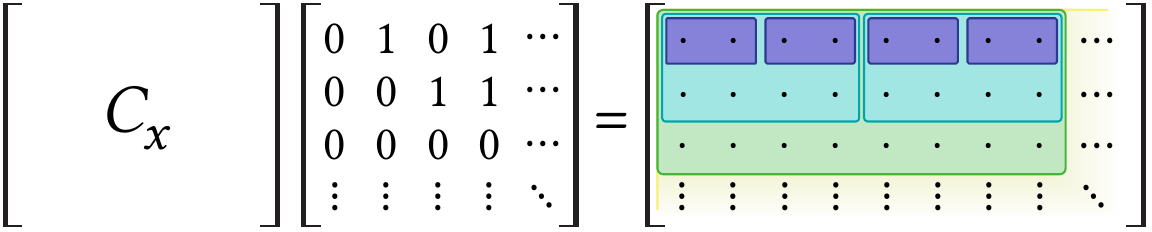}
    \vspace{-2mm}
    \caption{\label{fig:progressive theorem}%
        A visualization to support the proof of Theorem~\ref{th-digital dyadic sequence}. If an $m \times m$ generating matrix for a 1D digital dyadic sequence is multiplied with the matrix that encodes the sequence of integers $[0, 1, 2, 3, \ldots]$ in binary, the output matrix must ensure that specific blocks in the output form $2^1, 2^2, \ldots, 2^m$-point nets. However, to establish the net property for a $2^k$-point net, we only need to look at the first $k$ bits ($k$ rows) of the output. In the output we visualize the blocks of bits that need to be checked to verify the $2^1$-point ({\figColBoxA{blue color}}), $2^2$-point ({\figColBoxB{cyan color}}), and $2^3$-point ({\figColBoxC{green color}}) net property.
    }
    \vspace{-3mm}
\end{figure}

This \bluevar{new} theorem follows from our definitions: the last constraint in the paragraph after Definition~\ref{def-digital-dyadic-net} is a system of linear equations with the matrix being exactly hybrid matrix~\eqref{eq:hybrid matrix 2} from Definition~\ref{thm:progressive pair}. However, the resulting assertion is nontrivial; let us illustrate the point. 
We first start by analyzing the generation of a 1D digital dyadic sequence by a single matrix $C_x$. This matrix is then multiplied by a matrix $ {S}_{all}$ that consists of all possible inputs, i.e. all $m$ bit integers in sequence as binary numbers $[0,1,2,3,4,5,\ldots]$. Now we analyze the output matrix shown in Fig.~\ref{fig:progressive theorem}. In order for the matrix $C_x$ to generate a valid 1D dyadic sequence, the \figColBoxA{blue} boxes need to contain permutations of ${0,1}$, the \figColBoxB{cyan} permutations of ${00, 10, 01, 11}$, the \figColBoxC{green} boxes permutations of ${000, 100, 010, \ldots}$, and so on for all $k$ bit integers up to $m$. This analysis can be extended to a pair of generating matrices $(C_x, C_y)$ by observing that the first $2^k$ points need to form a $2^k$-point net. 
\blueit{All subsequent} blocks of $2^k$ points \blueit{are also dyadic nets; this is} a property of XOR-scrambling,
\blueit{which is} a bit harder to see. We can now analyze the design space by analyzing \blueit{the} properties of progressive pairs of matrices.

\redit{The known} explicit parametrization of the design space of digital dyadic sequences
\redvar{can be elegantly stated in terms of progressive pairs as follows.}
\begin{theorem}\label{th-progressive-pairs} \redvar{(See \cite[Theorem 1.2]{HoferSuzuki}.)}
A pair of matrices $(C_x,C_y)$ is progressive if and only if $C_x=L_xU$ 
and $C_y=L_y {P}U$ for some upper unitriangular matrix $U$ and some lower unitriangular matrices $L_x,L_y$.
\end{theorem}

Equivalently, $(C_x,C_y)=(L_xU_x,L_yU_y)$ is progressive, if and only if $U_yU_x^{-1}=P$, where we $LU$-decompose each matrix. The latter is equivalent to $U_xU_y^{-1}=P$ because $P=P^{-1}$, thus the condition is symmetric in $C_x$ and $C_y$.

\redit{A new short} proof \redvar{of the theorem} is given in Appendix~\ref{app:proof of progressive pairs}.
Just like in the case of dyadic nets, this explicit parametrization enables the possibility to enumerate all possible constructions directly and to sample from possible constructions. 
Again, we ask the question of \blueit{\emph{how many constructions are possible}}. We can determine that there are
\begin{equation*}
2^{3m(m-1)/2}
\end{equation*}
possible constructions. The number of constructions directly follows from the number of bits we can choose for designing arbitrary unitriangular matrices $L_x$, $L_y$, and $U$. Next, we wish to answer the question of how many unique sets of points can be generated. This means we are interested in \blueit{\emph{the number of constructions without considering permutations}}. 
We can determine that there are
\begin{equation}\label{eq-counting-characteristic-progressive}
2^{m(m-1)}
\end{equation}
possible constructions by observing that the characteristic matrix has the form $ {C} = L_y {P}L_x^{-1}$. The number of constructions follows from the number of bits we can choose for designing arbitrary unitriangular matrices $L_x$ and $L_y$. 
We therefore propose to use the GFaure construction to explore the complete design space of digital dyadic sequences in case the order in which the points are generated is not important (the GFaure construction ensures that the canonical order is a sequence order, though). While the construction was proposed before, it was not clear however that this particular construction was able to enumerate all possible sequences up to the order of the points.
In addition, for the reduced parameterization to be valid, we need to ensure that distinct choices of the pair $(L_x,L_y)$ lead to distinct characteristic matrices $ {C}=L_y {P}L_x^{-1}$. This is a somewhat nontrivial property proved in Appendix~\ref{app:proof of progressive pairs}, along with the other results of this subsection.

\subsection{Converting Digital Dyadic Nets to Digital Dyadic Sequences}
\label{subsec:net2sequence}
An important discovery of our work is that all digital nets can be reordered to become digital sequences.
We think this is quite striking, since digital nets like the Hammersley net have generally been considered to be non-extensible \cite{Keller13Nutshell,Pharr16PBRT}, whereas they actually are. Also, important digital nets like the LP-net can actually be reordered to become a sequence, a much more useful order for sampling applications.

When comparing the number of unique point sets that can be generated by the digital dyadic net construction~\eqref{eq-counting-characteristic-dyadic} and the digital dyadic sequence construction~\eqref{eq-counting-characteristic-progressive}, we notice that they can generate the same number of unique point sets. Therefore, we can conclude that every digital net can be reordered into a digital sequence. We would like to note that this conclusion may be a bit more complex than it seems. The conclusion is only possible because our enumeration is proven to be exhaustive and proven not to contain any duplicates. This justifies the proofs given in the appendix and also requires the joint treatment of digital nets and sequences, even though the original goal was just to understand sequences.

We first describe the process in the form of linear algebra, and later comment more explicitly on an algorithmic implementation.

\begin{theorem}\label{th-characteristic}
For each dyadic pair $(C_x,C_y)$ there is an invertible matrix $M$ such that $(C_xM,C_yM)$ is a progressive pair. 
\end{theorem}

In other words, the dyadic and progressive pairs have the same sets of characteristic matrices, and each digital dyadic net can be ordered to provide a digital dyadic sequence.

The matrix $M$ is explicitly constructed from the decomposition $C_y=LU {J}C_x$ provided by Theorem~\ref{th-dyadic-pairs}. A possible choice is
\begin{equation*}
    M=C_x^{-1} {J}U^{-1} {P} {J}.
\end{equation*}
The resulting pair has the form 
\begin{multline*}
(C_xM,C_yM)=( {J}U^{-1} {P} {J},L {P} {J})
=(
\underbrace{ {J}U^{-1} {P} {J}}_{L_x},
\underbrace{L {J} {P} {J}}_{L_y}\cdot {P})
\end{multline*}
by the identity $ {P} {J}= {J} {P} {J} {P}$.
Here $L_x= {J}U^{-1} {P} {J}$ and 
$L_y=L {J} {P} {J}$ are lower unitriangular matrices because for any upper unitriangular matrix $U'$ the matrix $ {J}U' {J}$ is lower unitriangular. So we arrive at the decomposition $(C_xM,C_yM)=(L_x,L_y {P})$ from Theorem~\ref{th-progressive-pairs} and therefore the pair is progressive. 

Based on this result, we propose the following construction for dyadic pairs, called \emph{GS-net}: 
\begin{equation}
(C_x,C_y)=(L_xM,L_y {P}M),
\end{equation}
where $L_x,L_y$ are arbitrary lower unitriangular matrices and $M$ is an arbitrary invertible matrix. This makes it clear that every digital dyadic net is simply a digital dyadic sequence reordered by an arbitrary invertible matrix.


We make the reordering explicit in Algorithm~\ref{alg:reordering}. 
It requires the computation of an inverse matrix and the computation of the LU-factorization in $GF(2)$. We implemented this ourselves, but libraries are available. Finally, we would like to recall that each dyadic sequence can be reordered by an upper unitriangular matrix $U$. Therefore, the reordering of a net into a sequence is not unique and the matrices returned by our proposed algorithm can still be right-multiplied by an arbitrary upper triangular matrix $U$.

\begin{algorithm} [tb]
    \caption{
        Converting a digital dyadic net into a digital dyadic sequence
    }
    \label{alg:reordering}
    \KwIn{
         a pair of matrices $(C_x, C_y)$ that are known to create a digital dyadic net with~\eqref{eq:XY S};\\
    }
    \KwOut{a pair of matrices $(C_{\textit{xnew}},C_{\textit{ynew}})$ that creates a digital dyadic sequence with the same points.
    }
    Compute $C_x^{-1}$ and the characteristic matrix $ {C} = C_yC_x^{-1}$\;
    Compute $ {C}' =  {C} {J}$, flipping the columns of $C$\; 
    Compute the LU-factorization of $ {C}'$ yielding $L$ and $U$ unitriangular matrices \bluevar{using Algorithm~\ref{alg:lu} in the appendix}\; 
    Return the two matrices $C_{\textit{xnew}} =  {J}U^{-1} {P} {J}$ and $C_{\textit{ynew}} = L {P} {J}$.
\end{algorithm}


\subsection{Affine Digital Dyadic Nets and Sequences}

Seeking to enlarge the space of digital dyadic nets and sequences, we apply XOR scrambling introduced by Kolleg and Keller \shortcite{Kollig02Efficient}. Namely, instead of~\eqref{eq:XY S}, we use equation
\begin{equation}
    (X, Y) = (C_x\bluevar{S}, C_y\bluevar{S}) + (X_0, Y_0)  \label{eq:xor scrambling}
\end{equation}
with some fixed bit vectors $X_0, Y_0$. This gives $2m$ additional degrees of freedom if we take the order of points into account.

Then we \blueit{rewrite \eqref{eq:xor scrambling} as}
$$
    Y
        = C_y C_x^{-1} (X + X_0) + Y_0
        = C_y C_x^{-1} X + Y_0^{\prime}
$$        
for some constant bit-vector $Y_0^{\prime}$.
This means that for any pair $(X_0, Y_0)$ of constant bit-vectors used to XOR-scramble a digital dyadic sequence, there exists a single bit vector $Y_0^{\prime}$ that would produce the same set of points by scrambling only along the $y$ axis.
This means that the actual dimension of xor-scrambling is only $m$.
Note that this $Y_0^{\prime}$ vector may not be absorbed within any matrix multiplication, so it represents a new degree of freedom.
Incidentally, it exactly compensates the $m$ diagonal bits we lost earlier.
This makes the dimension of affine digital dyadic sequences exactly $m^2$, allocated as $m(m-1)/2$ bits below the diagonal of each of the $C_x$ and $C_y$ matrices, and $m$ bits for XOR scrambling of $y$.



\begin{figure}
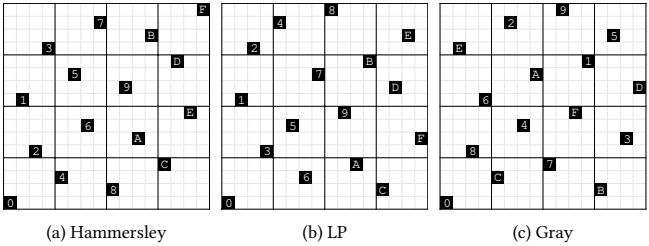

    \setlength{\unit}{(1\columnwidth - 4\gap)/3}
    \centering
\begin{tabular*}{1\columnwidth}{@{}c@{\extracolsep{\fill}}c@{\extracolsep{\fill}}c@{}} 
    \includegraphics[width=1\unit]{images/partitions/Hammersley.pdf}
&    \includegraphics[width=1\unit]{images/partitions/LP.pdf}
&    \includegraphics[width=1\unit]{images/partitions/gray4x4.pdf}
\\
{\footnotesize (a) Hammersley} &
{\footnotesize (b) LP} &
{\footnotesize (c) Gray}
\end{tabular*}
\vspace{-0.3cm}
    \caption{\label{fig:gray partition}%
        The first 16-point nets of 256-point (a) Hammersley (b) LP (c) Gray sequences, with possible indexings (in hexadecimal) of the points.
    }
    \vspace{-5mm}
\end{figure}

\section{Synthetic Sequences}
\label{sec:synthetic sequences}

In this section, we demonstrate examples of applying the knowledge gained through Section~\ref{sec:digital sequences} to synthesizing dyadic sequences with favorable properties, and in the following section we give a third demonstration of a whole class of sequences. 


\subsection{Designing Dyadic Nets and Sequences by Search and Optimization}

For applications of Monte Carlo integration, such as rendering, one often searches for digital dyadic nets and sequences that are optimal in a sense.
Common optimization targets include minimum distance between points and star discrepancy \cite{Keller06Myths}.
\emph{Star discrepancy} $D^*$ measures how the fraction $n(x,y)/2^m$ of points of an $m$-net in a rectangle $[0,x)\times[0,y)$ deviates from the area of the rectangle; formally, $D^*=\sup|xy-n(x,y)/2^m|$, where the supremum is over all $x,y$ in $[0,1]$.
There are multiple simple ways to search for dyadic nets and sequences with specific properties.
For example, for a small bit depth ($m=4, N=16$), it is feasible to perform an exhaustive search over all digital dyadic nets or sequences.
For larger bit depths, one can use heuristic search, e.g. hill climbing. At each iteration, all possible one bit mutations to an existing net or sequence can be evaluated and the optimum mutation can be chosen until no mutation brings an improvement in the objective.
One interesting possibility is a hierarchical search. Due to the incremental nature of nets and sequences we may optimize progressively, e.g. first searching for a $k$ bit solution and then optimizing for additional bits while retaining the leading samples (bits) in the net or sequence. 

The important aspect of our work is that any of these explorations can be done efficiently. The complete characterization of the design space for nets and sequences enables us to only visit candidates, e.g. generator matrices for a digital construction, that are viable. This improves over previous work that explored a much larger set of constructions and then filtered out invalid candidates. E.g. Gr{\"u}nschlo{\ss} et al. \shortcite{Gruenschlos08Nets} use heuristic search for alternative $C_x$ matrices to produce Hammersley-like nets. Since the percentage of viable candidates can be small, a search using our explicit construction can bring orders of magnitude speed-up.

Instead of searching for, e.g. new pairs $(C_x, C_y)$, different dyadic nets are typically derived from these, especially Sobol, by applying the dyadic-preserving Owen's scrambling \shortcite{Owen98Scrambling}. Faithful application of Owen's scrambling is computationally challenging, and common implementations offer different trade-offs between speed, quality, and flexibility \cite{Helmer2021Stochastice, Kollig02Efficient, Burley2020Scrambling}. Our explicit characterization and understanding of sequences provide multiple advantages. It is hard to predict the properties of an Owen scrambled sequence. The above-mentioned
hierarchical search would not be possible. Also, there are advantages when creating nets from sequences and distributing them over pixels, see e.g. PBRT~\cite{Pharr16PBRT} or Ahmed and Wonka~\shortcite{Ahmed20Screen}. 
In Fig.~\ref{fig:teaser}, we show example digital dyadic sequences.


\subsection{Hammersley Sequences}

A Hammersley net is generated by a pair of matrices $( {J},  {I})$. Executing Algorithm~\ref{alg:reordering} 
by manual computation, we arrive at the progressive pair of matrices $(  {J} {P} {J},  {P} {J})$. For example, the 256-point Hammersley net can be sequenced using the matrix pair
\begin{equation}
    \scalemath{0.8}{\left( \left(\begin{array}{ccccccccccc}
             1      & \oh      & \oh      & \oh      & \oh      & \oh      & \oh      & \oh     \\
             1      & 1      & \oh      & \oh      & \oh      & \oh      & \oh      & \oh     \\
             1      & \oh      & 1      & \oh      & \oh      & \oh      & \oh      & \oh     \\
             1      & 1      & 1      & 1      & \oh      & \oh      & \oh      & \oh     \\
             1      & \oh      & \oh      & \oh      & 1      & \oh      & \oh      & \oh     \\
             1      & 1      & \oh      & \oh      & 1      & 1      & \oh      & \oh     \\
             1      & \oh      & 1      & \oh      & 1      & \oh      & 1      & \oh     \\
             1      & 1      & 1      & 1      & 1      & 1      & 1      & 1     \\
        \end{array}\right),
        \left(\begin{array}{ccccccccccc}
             1      & 1      & 1      & 1      & 1      & 1      & 1      & 1     \\
             1      & \oh      & 1      & \oh      & 1      & \oh      & 1      & \oh     \\
             1      & 1      & \oh      & \oh      & 1      & 1      & \oh      & \oh     \\
             1      & \oh      & \oh      & \oh      & 1      & \oh      & \oh      & \oh     \\
             1      & 1      & 1      & 1      & \oh      & \oh      & \oh      & \oh     \\
             1      & \oh      & 1      & \oh      & \oh      & \oh      & \oh      & \oh     \\
             1      & 1      & \oh      & \oh      & \oh      & \oh      & \oh      & \oh     \\
             1      & \oh      & \oh      & \oh      & \oh      & \oh      & \oh      & \oh     \\
        \end{array}\right) \right)}\,.
\end{equation}
%
See Fig.~\ref{fig:gray partition}(a) for the layout of the first 16-point net in the 256-point sequence with possible indexing.

\subsection{LP Sequences} \label{ssec-LP}
While the Hammersley nets have the worst known star discrepancy among digital dyadic nets, the LP nets of Larcher and Pillichshammer \shortcite{Larcher03Sums} are believed to have the best one according to numerical experiments measuring the star discrepancy. 
We use $2^{2^k}$-point LP nets as a second example of converting a net into a sequence. 

An LP net is generated by \blueit{the} pair of matrices
\bluevar{$(J,U_{\bluevar{\mathrm{LP}}})$} 
given by~\eqref{eq-C_LP}. The corresponding sequence is generated by the pair $(L_{\bluevar{\mathrm{LP}}},PJ)$, where \bluevar{$L_{\bluevar{\mathrm{LP}}}=L$} is the lower-unitriangular matrix with the entries 
\begin{equation}
\label{eq-LP}
L_{11}=1,L_{i1}= L_{1j}=0, L_{ij}=\binom{i-2}{j-2}\mod 2 \quad\text{ for }2\le i,j\le m.
\end{equation}
Indeed, the pair $(L_{\bluevar{\mathrm{LP}}},PJ)$ is progressive by Theorem~\ref{th-progressive-pairs} because $PJ=JPJ\cdot P=:L_y\cdot P$.
The pair $(L_{\bluevar{\mathrm{LP}}},PJ)$ generates the same net because the characteristic matrices are the same: $PJL^{-1}_{\bluevar{\mathrm{LP}}}=\bluevar{U_{\bluevar{\mathrm{LP}}}J}$
if $m$ is a power of $2$. This identity is proved at the end of Appendix~\ref{app:proof of progressive pairs}.
For example, the following pair of matrices creates a 256-point LP net as a digital dyadic sequence:
\begin{equation}
    \scalemath{0.8}{\left( \left(\begin{array}{ccccccccccc}
             1      & \oh      & \oh      & \oh      & \oh      & \oh      & \oh      & \oh \\
             \oh      & 1      & \oh      & \oh      & \oh      & \oh      & \oh      & \oh \\
             \oh      & 1      & 1      & \oh      & \oh      & \oh      & \oh      & \oh \\
             \oh      & 1      & \oh      & 1      & \oh      & \oh      & \oh      & \oh \\
             \oh      & 1      & 1      & 1      & 1      & \oh      & \oh      & \oh \\
             \oh      & 1      & \oh      & \oh      & \oh      & 1      & \oh      & \oh \\
             \oh      & 1      & 1      & \oh      & \oh      & 1      & 1      & \oh \\
             \oh      & 1      & \oh      & 1      & \oh      & 1      & \oh      & 1 \\
        \end{array}\right),
        \left(\begin{array}{ccccccccccc}
             1      & 1      & 1      & 1      & 1      & 1      & 1      & 1 \\
             1      & \oh      & 1      & \oh      & 1      & \oh      & 1      & \oh \\
             1      & 1      & \oh      & \oh      & 1      & 1      & \oh      & \oh \\
             1      & \oh      & \oh      & \oh      & 1      & \oh      & \oh      & \oh \\
             1      & 1      & 1      & 1      & \oh      & \oh      & \oh      & \oh \\
             1      & \oh      & 1      & \oh      & \oh      & \oh      & \oh      & \oh \\
             1      & 1      & \oh      & \oh      & \oh      & \oh      & \oh      & \oh \\
             1      & \oh      & \oh      & \oh      & \oh      & \oh      & \oh      & \oh \\
        \end{array}\right) \right)}\,. \label{eq:LP sequence}
\end{equation}
%
See Fig.~\ref{fig:gray partition}(b). 
The analysis of the star discrepancy of all possible sub-sequences is an interesting open question.


\subsection{Converting Point Sets to Digital Dyadic Sequences}

Algorithm~\ref{alg:converting} can be used to test if an arbitrary set of points (including sets of points that are known dyadic nets) allows for a digital construction. If it is possible, the algorithm produces the pair of matrices that generate the given set of points as a digital dyadic sequence. 

\begin{algorithm} [tb]
    \caption{
        Converting a point set into a digital dyadic sequence
    }
    \label{alg:converting}
    \KwIn{
         an arbitrary $2^m$-point set in $[0,1)^2$;\\
    }
    \KwOut{a pair of matrices that creates a digital dyadic sequence with the same points using~\eqref{eq:XY S}, if it exists.
    }
    Order the points sequentially along the $x$-axis and
    set $C_x  = {J}$\; 
    Use the $y$-coordinates of the points at power-of-two indices as columns of the prospected $C_y$\; 
    Conduct one pass through the other points to validate if they coincide with the ones given by~\eqref{eq:XY S}, and check if $C_y$ is progressive\;
    If this step works, we have a digital dyadic net that can subsequently be converted to a sequence using Algorithm~\ref{alg:reordering}.
\end{algorithm}


\subsection{Gray Sequences}
We applied the conversion algorithm above to the Gray Code family of nets introduced by Ahmed and Wonka \shortcite{Ahmed2021Optimizing}, obtained by using a Gray Code-ordered van der Corput sequence to supplement the $x$-offsets of the points along the columns and the $y$-offsets along the rows. These nets turned out to be digital.

A 256-point Gray net (from now on, we will drop the word ``Code'') is generated by the pair of matrices $(J, C')$, where
\begin{equation}
    C'=
\scalemath{0.8}{\left(\begin{array}{ccccccccccc}
             1      & \oh      & \oh      & \oh      & \oh      & \oh      & \oh      & \oh \\
             1      & 1      & \oh      & \oh      & \oh      & \oh      & \oh      & \oh \\
             1      & 1      & 1      & \oh      & \oh      & \oh      & \oh      & \oh \\
             1      & 1      & 1      & 1      & \oh      & \oh      & \oh      & \oh \\
             \oh      & \oh      & \oh      & \oh      & 1      & 1      & \oh      & \oh \\
             \oh      & \oh      & \oh      & \oh      & \oh      & 1      & 1      & \oh \\
             \oh      & \oh      & \oh      & \oh      & \oh      & \oh      & 1      & 1 \\
             \oh      & \oh      & \oh      & \oh      & \oh      & \oh      & \oh      & 1 \\
        \end{array}\right)}
        \,. \label{eq:Gray net}
\end{equation}
The corresponding sequence is generated by
\begin{equation}
    \scalemath{0.8}{\left( \left(\begin{array}{ccccccccccc}
             1      & \oh      & \oh      & \oh      & \oh      & \oh      & \oh      & \oh \\
             \oh      & 1      & \oh      & \oh      & \oh      & \oh      & \oh      & \oh \\
             1      & 1      & 1      & \oh      & \oh      & \oh      & \oh      & \oh \\
             \oh      & \oh      & \oh      & 1      & \oh      & \oh      & \oh      & \oh \\
             1      & \oh      & \oh      & \oh      & 1      & \oh      & \oh      & \oh \\
             1      & 1      & \oh      & \oh      & 1      & 1      & \oh      & \oh \\
             1      & \oh      & 1      & \oh      & 1      & \oh      & 1      & \oh \\
             1      & 1      & 1      & 1      & 1      & 1      & 1      & 1 \\
        \end{array}\right),
        \left(\begin{array}{ccccccccccc}
             1      & 1      & 1      & 1      & 1      & 1      & 1      & 1 \\
             \oh      & 1      & \oh      & 1      & \oh      & 1      & \oh      & 1 \\
             1      & \oh      & \oh      & 1      & 1      & \oh      & \oh      & 1 \\
             \oh      & \oh      & \oh      & 1      & \oh      & \oh      & \oh      & 1 \\
             1      & 1      & 1      & 1      & \oh      & \oh      & \oh      & \oh \\
             1      & \oh      & 1      & \oh      & \oh      & \oh      & \oh      & \oh \\
             1      & 1      & \oh      & \oh      & \oh      & \oh      & \oh      & \oh \\
             1      & \oh      & \oh      & \oh      & \oh      & \oh      & \oh      & \oh \\
        \end{array}\right) \right)}\,. \label{eq:Gray sequence}
\end{equation}
See Fig.~\ref{fig:gray partition}(c). Note that matrix~\eqref{eq:Gray net} generating the Gray net has the same structure for all even $m$ we tried, whereas matrices~\eqref{eq:Gray sequence} look different for different sizes of the sequence. 



\subsection{Validation}

Let us report the following experiments that clearly demonstrate the utility of sequences over nets for progressive sampling: We considered the $256$-point LP net. First, we took a random ordering of the points and computed the star discrepancies of the first $16$ points, the first $32$ points, the first $48$ points, etc. Then we took the sequence ordering given by our algorithm (see Table~\ref{tab:methods_overview_part2}) and computed similar star discrepancies. Finally, we computed the ratios of the $16$ numbers for the random ordering to the $16$ numbers for the sequence ordering, as depicted in Table~\ref{tab:ratio}.
\begin{table*}[htb]
\caption{The ratio of the star discrepancy of the first $N$ points of the $256$-point LP net for the random ordering to the one for the sequence ordering}
    \centering
    \bluevar{
    \begin{tabular}{|c|c|c|c|c|c|c|c|c|c|c|c|c|c|c|c|c|}
    \hline
         $N$ & 16 & 32 & 48 & 64 & 80 & 96 & 112 & 128
           & 144 & 160 & 176 & 192 & 208 & 224 & 240 & 256\\
    \hline  
         Ratio & 2.22 & 2.25 & 2.23 & 2.60 & 2.59 & 2.71 & 3.04 & 3.94 & 2.98 & 3.16 & 2.65 & 2.25 & 1.85 &1.54 & 1.25 &1.00 \\
         \hline
    \end{tabular}
    }
    \label{tab:ratio}
\end{table*}
\bluevar{
 We obtained that the first number has become $2.22$ times smaller, the second one $2.25$ times smaller, etc. Note that the original order would have worse results than the randomly reordered one.
}


\section{Self-Similar Dyadic Sequences\label{sec:xi-sequences}}

In this section, we propose a family of digital dyadic sequences that is \remove{} efficient to compute and invert. At the same time, the family is large enough to allow for a meaningful exploration of the design space in order to select high-quality representatives.
Our construction is inspired by recursive tiling techniques for distributing blue-noise samples \cite{Ostromoukhov07Polyominoes,Wachtel14Fast,Ahmed17ART}.

\begin{figure}
    \setlength{\unit}{(10\columnwidth - 30\gap)/35}              
   {\scriptsize
    \begin{tabular*}{\columnwidth}{@{}
c@{\extracolsep{\fill}}c@{\extracolsep{\fill}}c@{\extracolsep{\fill}}c@{}}%
        \includegraphics[height=1\unit]{images/teaser/xi-0.pdf}&%
        \includegraphics[height=1\unit]{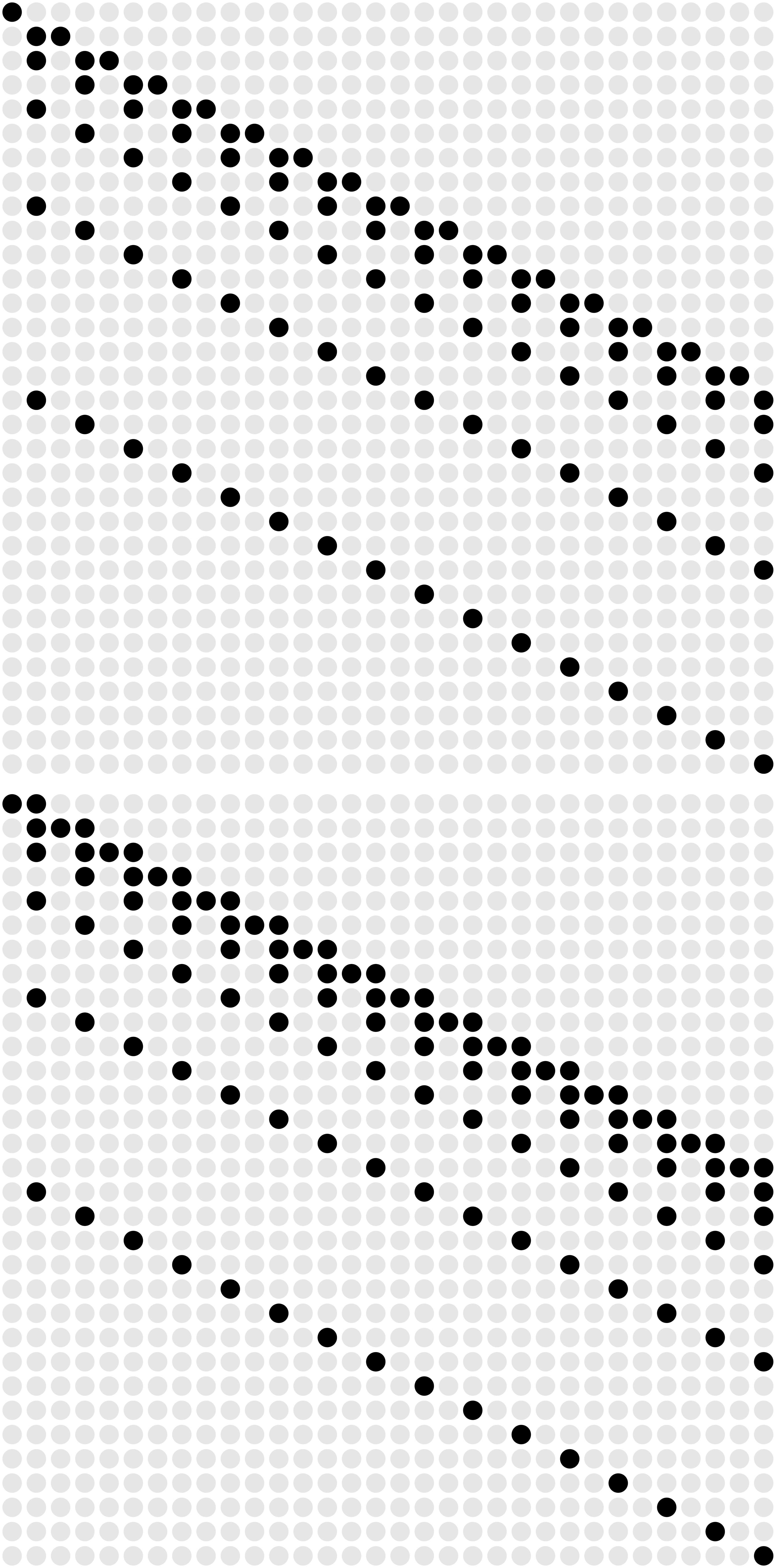}&%
        \includegraphics[height=1\unit]{images/teaser/xi-0-points.pdf}&%
        \includegraphics[height=1\unit]{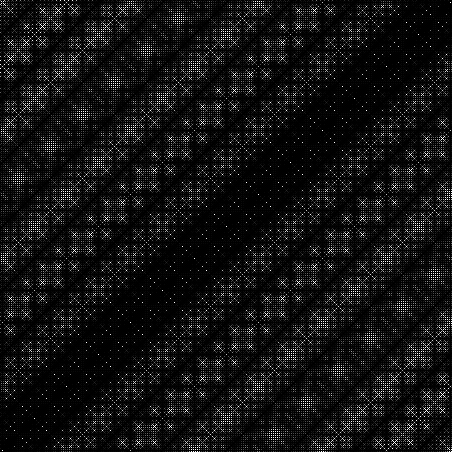}\\%
        \includegraphics[height=1\unit]{images/teaser/xi-1.pdf}&%
        \includegraphics[height=1\unit]{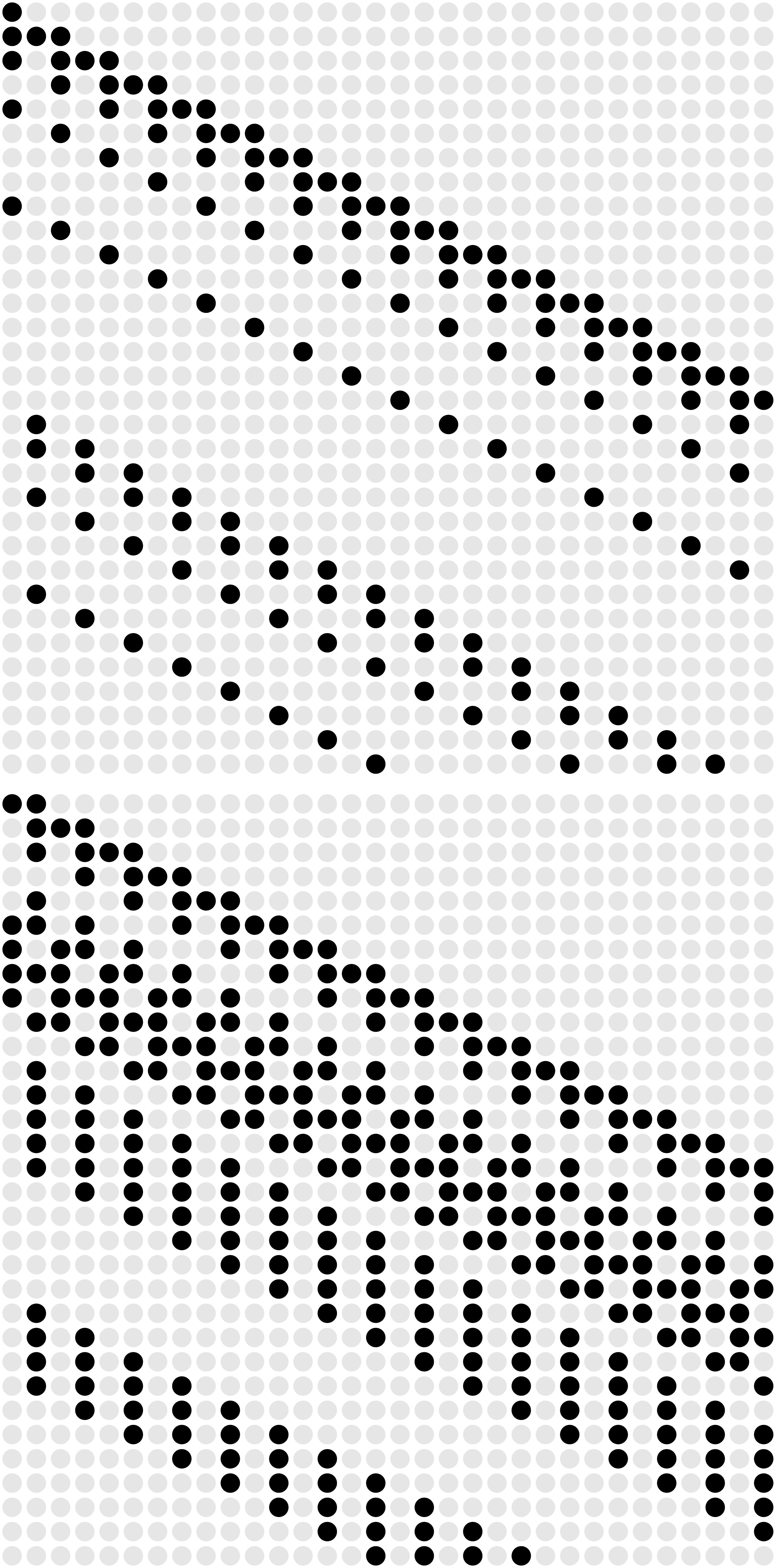}&%
        \includegraphics[height=1\unit]{images/teaser/xi-1-points.pdf}&%
        \includegraphics[height=1\unit]{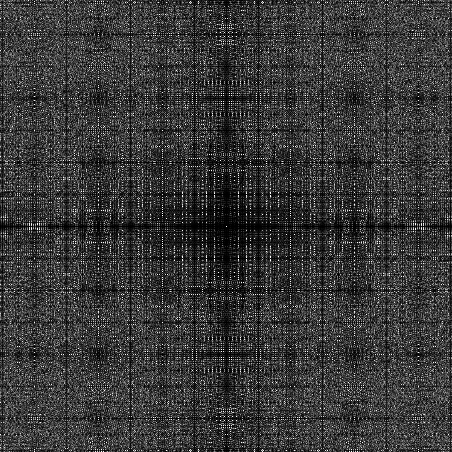}\\%
        \includegraphics[height=1\unit]{images/teaser/xi-2.pdf}&%
        \includegraphics[height=1\unit]{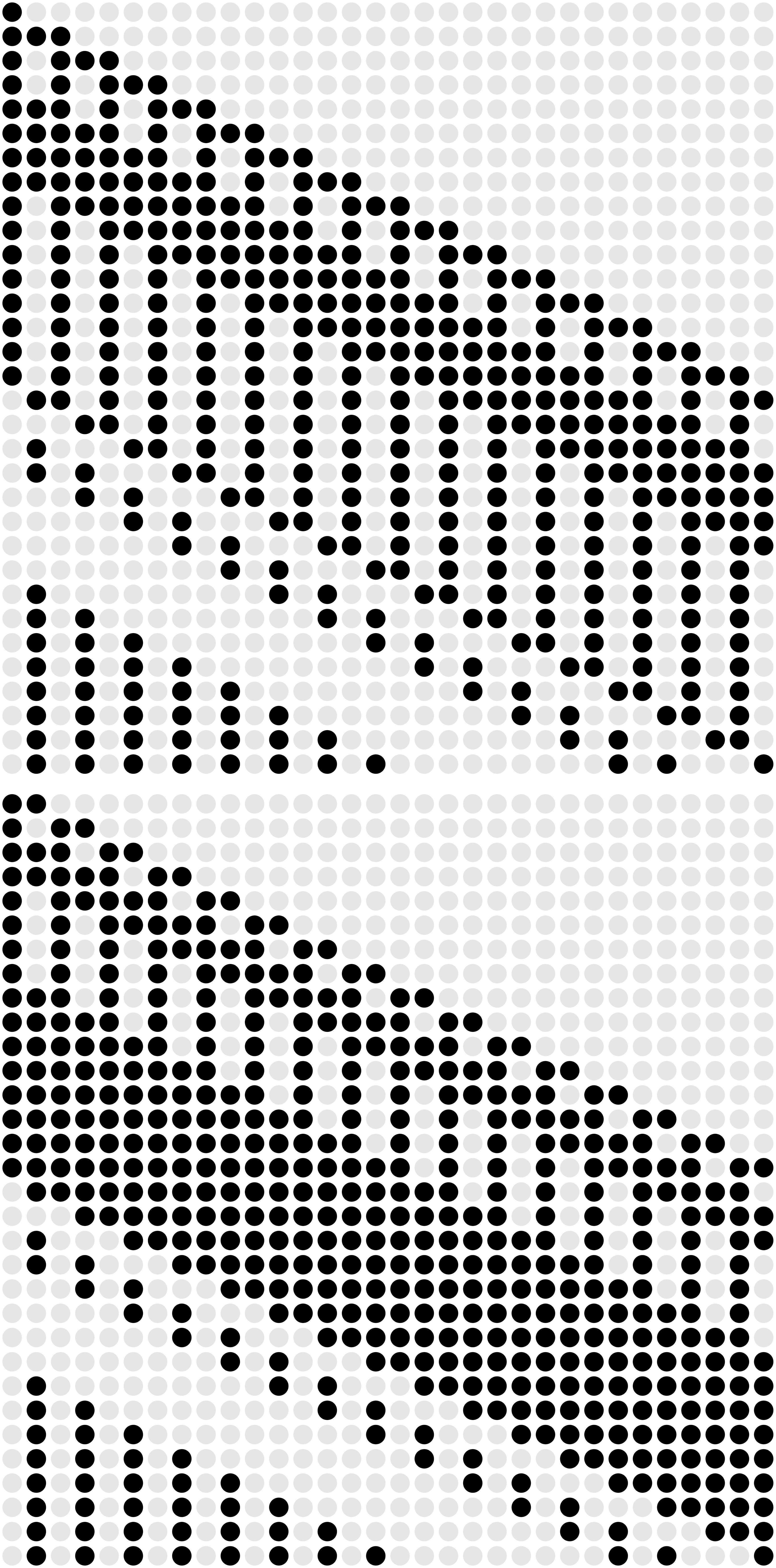}&%
        \includegraphics[height=1\unit]{images/teaser/xi-2-points.pdf}&%
        \includegraphics[height=1\unit]{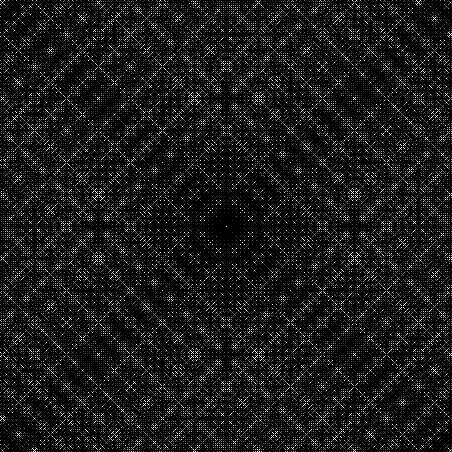}
   \end{tabular*}}
    \vspace{-2mm}
    \caption{%
       \label{fig:xi}
       Structure, generating matrices, points, and periodograms of various self-similar $\xi$-sequences. The plots to the left depict the first few points; larger circles show the points closer to the beginning of the sequence. The top plots show the $\xi_0$-sequence generated by $X=Y=(1,0,\dots,0)$.
       }
\end{figure}

We start with the geometric design requirements. After extensive experiments in trying to design a self-similar sequence, we settled for the following idea. 
We are looking for an infinite dyadic sequence such that every fourth point belongs to the bottom-left quadrant $[0,1/2)^2$ and for each $k=0,1,2,\dots$ the first $2^k$ points appearing in this quadrant form a scaled version of the first $2^k$ points of the whole sequence in $[0,1)^2$, with the same order. See Fig.~\ref{fig:xi} to the top-left.
One can formulate this design requirement as
\begin{equation}
    p_{4i} = p_i/2 \qquad\text{for }i=0,1,2,\dots\,.   \label{eq:self-similarity condition}
\end{equation}
This means that a point with four times the sequence number is scaled by the factor $1/2$. In practice, we restrict ourselves to finite digital dyadic sequences.
A digital dyadic sequence $p_0,p_1,\dots,p_{2^m-1}$ is called \emph{self-similar} if for each $i<2^{m-2}$ we have~\eqref{eq:self-similarity condition} with the right side truncated to $m$ digits (i.e., if the $(m+1)$-th digit after the point in the dyadic decomposition of one of the coordinates of $p_i/2$ is $1$ then it is replaced by $0$).
This simple equation is all that we need to realize the sequence digitally.

We break it down into three elements.
First, we note that multiplying the sequence number by $4$ means shifting the bit vector $S$ down two steps, inserting two leading zeros that effectively wipe out two columns of the progressive matrix $C_x$.
Second, \eqref{eq:self-similarity condition} 
means that the subsequent columns of the matrix should bear similar information to the wiped ones.
Finally, dividing the coordinates by 2 is equivalent to pushing all these subsequent columns one row down, inserting zeros.
This leads to the following self-similar matrix structure
\newcommand{\leftbar}[1]{\multicolumn{1}{|c}{#1}}
\newcommand{\rightbar}[1]{\multicolumn{1}{c|}{#1}}
\newcommand{\highlight}[1]{{\color{red}#1}}
\newcommand{\highlightmore}[1]{{\color{blue}#1}}
\begin{equation}
    C_x = \left(\begin{array}{cccccc}
         a_{0}  & \leftbar{\highlight{b_{0}}}   & \leftbar{\oh}       & \leftbar{\oh}       & \leftbar{\oh}       & \cdots \\ \cline{1-1}
         a_{1}  &          \highlight{b_{1}}    & \leftbar{a_{0}}   & \leftbar{b_{0}}   & \leftbar{\oh}       & \cdots \\ \cline{1-2}
         a_{2}  &          \highlight{b_{2}}    & a_{1}             & \leftbar{b_{1}}   & \leftbar{a_{0}}   & \cdots \\ \cline{1-3}
         a_{3}  &          \highlight{b_{3}}    & a_{2}             & b_{2}             & \leftbar{a_{1}}   & \cdots \\ \cline{1-4}
         a_{4}  &          \highlight{b_{4}}    & a_{3}             & b_{3}             & a_{2}             & \cdots \\
         \vdots & \vdots            & \vdots            & \vdots            & \vdots            & \ddots
    \end{array}\right)\,,   \label{eq:new b's}
\end{equation}
for some column bit vectors $A=(a_0, a_1, \cdots)$ and $B=(b_0, b_1, \cdots)$ that, by definition, represent the $x$-coordinates of the second and third point in the sequence (equal to $0.a_0a_1\dots$ and $0.b_0b_1\dots$ respectively).
We take $A$ as a design parameter that can be freely chosen by the user, except for $a_0$ which has to be $1$. It then remains to set or compute the bits of $B$ so as to ensure a progressive matrix. We can iteratively expand the matrix.
The key insight is that each subsequent expansion includes only one new unknown bit from $B$, except for the first step where we have two bits when expanding from a $1 \times 1$ to a $2 \times 2$ matrix.
Formally, we have the following statement:
\begin{lemma}
    \label{lemma:unique solution}
    For each bit vector $A=(1,a_1,a_2,\cdots,a_{m-1})$ starting with bit $1$ there are exactly two bit vectors $B=(b_0, b_1, \cdots, b_{m-1})$ such that matrix~\eqref{eq:new b's} is progressive.
\end{lemma}
\begin{proof}
    Take an arbitrary bit $b_0$ and let us determine bits $b_1,b_2,\dots$ inductively.

    Assume that $b_0,\dots,b_{k-1}$ have already been determined so that the first $k$ leading principal minors of the matrix in Eq.~\eqref{eq:new b's} equal $1$. Let us show that then there is a unique $b_k$ such that the $(k+1)$-th leading principal minor $\Delta$ equals $1$; for clarity, we illustrate the argument for $k=3$. Take the Laplace expansion along the last row:
    \begin{multline*}
    \Delta=
    \left|\begin{array}{cccc}
        1      & b_{0}  & \oh      & \oh      \\
        a_{1}  & b_{1}  & 1      & b_{0}  \\
        a_{2}  & b_{2}  & a_{1}  & b_{1}  \\
        a_{3}  & b_{3}  & a_{2}  & b_{2}    
   \end{array}\right|
   =
   a_{3}
   \left|\begin{array}{ccc}
        b_{0}  & \oh      & \oh      \\
        b_{1}  & 1      & b_{0}  \\
        b_{2}  & a_{1}  & b_{1}        
   \end{array}\right|
   +\\+
   b_{3}
   \left|\begin{array}{ccc}
        1      &  \oh      & \oh      \\
        a_{1}  &  1      & b_{0}  \\
        a_{2}  &  a_{1}  & b_{1}         
   \end{array}\right|
   +
   a_2
   \left|\begin{array}{ccc}
        1      & b_{0}  & \oh      \\
        a_{1}  & b_{1}  & b_{0}  \\
        a_{2}  & b_{2}  & b_{1}  
   \end{array}\right|
   + 
   b_2
   \left|\begin{array}{ccc}
        1      & b_{0}  & \oh      \\
        a_{1}  & b_{1}  & 1      \\
        a_{2}  & b_{2}  & a_{1}  
   \end{array}\right|.
   \end{multline*}
   By the structure of the matrix in Eq.~\eqref{eq:new b's}, the coefficient at $b_k$
   here
   $$
   \left|\begin{array}{ccc}
        1      &  \oh      & \oh      \\
        a_{1}  &  1      & b_{0}  \\
        a_{2}  &  a_{1}  & b_{1}         
   \end{array}\right|
   =
   \left|\begin{array}{cc}
        1      & b_{0}  \\
        a_{1}  & b_{1}         
   \end{array}\right|
   $$
   is exactly the $(k-1)\times(k-1)$ leading principal minor, equal to $1$ by the inductive hypothesis. Thus there is a unique $b_k$ such that $\Delta=1$. 
   Note that we compute the determinant of a $(k+1) \times (k+1)$ matrix using the assumption for the $(k-1) \times (k-1)$ matrix. Depending on $k$ being even or odd the matrix structure is slightly different, and we only illustrate the case for $k$ being odd here, but the even case is analogous.
\end{proof}

Similarly to the case of arbitrary progressive matrices in Section~\ref{sec:digital sequences}, there is a closed-form solution for the self-similar case.
Starting with the simpler case $A=(1,0,0,\dots)$, i.e. setting the $x$-coordinate of the second point in the sequence to be $1/2=0.100\dots $, we get the following assertion:
if $A=(1,0,0,\dots)$ in progressive matrix~\eqref{eq:new b's}, then $B$ is one of the vectors
\begin{align}
        (0,\xi_1,\xi_2,\dots)&=(0,1,1,0,1,0,0,0,1,0,0,0,0,0,0,0,1,\dots) \,,\\
        (1,\xi_1,\xi_2,\dots)&=(1,1,1,0,1,0,0,0,1,0,0,0,0,0,0,0,1,\dots),
    \end{align}
where $\xi_i=1$ if and only if $i$ is a power of $2$.
In other words, the third point of the sequence has one of the (truncated) $x$-coordinates 
\begin{align} \label{eq-xi}
        \xi&=0.01101000100000001\ldots = \frac{1}{2}\sum_{k=0}^{\infty} 2^{-2^k}\,,\\
        \label{eq-xi-plus}
        \xi^{+}&=0.11101000100000001\ldots = \frac{1}{2} + \frac{1}{2}\sum_{k=0}^{\infty} 2^{-2^k}\,.
\end{align}
This is proved by induction similarly to Lemma~\ref{lemma:unique solution}.
This gives us the following self-similar progressive matrix:
\begin{equation}
    C_{\xi} = \left(\begin{array}{ccccccccccc}
         1      & \oh      & \oh      & \oh      & \oh      & \oh      & \oh      & \oh      & \oh      & \oh      & \cdots \\
         \oh      & 1      & 1      & \oh      & \oh      & \oh      & \oh      & \oh      & \oh      & \oh      & \cdots \\
         \oh      & 1      & \oh      & 1      & 1      & \oh      & \oh      & \oh      & \oh      & \oh      & \cdots \\
         \oh      & \oh      & \oh      & 1      & \oh      & 1      & 1      & \oh      & \oh      & \oh      & \cdots \\
         \oh      & 1      & \oh      & \oh      & \oh      & 1      & \oh      & 1      & 1      & \oh      & \cdots \\
         \oh      & \oh      & \oh      & 1      & \oh      & \oh      & \oh      & 1      & \oh      & 1      & \cdots \\
         \oh      & \oh      & \oh      & \oh      & \oh      & 1      & \oh      & \oh      & \oh      & 1      & \cdots \\
         \oh      & \oh      & \oh      & \oh      & \oh      & \oh      & \oh      & 1      & \oh      & \oh      & \cdots \\
         \oh      & 1      & \oh      & \oh      & \oh      & \oh      & \oh      & \oh      & \oh      & 1      & \cdots \\
         \oh      & \oh      & \oh      & 1      & \oh      & \oh      & \oh      & \oh      & \oh      & \oh      & \cdots \\
         \vdots & \vdots & \vdots & \vdots & \vdots & \vdots & \vdots & \vdots & \vdots & \vdots & \ddots 
    \end{array}\right)\,,   \label{eq:C_xi}
\end{equation}
and a similarly looking matrix for $C_{\xi^{+}}$.
Note that this matrix looks quite different from the conventional upper-triangular form of Sobol matrices.
Note also that $\xi$ appears also along the diagonal.

Now we move to the general case, where we have
\begin{lemma}\label{th-b-via-a}
    For arbitrary $A=(1,a_1,a_2,\dots)$ the two vectors 
    \begin{align}
    B&=\xi(A):=(0,\,\xi_1,\,a_1\xi_1+\xi_2,\,a_2\xi_1+a_1\xi_2+\xi_3, \,\dots) \,\\
    B^{+}&=\xi^{+}(A):=\xi(A)+A
    \end{align}
    are all the vectors making matrix~\eqref{eq:new b's} progressive.
\end{lemma}
In other words, the $x$-coordinate of the third point in the sequence is obtained by the carryless multiplication of the $x$-coordinate of the second point by one of numbers \eqref{eq-xi}--\eqref{eq-xi-plus}, followed by usual multiplication by $2$.
\begin{proof}
    Multiply $C_{\xi}$ (or $C_{\xi^{+}}$) by
    the lower unitriangular matrix
    \begin{equation}
        L = \left(\begin{array}{cccc}
             1   & \oh   & \oh   & \cdots \\
             a_1 & 1   & \oh   & \cdots \\
             a_2 & a_1 & 1   & \cdots \\
             a_3 & a_2 & a_1 & \cdots \\
             \vdots & \vdots & \vdots & \ddots
        \end{array}\right)
    \end{equation}
    from the left. This keeps the matrix progressive (LU-decomposable), preserves structure~\eqref{eq:new b's}, and places $A$ and $\xi(A)$ (or $\xi^{+}(A)$) in the first two columns. Then the proof is concluded by Lemma~\ref{lemma:unique solution}.
\end{proof}

Denoting the corresponding matrix by $C_{\xi(A)}:=LC_\xi$ and $C_{\xi^{+}(A)}:=LC_{\xi^{+}}$, we have the following statement to put it together:
\begin{theorem}
    \label{th:xi pairing}
    For any pair of bit vectors $X$ and $Y$ starting from $1$  in the first position, the pairs $(C_{\xi(X)},C_{\xi^{+}(Y)})$ and $(C_{\xi^{\bluevar{+}}(X)},C_{\xi(Y)})$ are progressive. Conversely, each progressive pair of matrices of form~\eqref{eq:new b's} equals either $(C_{\xi(X)},C_{\xi^{+}(Y)})$ or $(C_{\xi^{+}(X)},C_{\xi(Y)})$ for some $X$ and $Y$ starting from $1$.
\end{theorem}

\begin{figure*}
    \centering
  \setlength{\unit}{(\textwidth - 6\gap)/4}
  {\scriptsize
    \begin{tabular*}{\textwidth}{@{}c@{\extracolsep{\fill}}c@{\extracolsep{\fill}} c@{\extracolsep{\fill}}c@{}}
        \includegraphics[width=1\unit]{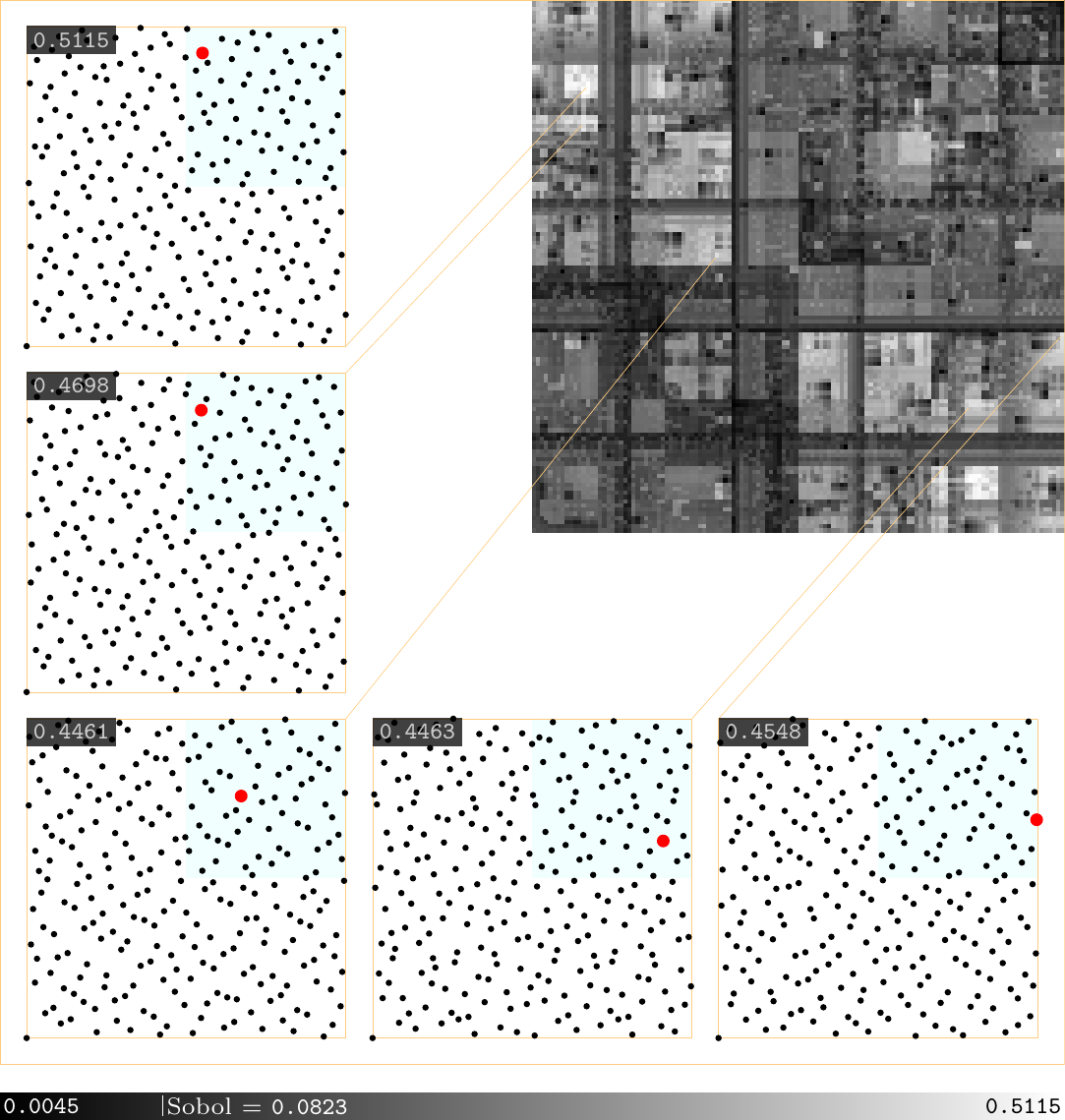}&
        \includegraphics[width=1\unit]{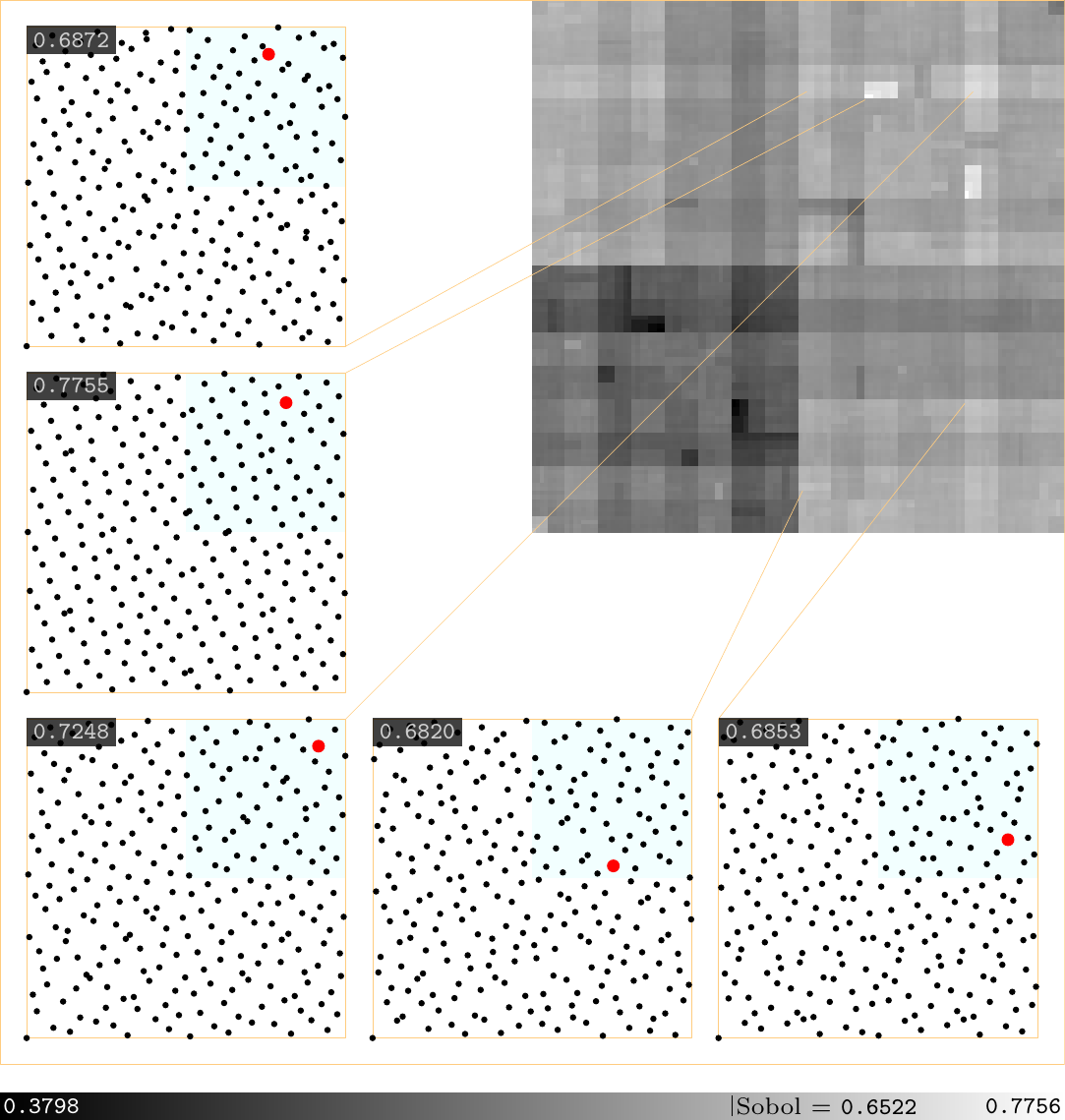}&
        \includegraphics[width=1\unit]{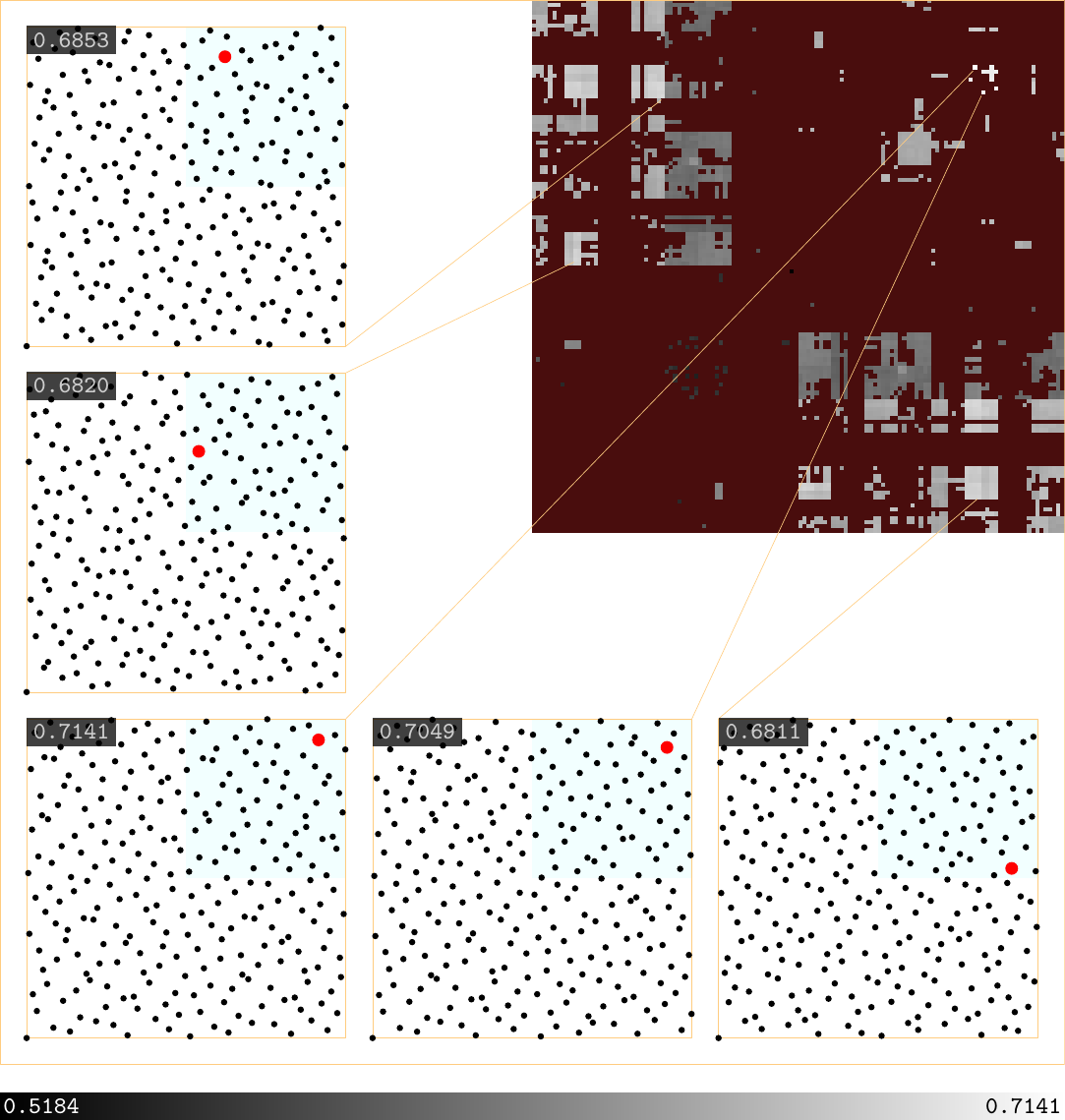}&
        \includegraphics[width=1\unit]{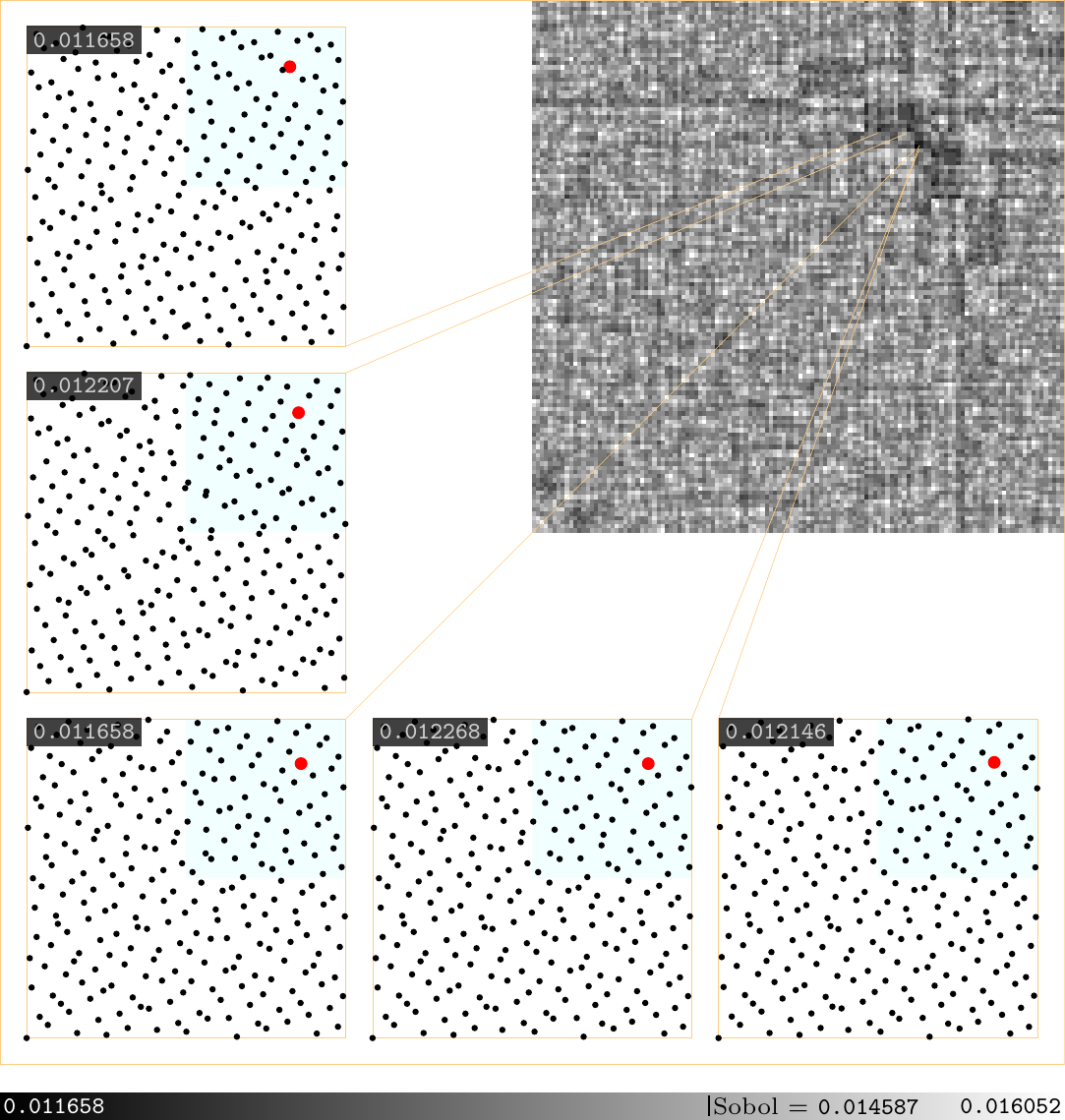}\\[2mm]
        (a) Minimum distance between points & 
        (b) Average distance to the closest neighbor &
        (c) The same for the minimum distance $> 0.3$ &
        (d) Star discrepancy
    \end{tabular*}}
    \vspace{-3mm}
    \caption{%
        \label{fig:atlas}
         An exhaustive exploration of $\xi$-sequences containing 256 points. Each pixel in the top-right square $[1/2;1)^2$ represents a possible position of the second point of the sequence, which uniquely determines the whole 
         $\xi$-sequence. There are $128\times 128$ pixels  in total and 16K possible $\xi$-sequences. The color depicts one of the (suitably normalized) quality measures (a)--(d) of the resulting $\xi$-sequence. 
         In (c), the average distance to the closest neighbor is shown only for the sequences such that the minimum distance is greater than $0.3$; the remaining pixels are brown.
         A few best-quality representatives are in the insets; each red bold point depicts the second point in the sequence. For comparison, the same quality measures for the Sobol sequence are shown in the legends.  
     }
    \vspace{-4mm}
\end{figure*}

\begin{proof}
    By construction, the matrices $C_{\xi(X)},C_{\xi^{+}(Y)}$ have form~\eqref{eq:new b's} with $(a_0,b_0)=(1,0)$ and $(1,1)$ respectively; i.e., our pair has form
    \begin{equation*}
    \left(
    \left(\begin{array}{ccccc}
         1      & \oh   & \oh   & \oh   & \cdots \\ 
         a_{1}  & b_{1} & 1     & \oh   & \cdots \\ 
         a_{2}  & b_{2} & a_{1} & b_{1} & \cdots \\ 
         \vdots & \vdots& \vdots& \vdots& \ddots
    \end{array}\right)\,,
    \left(\begin{array}{ccccc}
         1      & 1     & \oh   & \oh   & \cdots \\ 
         c_{1}  & d_{1} & 1     & 1     & \cdots \\ 
         c_{2}  & d_{2} & c_{1} & d_{1} & \cdots \\ 
         \vdots & \vdots& \vdots& \vdots& \ddots
    \end{array}\right)\right)
\end{equation*}
for some $a_i,b_i,c_i,d_i$. The two matrices are themselves progressive by Lemma~\ref{th-b-via-a}. It remains to prove that hybrid matrix~\eqref{eq:hybrid matrix 2} has determinant $1$ for $r=1,\dots,k-1$. 
We compute the determinant recursively. Since the first row of~\eqref{eq:hybrid matrix 2} is $(1\, 0\, \dots\, 0)$ in our case, it follows that removing the first row and the first column preserves the determinant:
\begin{equation*}
    \left|\begin{array}{ccccc}
         {1}      & \oh   & \oh   & \oh   & {\cdots} \\ \cline{2-5}
         a_{1}  & \leftbar{b_{1}} & {1}     & {.}   & \rightbar{\cdots} \\ 
         a_{2}  & \leftbar{b_{2}} & {a_{1}}& b_{1} & \rightbar{\cdots} \\ 
         \vdots & \leftbar{\vdots}& {\vdots}& \vdots& \rightbar{\ddots} \\ 
         {1}    & \leftbar{1}     & \oh     & \oh   & \rightbar{\cdots} \\ 
         c_{1}  & \leftbar{d_{1}} & {1}     & 1     & \rightbar{\cdots} \\ 
         c_{2}  & \leftbar{d_{2}} & {c_{1}} & d_{1} & \rightbar{\cdots} \\ 
         \vdots & \leftbar{\vdots}& {\vdots}& \vdots& \rightbar{\ddots}  \\ \cline{2-5}
    \end{array}\,\right|
    =
    \left|\,\begin{array}{cccc}
         \cline{1-4}
         \leftbar{b_{1}} & {1}     & \oh   & \rightbar{\cdots} \\ 
         \leftbar{b_{2}} & {a_{1}}& b_{1} & \rightbar{\cdots} \\ 
         \leftbar{\vdots}& {\vdots}& \vdots& \rightbar{\ddots} \\ 
         \leftbar{\highlight{1}}     & \highlight{0}     & \highlight{0}   & \rightbar{\highlight{\cdots}} \\ 
         \leftbar{d_{1}} & {1}     & 1     & \rightbar{\cdots} \\ 
         \leftbar{d_{2}} & {c_{1}} & d_{1} & \rightbar{\cdots} \\ 
         \leftbar{\vdots}& {\vdots}& \vdots& \rightbar{\ddots}  \\ 
         \cline{1-4}
    \end{array}\,\right|.
    \end{equation*}
In the resulting matrix, the $(k-r)$-th row (highlighted) becomes $(1\, 0\, \dots\, 0)$. Now removing this row and the first column, we get
\begin{equation*}
    \left|\begin{array}{cccc}
         \cline{2-4}
         b_{1} & \leftbar{1}     & \oh   & \rightbar{\cdots} \\ 
         b_{2} & \leftbar{a_{1}} & b_{1} & \rightbar{\cdots} \\ 
         \vdots& \leftbar{\vdots}& \vdots& \rightbar{\ddots} \\ 
         \cline{2-4}
         \highlight{1}       & \highlight{0}   & \highlight{0}  & \highlight{\cdots} \\ 
         \cline{2-4}
         d_{1} & \leftbar{1}     & 1     & \rightbar{\cdots} \\ 
         d_{2} & \leftbar{c_{1}} & d_{1} & \rightbar{\cdots} \\ 
         \vdots& \leftbar{\vdots}& \vdots& \rightbar{\ddots}  \\ 
         \cline{2-4}
    \end{array}\,\right|
    =
    \left|\,\begin{array}{ccc}
    \cline{1-3}
         \leftbar{{1}}     & \oh   & \rightbar{{\cdots}} \\ 
         \leftbar{a_{1}} & b_{1} & \rightbar{\cdots} \\ 
         \leftbar{\vdots}& \vdots& \rightbar{\ddots} \\ \cline{1-3}
         \\[-0.3cm] \cline{1-3}
         \leftbar{{1}}     & {1}     & \rightbar{{\cdots}} \\ 
         \leftbar{c_{1}} & d_{1} & \rightbar{\cdots} \\ 
         \leftbar{\vdots}& \vdots& \rightbar{\ddots} \\ \cline{1-3}
    \end{array}\,\right|.
    \end{equation*}
The resulting matrix has the same structure as initial matrix~\eqref{eq:hybrid matrix 2}, only $r$ and $k$ are reduced by $1$ and $2$ respectively. We now repeat the procedure, each time removing the row $(1\, 0\, \dots\, 0)$ and the first column in the current matrix until we get $r=0$ or $r=k$. The resulting determinant will be a leading principal minor of $C_{\xi(X)}$ or $C_{\xi^{+}(Y)}$. It equals $1$ because $C_{\xi(X)}$ or $C_{\xi^{+}(Y)}$ are themselves progressive.

The converse assertion follows from Lemma~\ref{th-b-via-a} because the pairs $(C_{\xi(X)},C_{\xi(Y)})$ and $(C_{\xi^{+}(X)},C_{\xi^{+}(Y)})$ are never progressive.
\end{proof}

We find these results striking, since this constant $\xi$ \remove{} handles all the effort of matrix derivation and pairing we encountered in Section~\ref{sec:digital sequences}.
Thus, the resulting self-similar sequences generated by $(C_{\xi(X)},C_{\xi^+(Y)})$ for various $X$ and $Y$ are
deservedly called \emph{$\xi$-sequences}. See Fig.~\ref{fig:xi} and~\ref{fig:atlas}. (We do not need to consider the pairs $(C_{\xi^+(X)},C_{\xi(Y)})$ because this leads just to the swap of $x$- and $y$-coordinates.) The sequence generated by $(C_{\xi},C_{\xi^+})$ is referred as $\xi_0$-sequence. 
On the practical side, we truncate the coordinates of the $2^m$ points not to $m$ bits as required by the digital construction, but to the machine precision of $M=32$ (or $M=64$) bits. In other words, we always consider the first $2^m$ points of a $2^M$-point $\xi$-sequence. We note that the special structure of $\xi$ reduces the carryless multiplication to $ {O}\left(\log M\right)$ from $ {O}(M)$.
In 32-bit precision, it takes precisely $5$ shift-and-xor operations.
This makes setting up a $\xi$-sequence, given $(X,Y)$, 
incredibly low-cost, even cheaper than generating a random sample for $(X,Y)$. 
Algorithm~\ref{alg:setup} summarizes the preceding derivations into a set of steps for creating a $\xi$-sequence.
\SetKwInOut{KwIn}{Input}
\SetKwInOut{KwOut}{Output}
\begin{algorithm} [tb]
    \caption{
        Constructing the first $4$ points of a $\xi$-sequence using 32-bit precision
        from two bit-vectors $X=(x_0,\dots)$ and $Y=(y_0,\dots)$ with $x_0 = y_0 = 1$.
        All additions and multiplications are carryless.
    }
    \label{alg:setup}
    \KwIn{
         Two bit-vectors $X$ and $Y$ (the coordinates of $p_1$);\\
    }
    \KwOut{The first $4$ points $p_0,\dots,p_3$ 
    of the sequence.
    }
    $\xi(A) \leftarrow (2^{-1} + 2^{-2} + 2^{-4} + 2^{-8} + 2^{-16})A$\;
    $B \leftarrow \xi(X)$\;
    $B^+ \leftarrow \xi(Y) + Y$\;
    $p_0 \leftarrow (0,0)$\;
    $p_1 \leftarrow (X,Y)$\;
    $p_2 \leftarrow (B,B^+)$\;
    $p_3 \leftarrow (X,Y) + (B,B^+)$.
    %
\end{algorithm}
\begin{algorithm} [tb]
    \caption{
        C code for retrieving a sample from a $\xi$-sequence. The sequence is encoded by the first $4$ points $p[0], \dots, p[3]$. 
        The input is the number seqNo of the sample. 
    }
    \label{alg:retrieving a sample}
{\small
    \begin{lstlisting}[language={C++},basicstyle=\ttfamily]
Point getSample(uint32_t seqNo) {
    uint32_t x(0), y(0);
    for (int depth = 0; depth < 16; depth++) {
        x ^= p[seqNo & 3].x >> depth;
        y ^= p[seqNo & 3].y >> depth;
        seqNo >>= 2;               
    }
    return {x, y};
}
    \end{lstlisting}
}
\end{algorithm}

\bluevar{To our knowledge,}
$\xi$-sequences are the fastest way to construct a dyadic sequence and draw different dyadic nets.
They enable constructing a \blueit{pseudo-random} sequence for each pixel.
However, the fast construction turns \blueit{out} to be not the only advantage: the self-similar structure actually embodies a wealth of versatility that we are going to discuss in the following subsections.


\subsection{Retrieving the Samples}

Points from $\xi$-sequences may be generated from the matrices as any standard digital sequence. The self-similar structure of the matrices, however, offers a short-hand, as outlined in Algorithm~\ref{alg:retrieving a sample}.

Thus, besides fast construction, drawing the samples is an aspect where $\xi$-sequences significantly stand out compared to all alternatives we are aware of.
The simple loop structure of $\xi$-sequences, and the few parameters, make it a good candidate for low-level optimization, using registers, for example, but our plain C code already delivers 
high performance, as depicted in Table~\ref{tab:speed}.
\begin{table}
    \caption{%
        \label{tab:speed}%
        Speed performance, in million points per second, of $\xi$-sequences and other state-of-the-art sampler generators.
    }
    \vspace{-2mm}
\centering{}%
\begin{tabular}{|l|r|r|r|r|r|r|}
    \hline
    $m = \log_2(N)$ &  8 & 12 &  16 &  20 &  24 &   28 \\
    \hline 
    Random              & 18 &  25 &  24 &  56 &  68 &  71 \\
    Sobol               & 32 &  35 &  28 &  53 &  57 &  51 \\
    Burley              & 11 &  16 &  13 &  31 &  39 &  39 \\
    $\xi$               & 17 &  26 &  24 &  57 &  74 &  73 \\
    $\xi_{256}$         & 32 & 113 & 146 & 185 & 411 & 410 \\
    \hline 
\end{tabular}
    \vspace{-4mm}
\end{table}

While the straightforward implementation of $\xi$-sequences shown in Algorithm~\ref{alg:retrieving a sample} already leads the competition, the self-similar structure of the sequence permits ever further speed up through a small lookup table.
In the basic implementation we derive the basis vectors from the first four points, but in fact any power-of-four number of points may be used in the same way as basis for retrieving subsequent points\notforarxiv{, as may be found in our supplementary code}.
The $\xi_{256}$ entry in Table~\ref{tab:speed} refers to an implementation that uses a $256$-vector look-up table, and it leads to 2-6$\times$ speed up.


\subsection{Memory Footprint\label{sec:memory}}

In addition to their \remove{} speed performance, $\xi$-sequences consume very little memory, and offer a flexible trade-off between memory and speed. The sequence may be compressed to a single 2D bit vector ${p}_1$, a pair of vectors $p_1,p_2$, 
or the standard four vectors $p_0,\dots,p_3$, 
and may be expanded to a 16, 256, or even 64K lookup table for two-step retrieval.
This flexibility with memory makes $\xi$-sequences quite GPU friendly.
Indeed, providing random numbers in a GPU context is a standing challenge, and $\xi$-sequences help in this direction by making it possible to run a unique random sequence in each thread, passing a table of parameters that identify the sequences.
We implemented a test of this, and generated over 17G points per second, roughly 1T per minute, on a TITAN Xp GPU.

\begin{figure}
    \centering
    \settoheight{\labelHeight}{Random $\xi$}
  \setlength{\unit}{(\columnwidth - \labelHeight - 4\gap) / 4}
  {\scriptsize
    \includegraphics[width=\columnwidth]{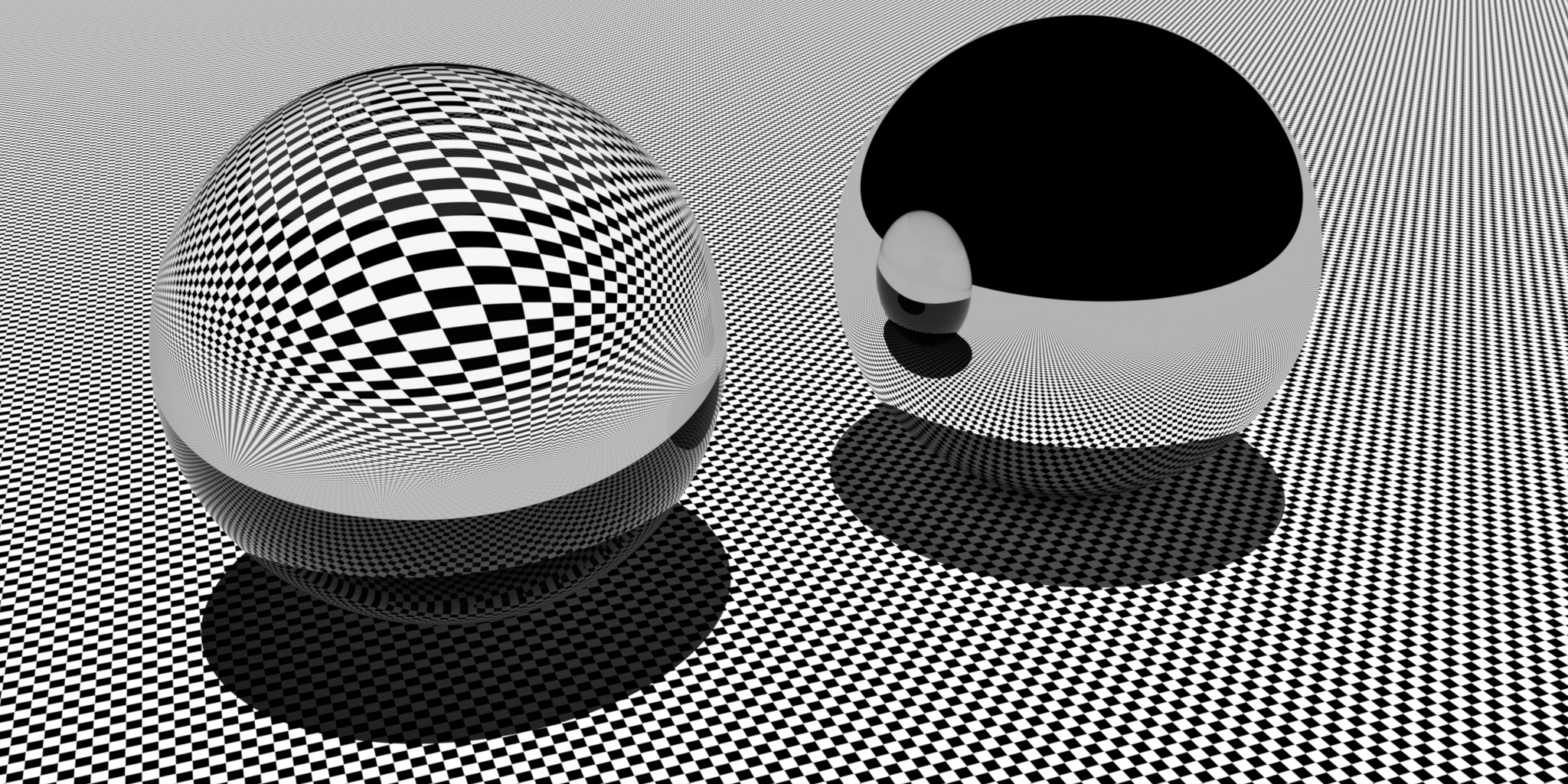}\\[1mm]%
    \begin{tabular*}{1\columnwidth}{@{}c@{\extracolsep{\fill}}c@{\extracolsep{\fill}}c@{\extracolsep{\fill}}c@{\extracolsep{\fill}}c@{}}
        \rotatebox{90}{\parbox{1\unit}{\centering (c) (0, 2)-Sequence}}&%
        \includegraphics[width=1\unit]{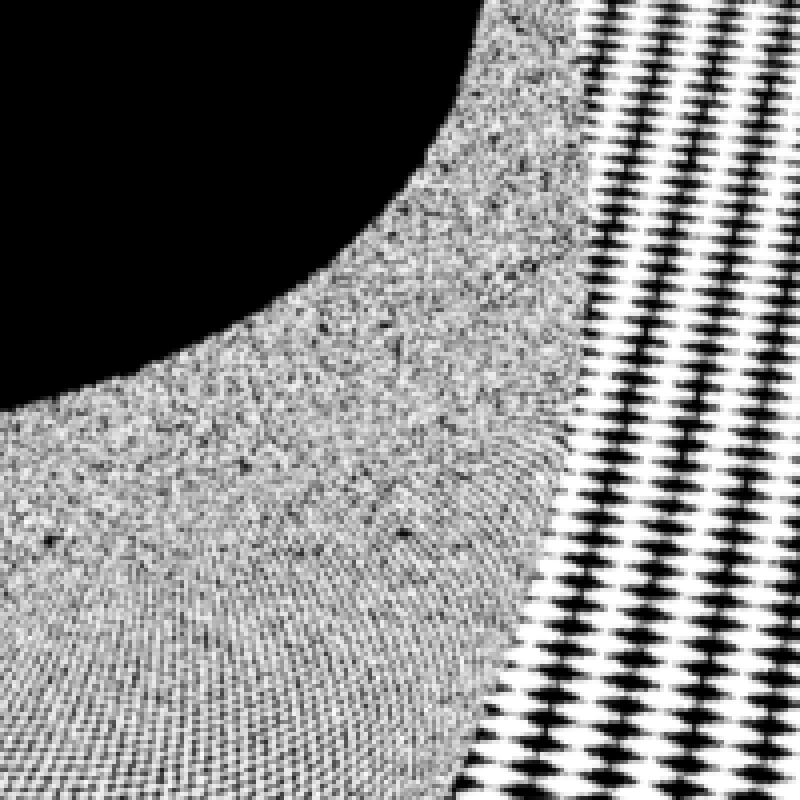}&%
        \includegraphics[width=1\unit]{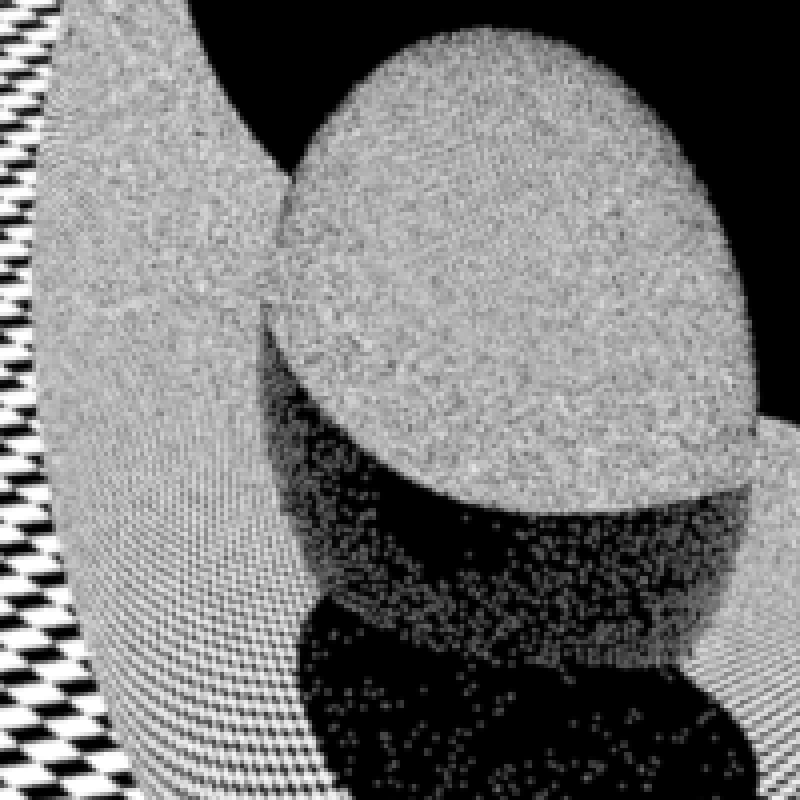}&%
        \includegraphics[width=1\unit]{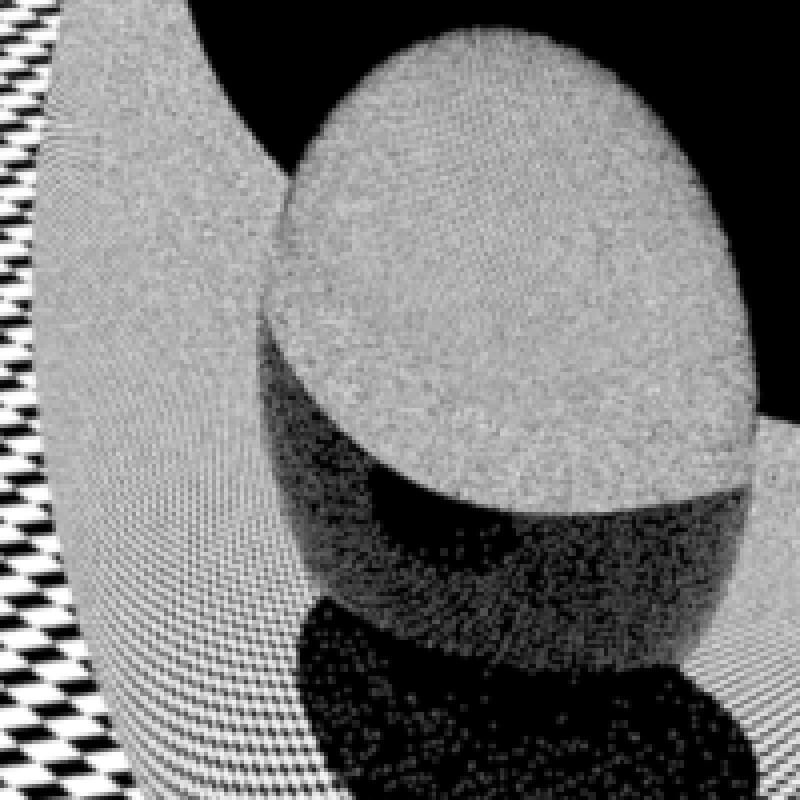}&%
        \includegraphics[width=1\unit]{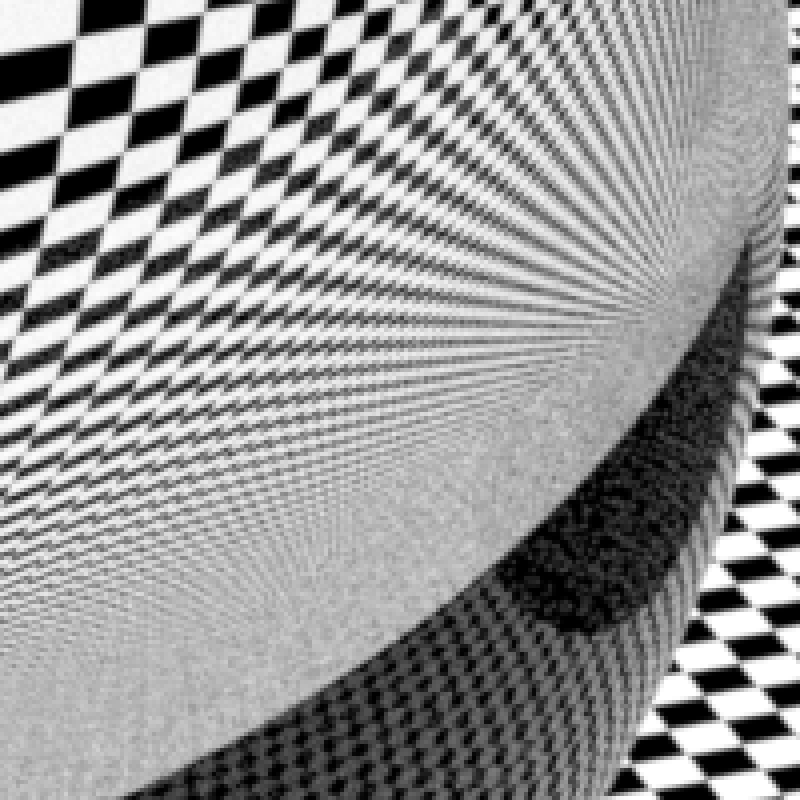}\\%
        \rotatebox{90}{\parbox{1\unit}{\centering (b) Random $\xi$}}&%
        \includegraphics[width=1\unit]{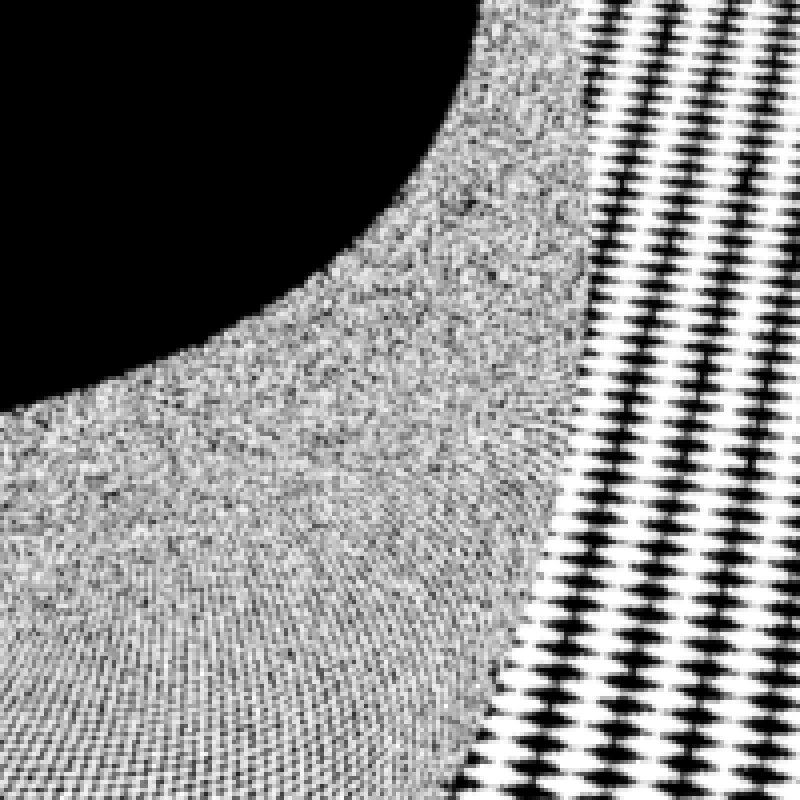}&%
        \includegraphics[width=1\unit]{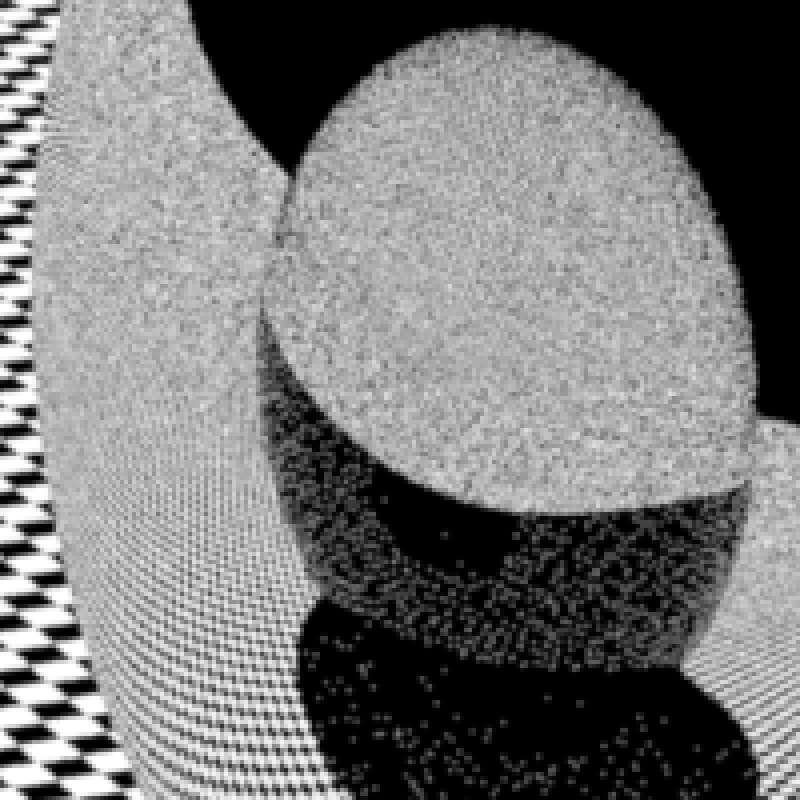}&%
        \includegraphics[width=1\unit]{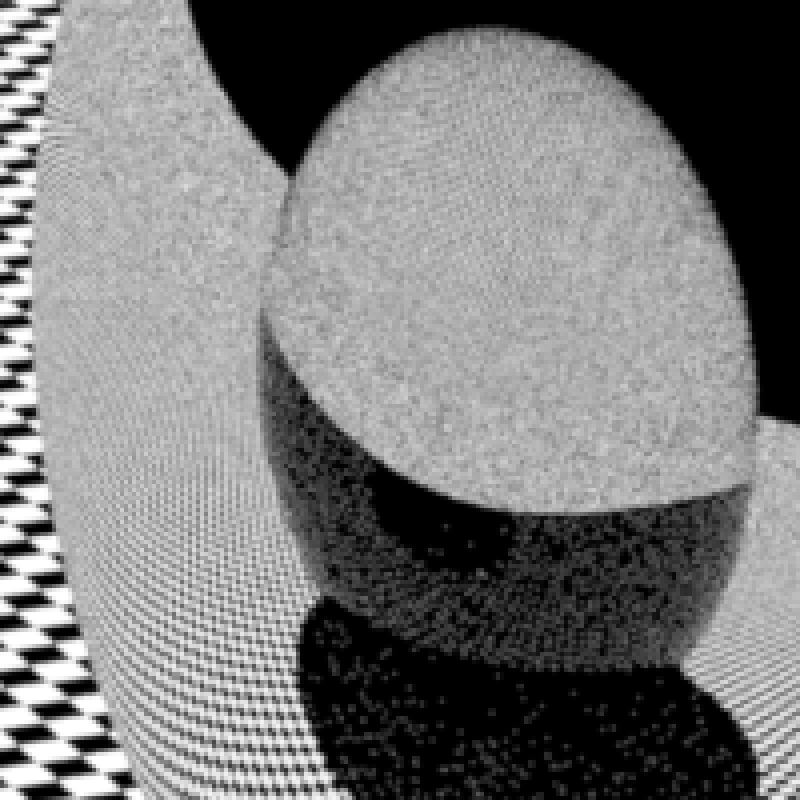}&%
        \includegraphics[width=1\unit]{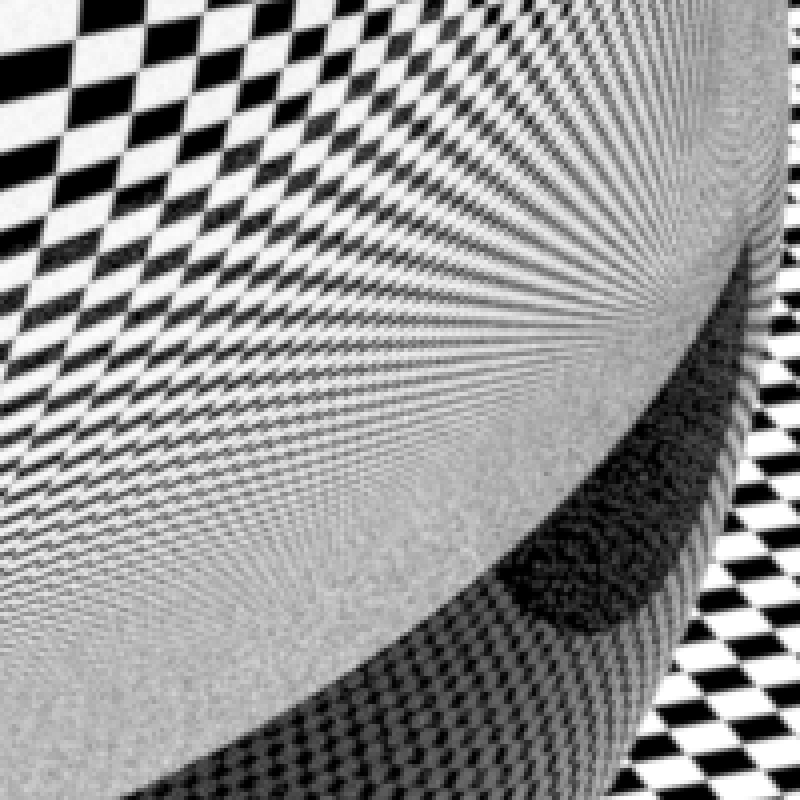}\\%
        \rotatebox{90}{\parbox{1\unit}{\centering (a) Sobol}}&%
        \includegraphics[width=1\unit]{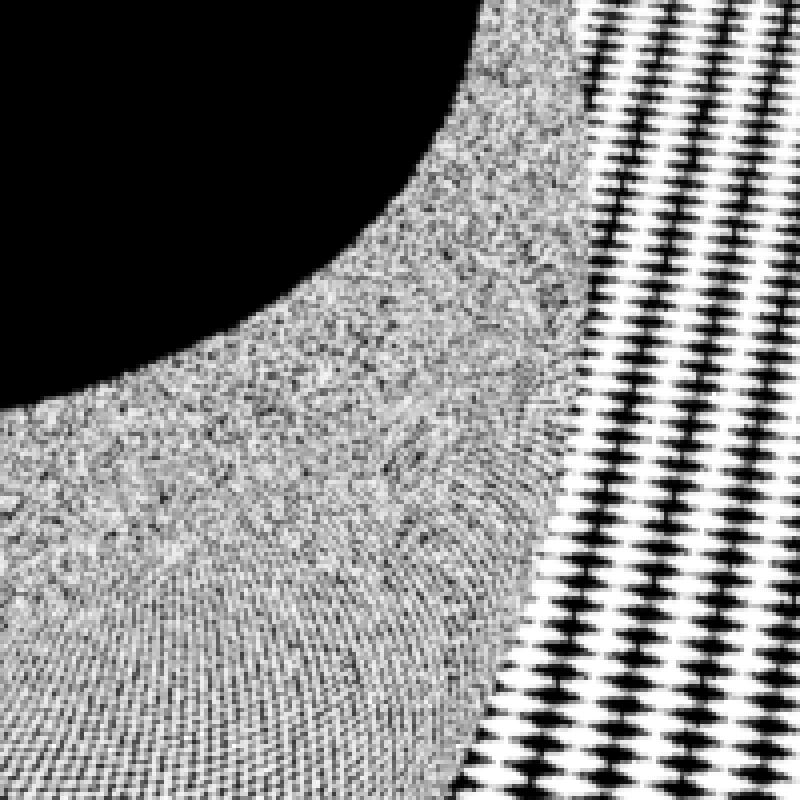}&%
        \includegraphics[width=1\unit]{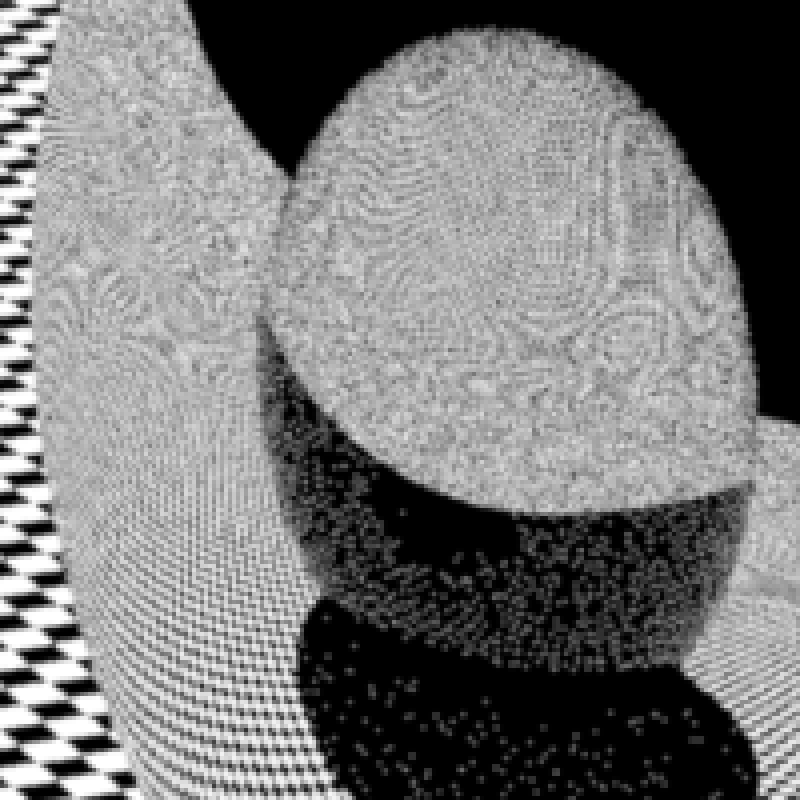}&%
        \includegraphics[width=1\unit]{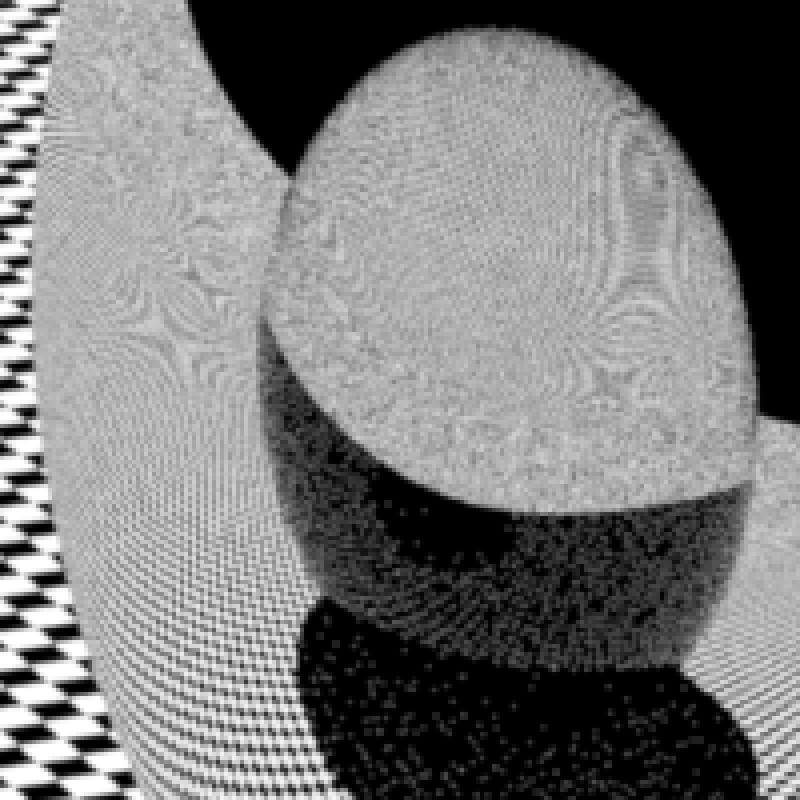}&%
        \includegraphics[width=1\unit]{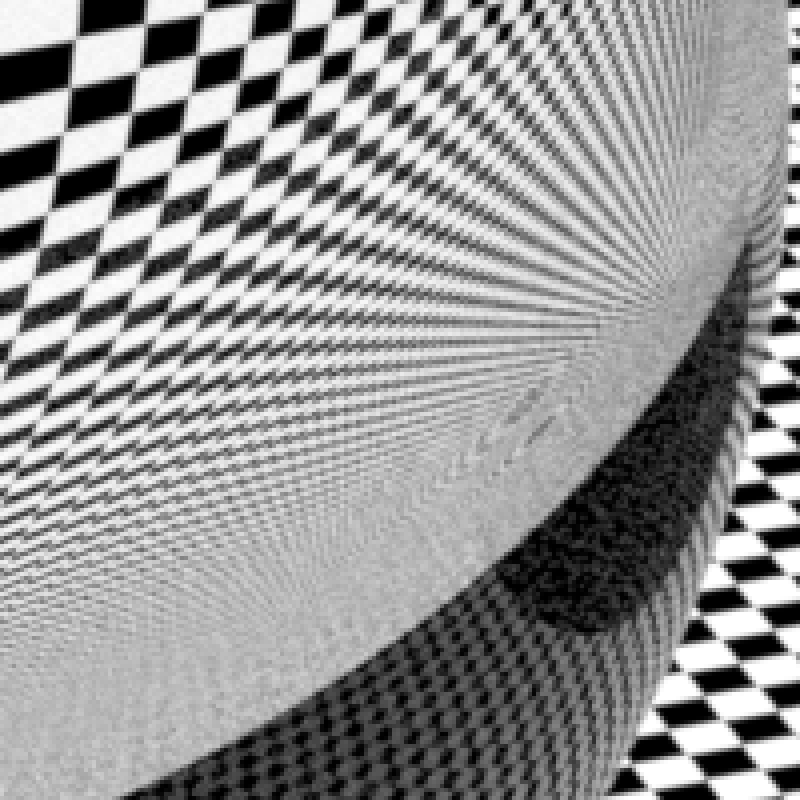}\\%
        & 2 SPP & 8 SPP & 16 SPP & 32 SPP
    \end{tabular*}}
    \caption{%
        \label{fig:sobolxi}
        \bluevar{Physically-based renderings} using various digital dyadic sequences: (a) The Global Sobol sampler based on the Sobol sequence; 
        (b) our adapted global sampler using a random $\xi$-sequence; and
        (c) low-discrepancy $(0, 2)$-sequence sampler, which uses independently shuffled Sobol 2D samples for each pixel and pair of dimensions. We placed our sampler in the middle for easier comparison. 
        \bluevar{The figure shows that the $\xi$-based Global sampler is less aliased than Sobol and less noisy than the pixel-based sampler.}
        The insets show selected parts of the scene at different sampling rates.  
        }
    \vspace{-2mm}
\end{figure}
\begin{figure*}
    \centering
    \settoheight{\labelHeight}{$\xi_0$}
  \setlength{\unit}{(\textwidth - \labelHeight - 6\gap)/6}
  {\scriptsize
    \begin{tabular*}{1\textwidth}{@{}c@{\extracolsep{\fill}}c@{\extracolsep{\fill}}c@{\extracolsep{\fill}}c@{\extracolsep{\fill}}c@{\extracolsep{\fill}}c@{\extracolsep{\fill}}c@{}}
        \rotatebox{90}{\parbox{1\unit}{\centering Sobol}}&%
        \includegraphics[width=1\unit]{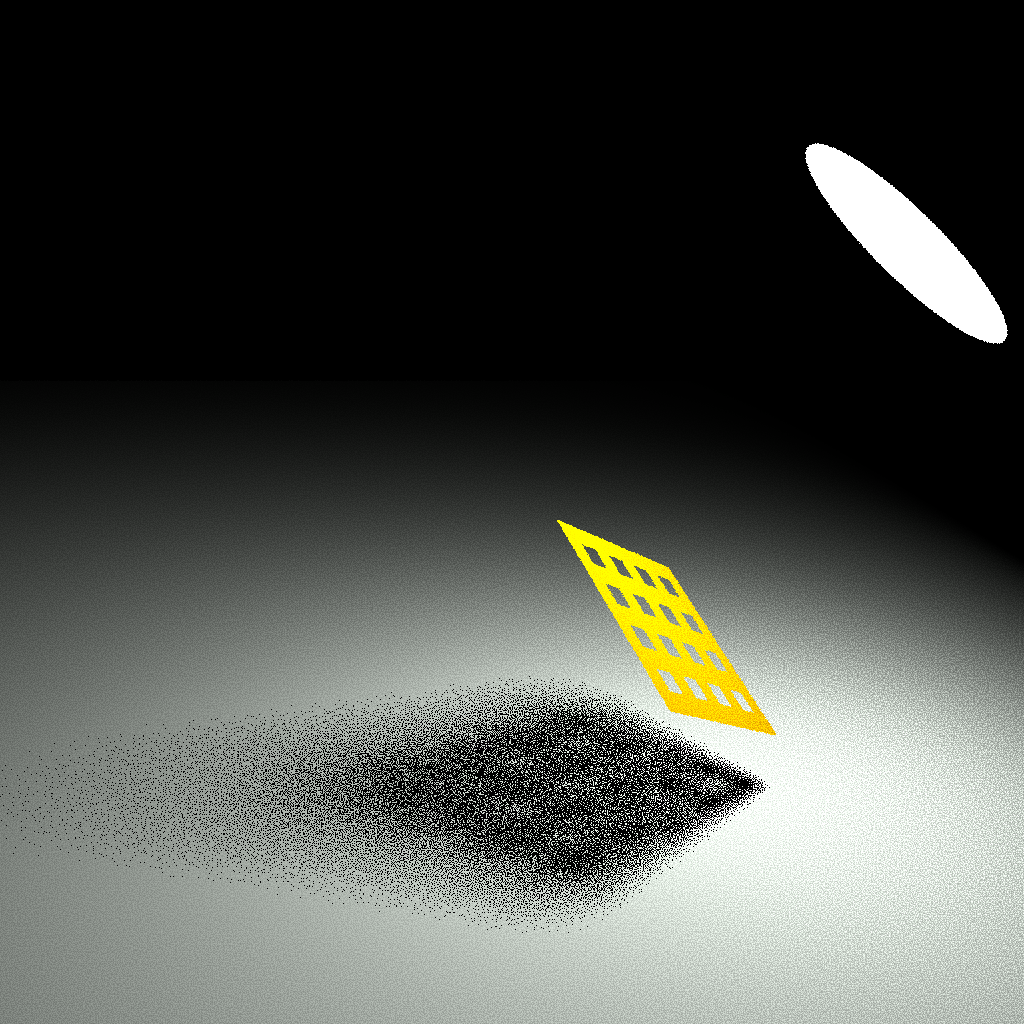}&%
        \includegraphics[width=1\unit]{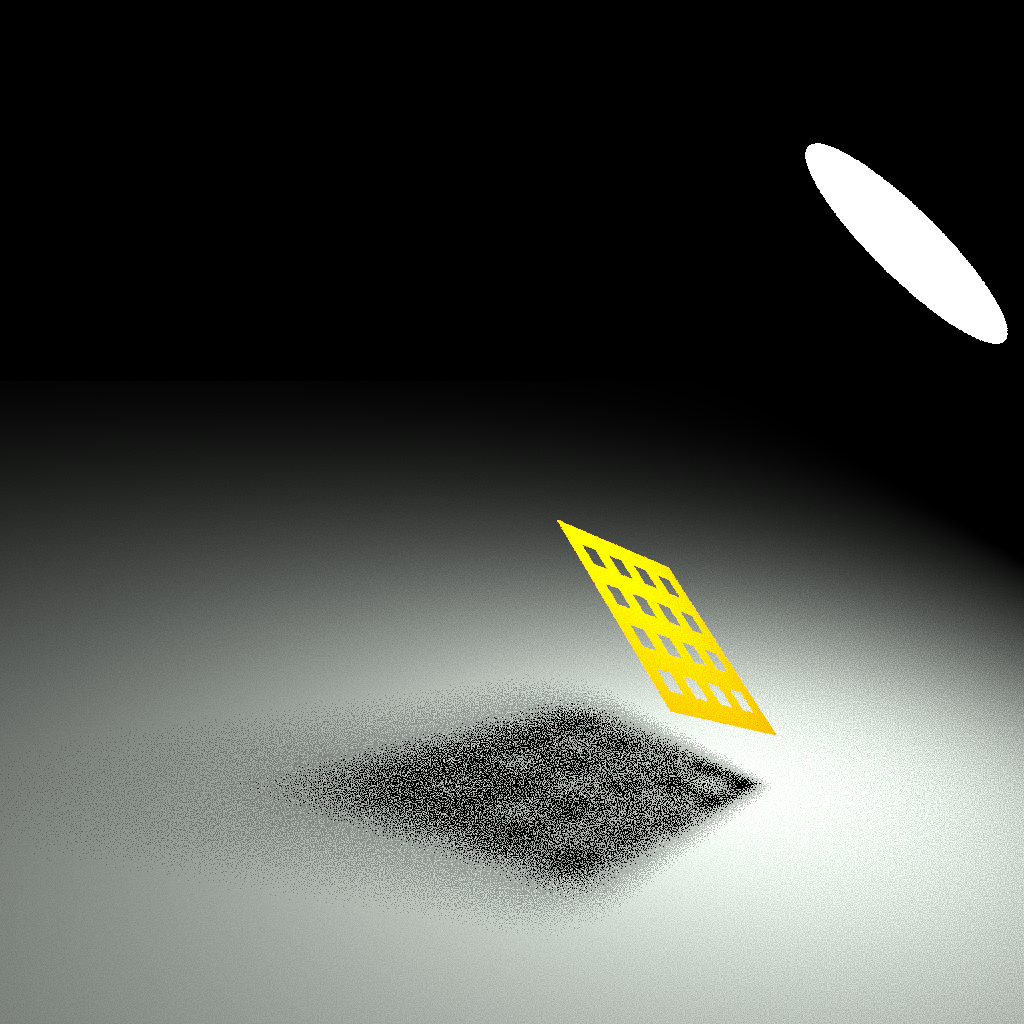}&%
        \includegraphics[width=1\unit]{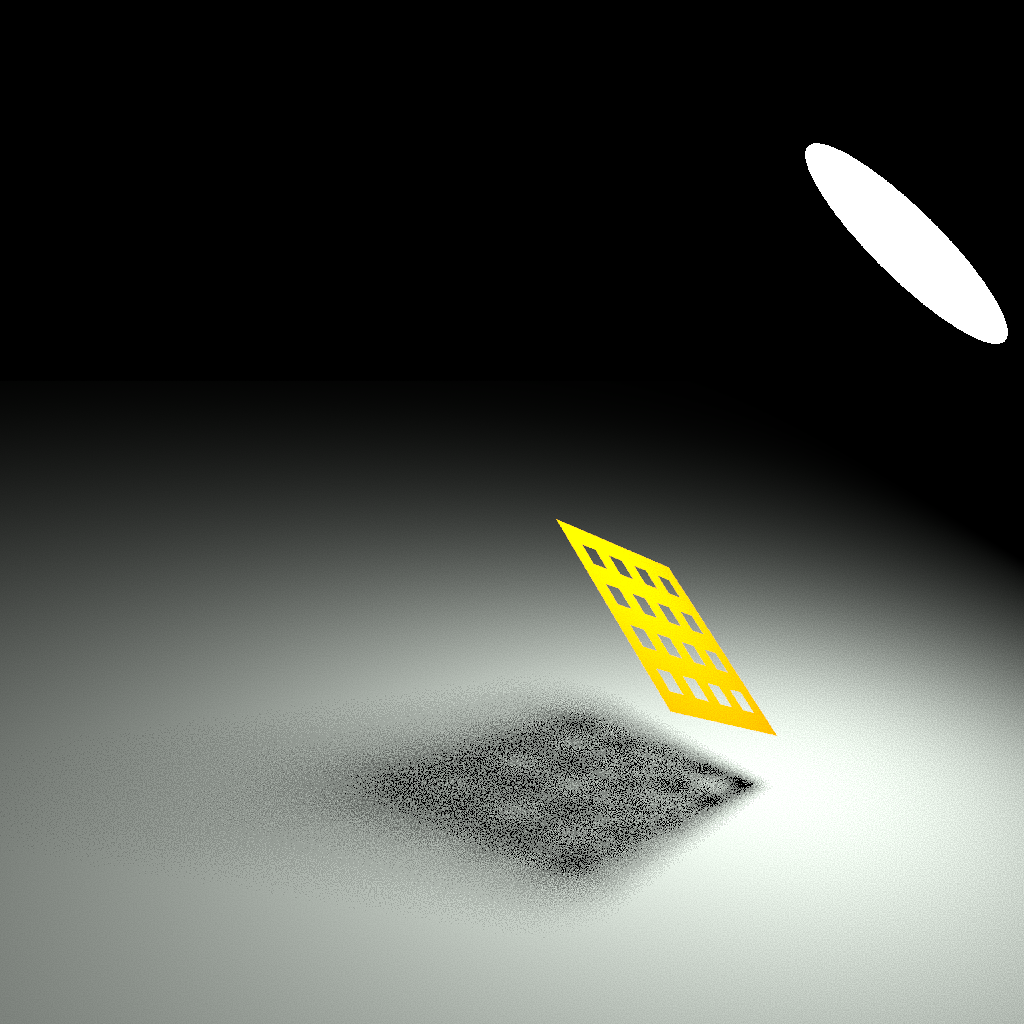}&%
        \includegraphics[width=1\unit]{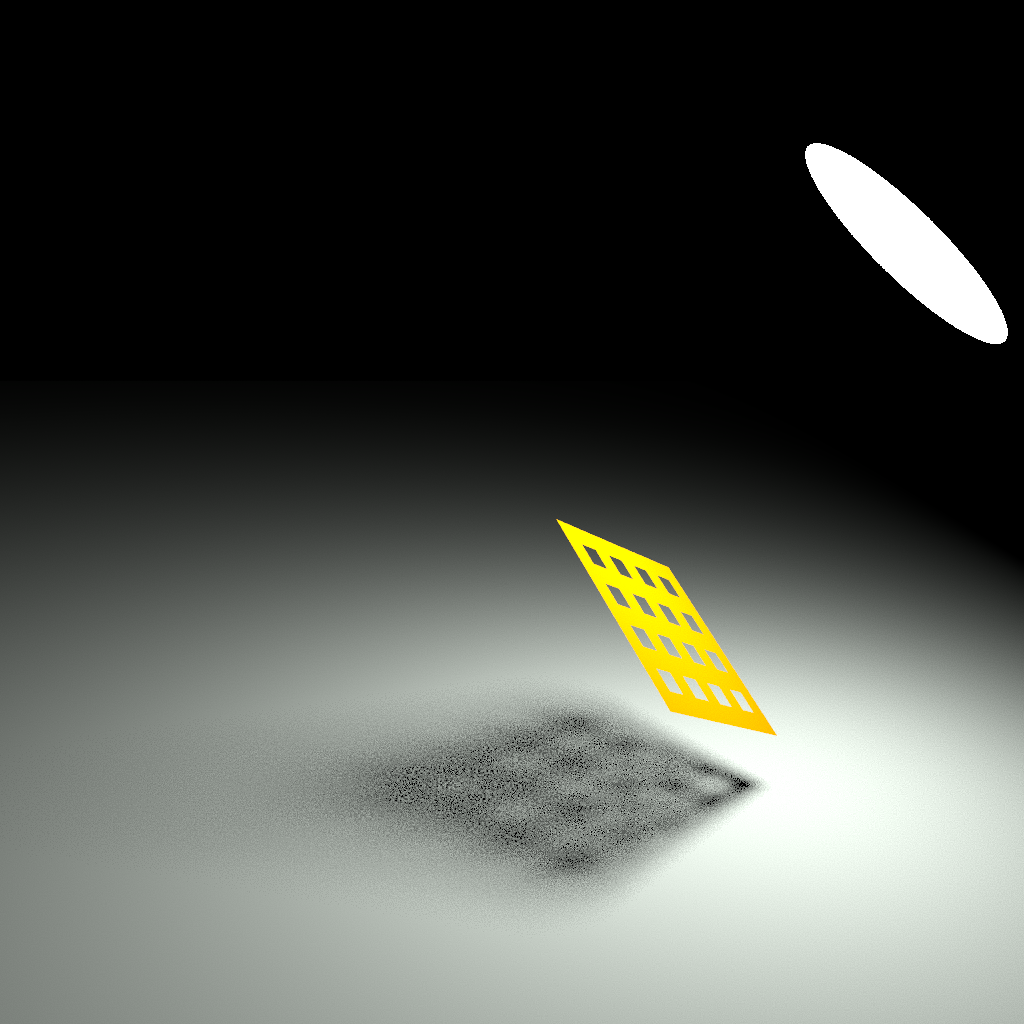}&%
        \includegraphics[width=1\unit]{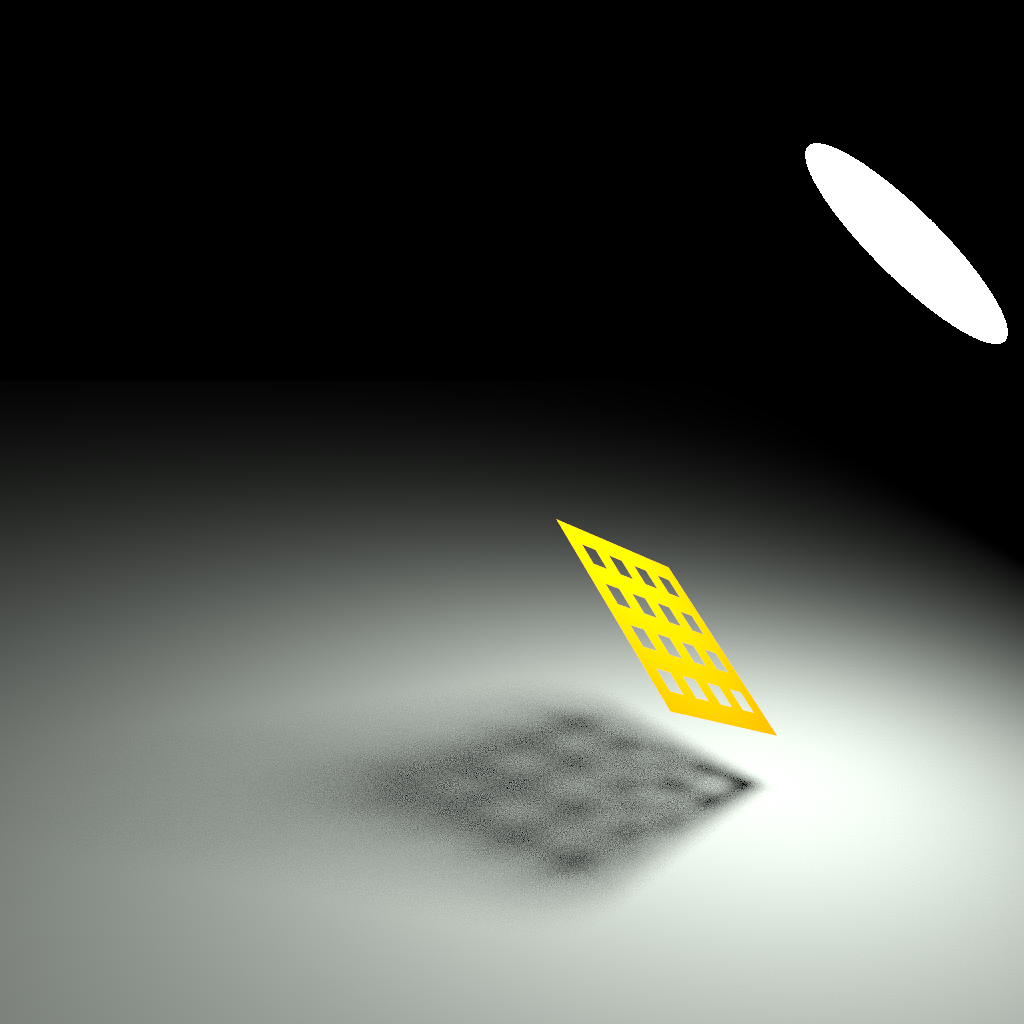}&%
        \includegraphics[width=1\unit]{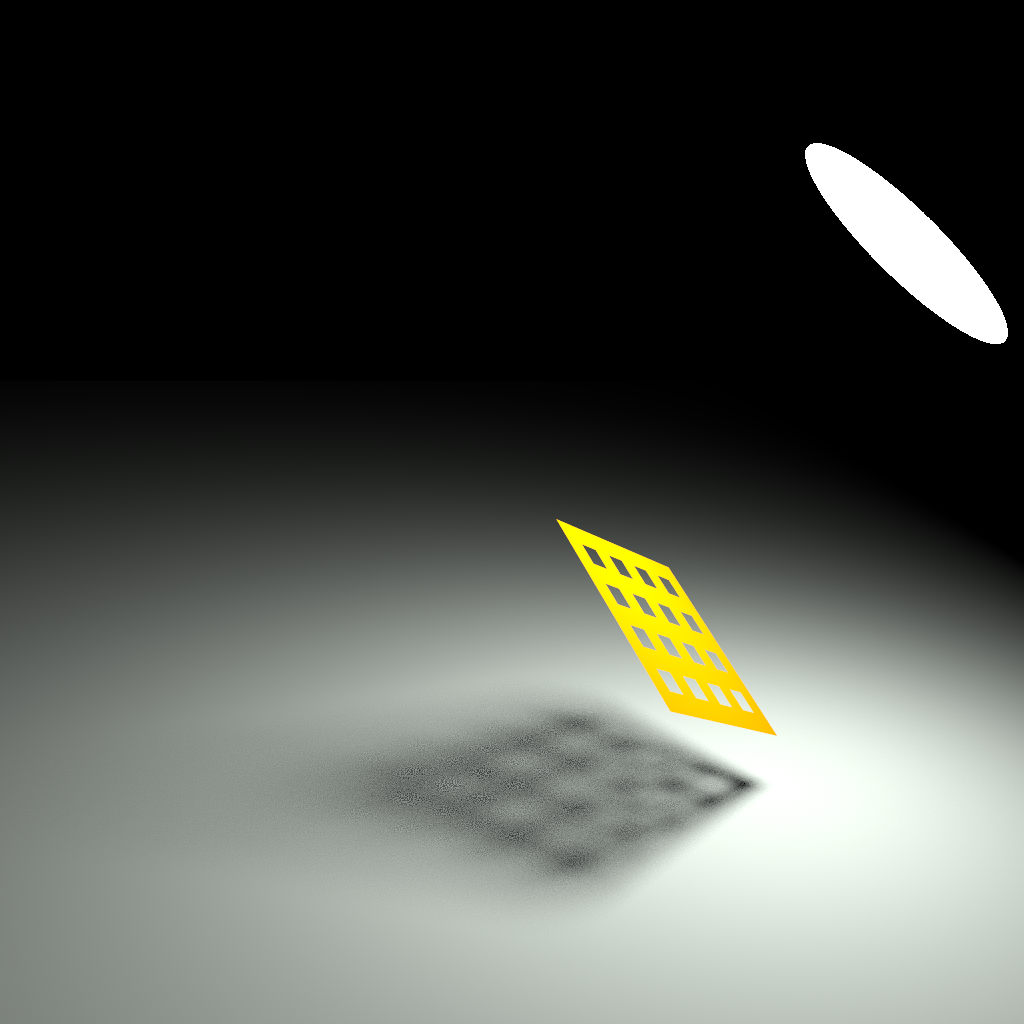}\\%
        \rotatebox{90}{\parbox{1\unit}{\centering $\xi_0$}}&%
        \includegraphics[width=1\unit]{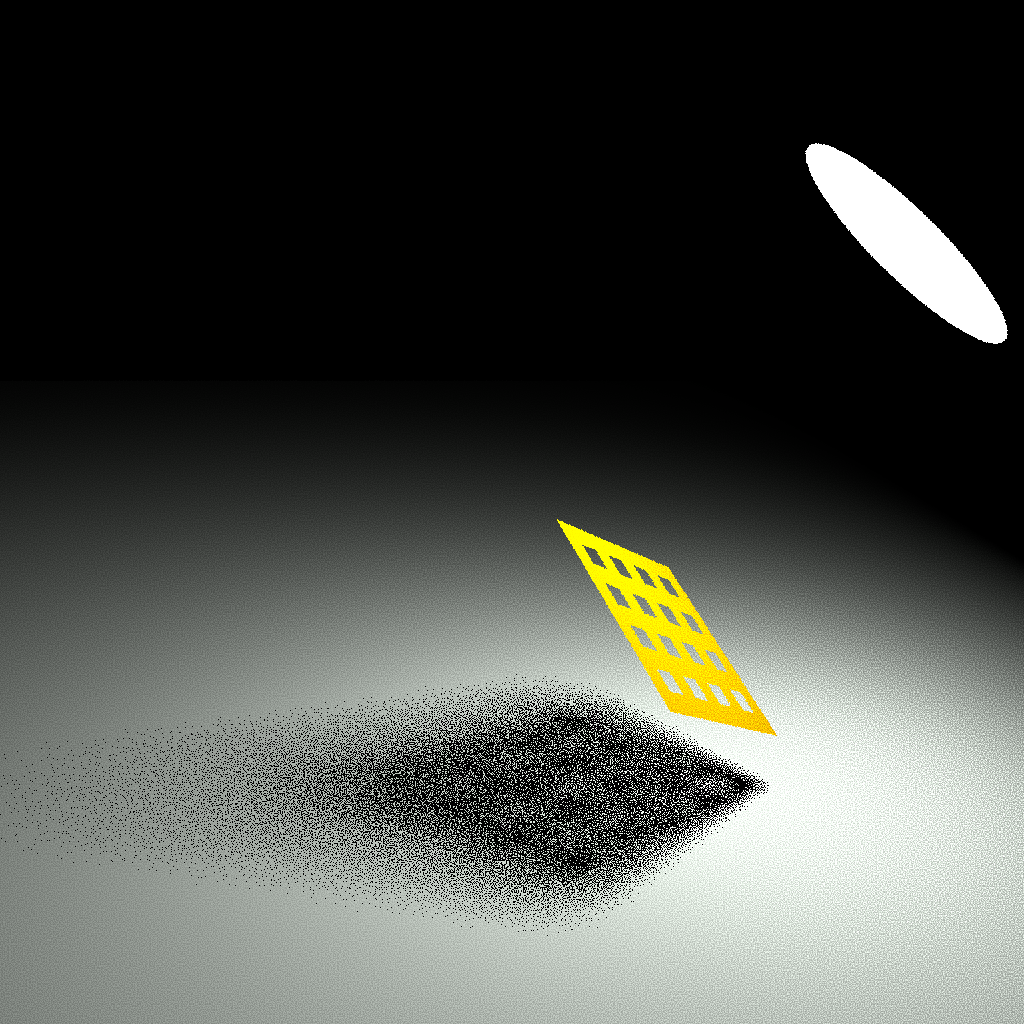}&%
        \includegraphics[width=1\unit]{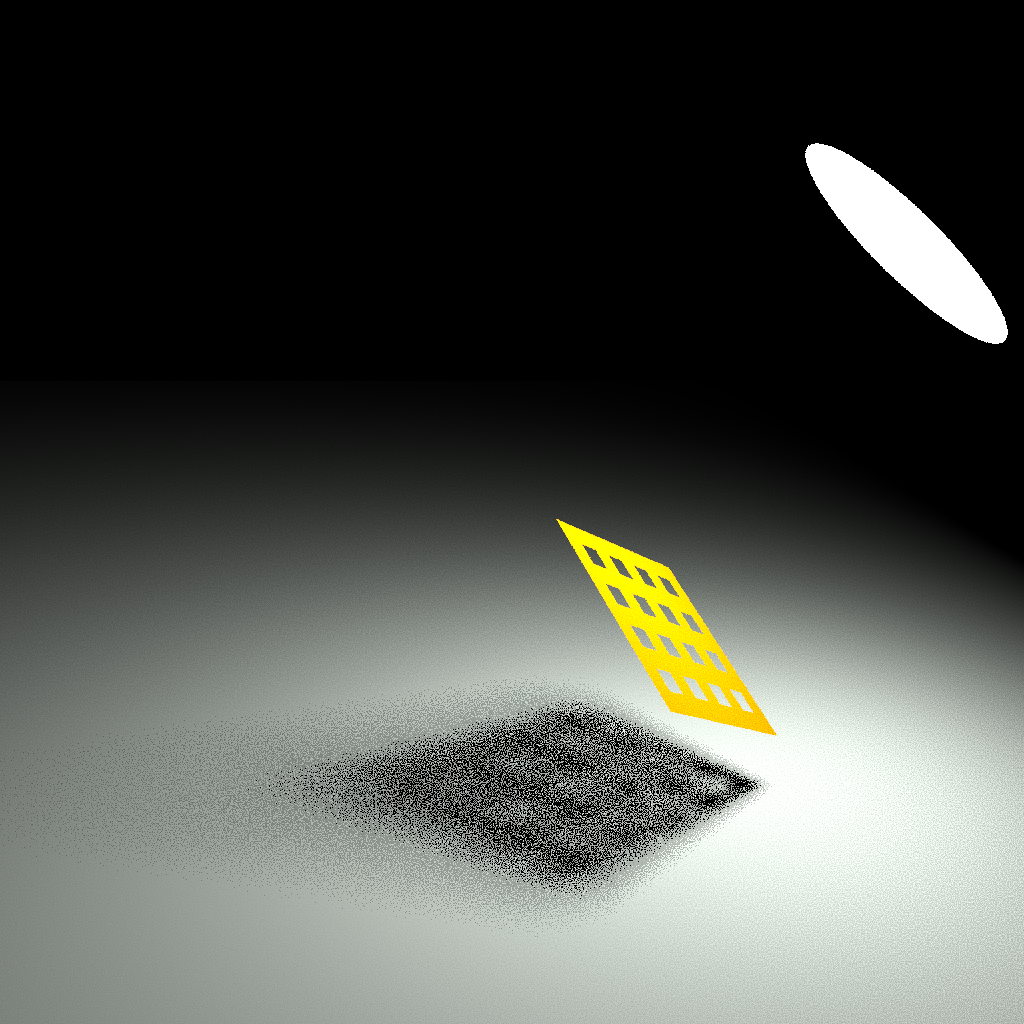}&%
        \includegraphics[width=1\unit]{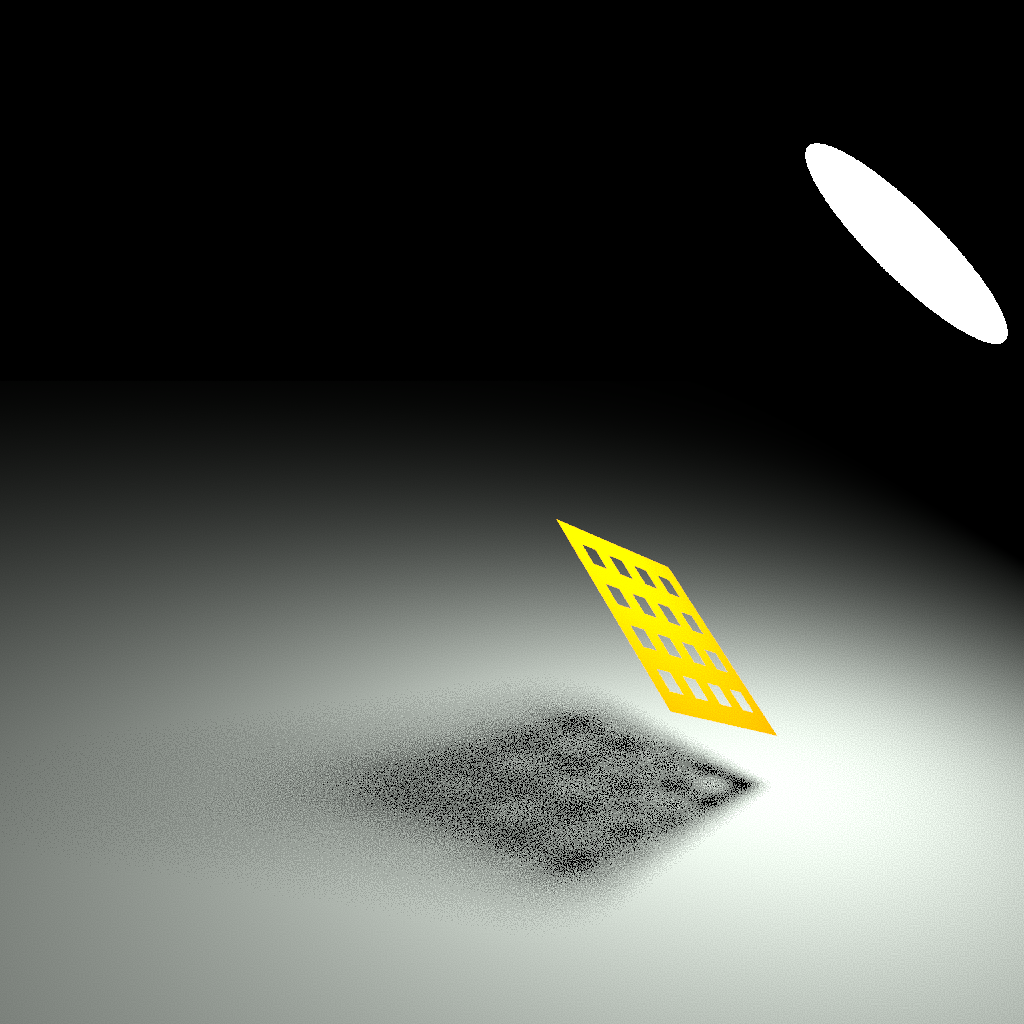}&%
        \includegraphics[width=1\unit]{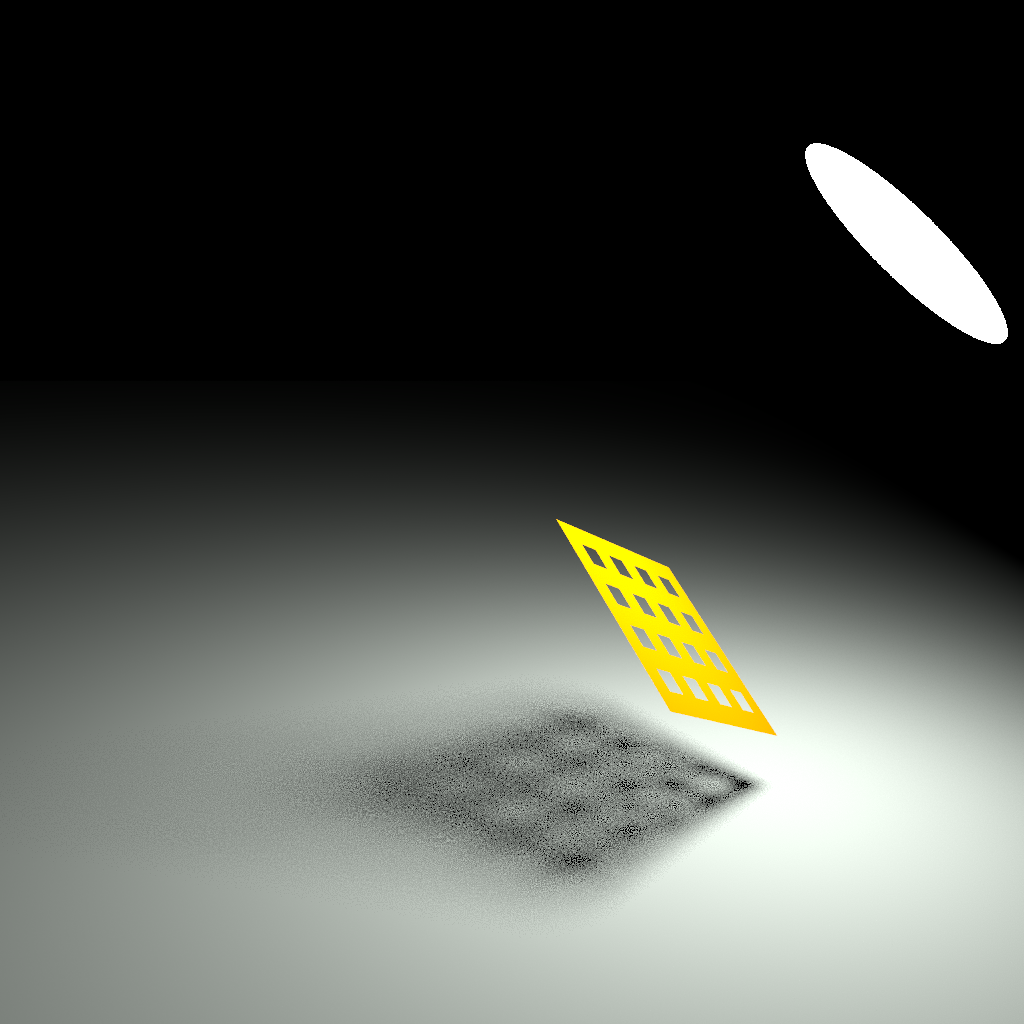}&%
        \includegraphics[width=1\unit]{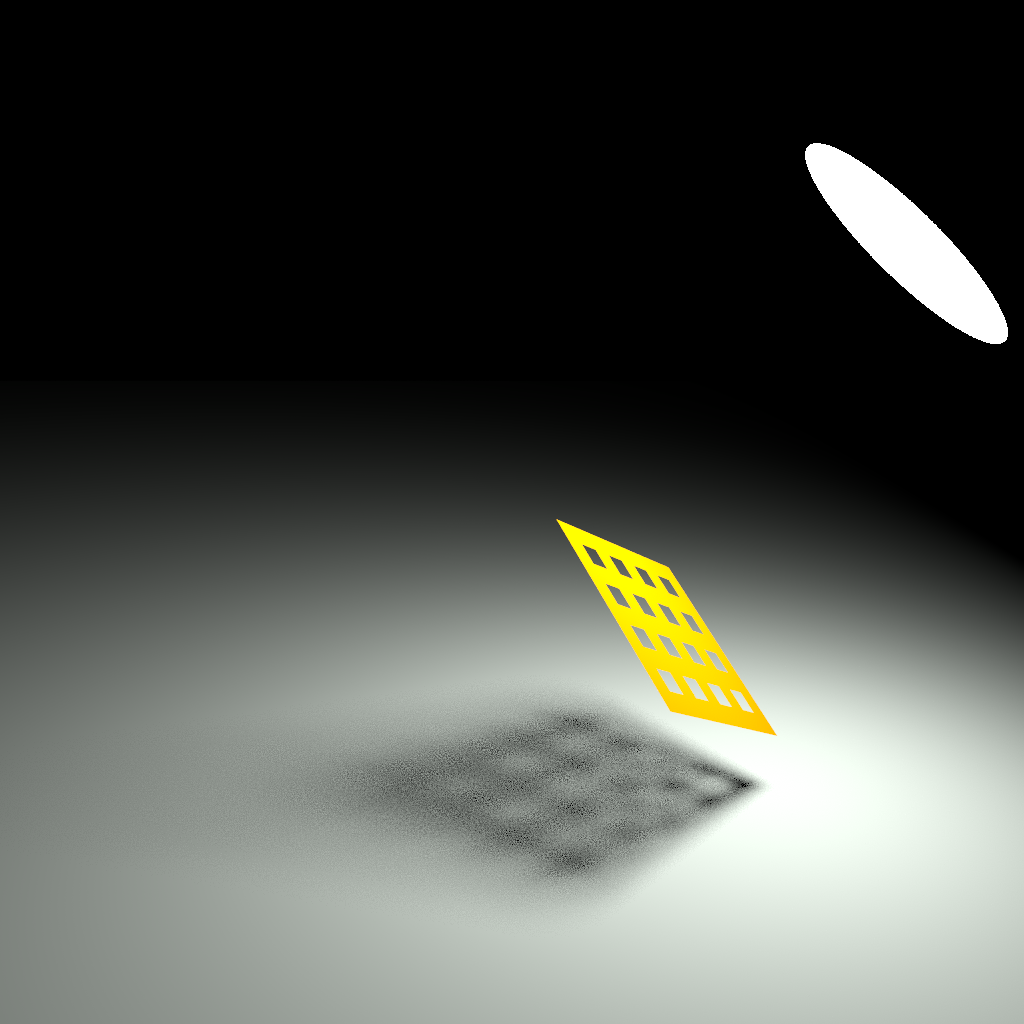}&%
        \includegraphics[width=1\unit]{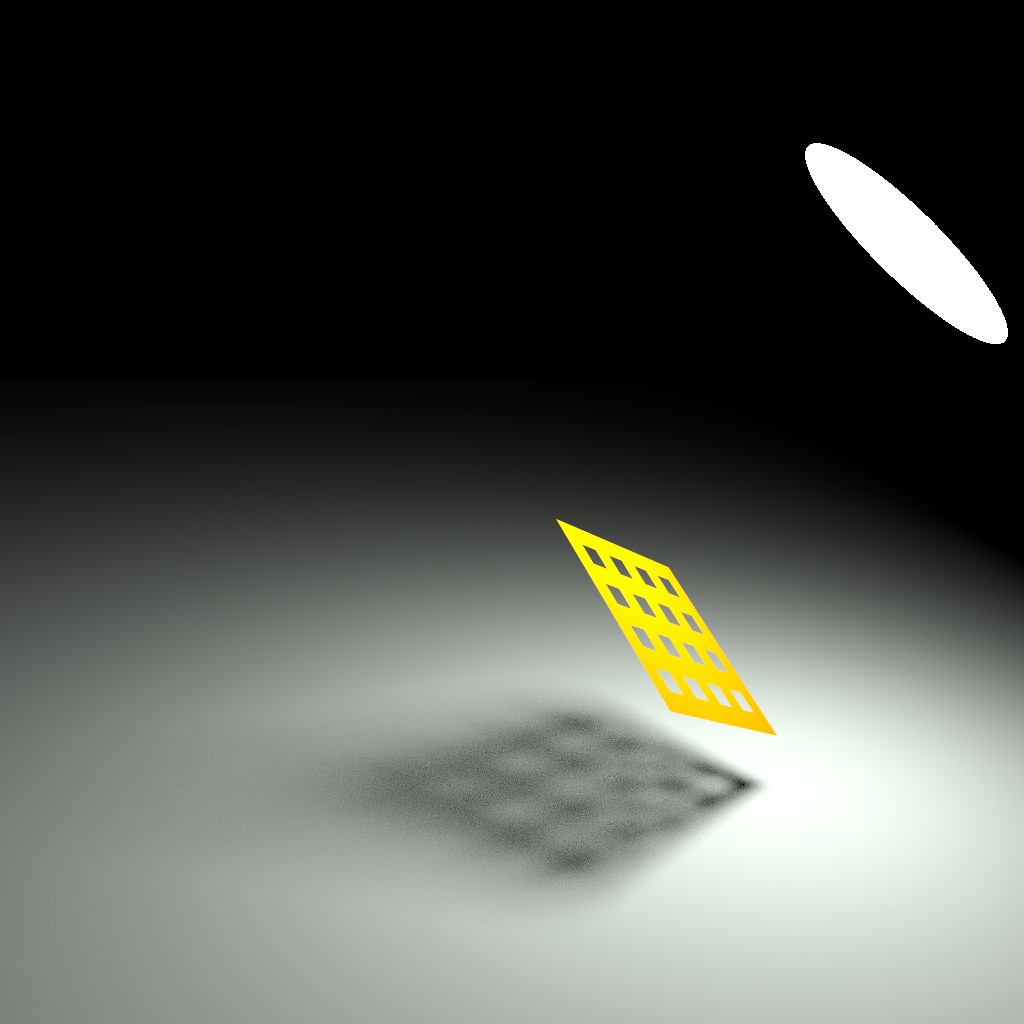}\\%
        & 1 SPP & 2 SPP & 4 SPP & 8 SPP & 16 SPP & 32 SPP
    \end{tabular*}}
    \caption{%
        \label{fig:rendering}
        Renderings obtained with the Z-sampler from Ahmed and Wonka~\shortcite{Ahmed20Screen} using Sobol and $\xi_0$ sequences for the generation of the samples at difference sampling rates (samples per pixel).
        \bluevar{The figure shows that the $\xi$-based Z-Sampler is less noisy than the Sobol-based one for the shown scene.}
    }
    \vspace{-2mm}
\end{figure*}

\subsection{Morton Ordering\label{sec:morton ordering}}

Instead of generating the $x$ and $y$ coordinates separately, they can be obtained jointly as a Morton-ordering \cite{Morton66Computer} index
\begin{equation}
    \mathbf{z} = 0.y_{0}x_{0}y_{1}x_{1}\cdots\,, \label{eq:yx-bits}
\end{equation}
and then we can split (\emph{unshuffle} \cite{Warren2012Hacker}) the bits later.
This can be done easily by
\begin{align}
    \mathbf{z}^{(i+1)} & \leftarrow \mathbf{z}^{(i)} + p_z[s^{(i)} \bmod 4] \cdot 2^{-2i}\,, 
    & \mathbf{s}^{(i+1)} & \leftarrow \lfloor \mathbf{s}^{(i)}/4 \rfloor\,
    \label{eq:s2z} 
\end{align}
where $i=0,\dots,15$, $\mathbf{z}^{(0)}=0$,  $\mathbf{s}^{(0)}$ is the point index in the sequence, $p_z[q]$ is obtained by interleaving (\emph{shuffling} \cite{Warren2012Hacker}) the bits of 
$y$- and $x$-coordinates of the $q$-th point $p[q]$ for $q=0,\dots,3$, and addition is carryless. This may not be very useful in generating the samples, because, in addition to the cost of unshuffle, it requires 64-bit processing to obtain the full 32-bit resolution of each axis.
Mapping to Morton indices, however, is useful in the reverse direction, where we are given the 
coordinates of a stratum, and asked to retrieve the exact location of the (first) sample in that stratum, and/or its sequence number, as encountered in stippling and importance sampling.
Note that the coordinates of the stratum prefix the coordinates of points inside it.
Thus, obtaining a point sequence number, given the stratum, is a partial inversion of the sequence, which we discuss in more detail next.


\subsection{Inverting the Sequence\label{sec:inversion}}

A straightforward approach to inverting a $\xi$-sequence, mapping a stratum index $\mathbf{z}$ into a sequence number $\mathbf{s}$, is to \emph{undo} the iterations
in Eq.~\eqref{eq:s2z}.
Let
\begin{align}
\mathbf{s} & = \cdots s_{54}s_{32}s_{10} \,, \\
\mathbf{z} & =0.z_{01}z_{23}\cdots \,, \\
p_z[q] & =0.p_{z01}[q]p_{z23}[q]\cdots \,
\end{align}
encode the sequence index of a point, its Morton index, and 
the first $4$ points
of the sequence written in base $4$.
Note the reversed indexing of the digits of the sequence number.
Given $\mathbf{z}^{(0)}=\mathbf{z}$, we may find $\mathbf{s}$ iteratively starting from $\mathbf{s}^{(0)}=0$ as follows:
\begin{align}
    q & \leftarrow q: z_{01}^{(i)} = p_{z01}[q]\,, \label{eq:z2s quadrant} \\
    \mathbf{s}^{(i+1)} & \leftarrow \mathbf{s}^{(i)} + q \cdot 2 ^ {2i}\,, \label{eq:z2s s} \\
    \mathbf{z}^{(i+1)} & \leftarrow (\mathbf{z}^{(i)} + p_z[q]) \cdot 4\,, \label{eq:z2s z}
\end{align}
in the $i$-th iteration.
The idea is that, if we examine Eq.~\eqref{eq:s2z}, we see that only one vector contributes to the most significant two bits of $\mathbf{z}$.
We can therefore deduce which of the four vectors was added in the first iteration of computing $\mathbf{z}$, and insert its index as a digit in the sequence number, as in Eq.~\eqref{eq:z2s s}.
We then subtract, or equivalently add, the originally added vector to $\mathbf{z}$.
Now, only one vector is responsible for the following pair of bits.
By shifting the bits up we can recursively apply the undo process.
Note that this process iteratively zeros $\mathbf{z}$ while building $\mathbf{s}$.


We implemented this and obtained the inversion rate of \textasciitilde 38M points per second (cf. Table~\ref{tab:speed}) that, while faster than any sampling technique we are aware of, is significantly slower than the forward production rate.
A much faster approach takes advantage of Morton ordering discussed in Section~\ref{sec:morton ordering}, noting that the inversion matrix bears a similar staggered pair-of-columns structure. This gives a speed up of up to 70\% over the forward rate, since only one value is evaluated rather than two. \notforarxiv{Implementation details may be found in the supplementary code and interactive demos.}


\subsection{Coding Complexity}

Again, this is an aspect where $\xi$-sequences \blueit{have an advantage}.
Not only the fact that their implementation follows straightforward steps, as listed in Algorithm~\ref{alg:setup},
but the intuition of these steps has a clear geometric interpretation, as we discussed earlier.
\notforarxiv{In the supplementary materials we provide a JavaScript-based interactive demonstration of these sequences that we encourage the reader to experiment with.}


\subsection{Atlas}

The relatively small two-dimensional space of $\xi$-sequences, along with their easy implementation on GPUs, makes it possible to scan them exhaustively and build an atlas for exploring these patterns and finding ones with favorable properties. Fig.~\ref{fig:atlas} shows example maps for the star discrepancy, minimum distance between points, and average distance to the nearest neighbor.


\subsection{Rendering Examples}
\label{ssec:rendering}
We show the impact of our work in two rendering examples.

The first example (Fig.~\ref{fig:sobolxi}) demonstrates that self-similar sequences provide better anti-aliasing than the standard Sobol sequence. To make a meaningful comparison we integrated the sequences into the Global Sobol Sampler coming with PBRT. The $\xi$-sequences 
offer: (1) 3$\times$ speedup in sample generation and 5$\times$ in inversion, leading to up to 37\% observed reduction in rendering times; (2) no inversion lookup tables, cf.\cite{Gruenschlos12Enumerating}; (3) notable anti-aliasing with randomized $\xi$-sequences, and (4) they are free; emphasizing the value of our in-depth study.

Let us make a few comments on the implementation. The sampler uses a \emph{higher-dimensional} Sobol sequence, generated by certain matrices $(I,P,C_2,C_3,\dots)$; the first two are used to sample the image plane. We transform the matrices so that the first two become generating matrices $(C_{\xi(X)},C_{\xi^+(Y)})$ of the desired $\xi$-sequence. For that purpose, we decompose $C_{\xi(X)}=L_xU$ and $C_{\xi^+(Y)}=L_yPU$ as in Theorem~\ref{th-progressive-pairs}. Then we right-multiply all the original matrices $(I,P,C_2,C_3,\dots)$ by $U$ (which means just reordering of the sequence) and left-multiply the first two matrices by $L_x$ and $L_y$ respectively (which means scrambling). The resulting sequence is used for sampling.

The second example (Fig.~\ref{fig:rendering}) also shows slightly better anti-aliasing. Here we integrated the sequences into the Z-Sampler \cite{Ahmed20Screen} that uses the samples as is, without any scrambling.
\notforarxiv{Full-resolution images may be 
found in the supplementary materials.}

\begin{figure}
    \centering
    \includegraphics[width=1\columnwidth]{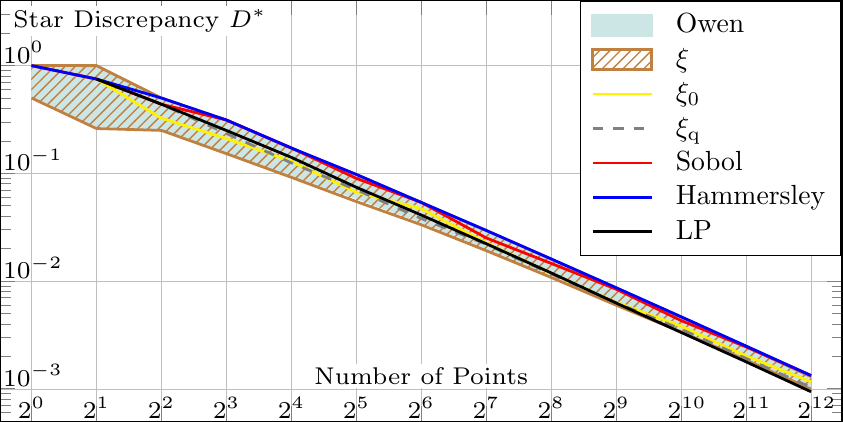}
    \caption{%
        \label{fig:discrepancy}
        Discrepancy comparison of XOR-scrambled $\xi$-sequences to unscrambled Sobol sequence.
        Hammersley and Larcher-Pillichshammer (LP) nets are shown for comparison.
        For Owen and $\xi$ we measured the discrepancy over a randomly drawn 64K sample. The yellow line refers to the $\xi_0$-sequence.
        The dashed line $\xi_{\text{q}}$ refers to the lowest measured discrepancy of the sampled sequences when the coordinates of the points are truncated to $\log_2(\text{Number of Points})$ bits.
    }
\end{figure}


\subsection{Discrepancy and Spectral Evaluation}

The new class of self-similar sequences introduced in the paper has several advantages, that we have evaluated in the preceding subsections, such as speed, memory usage, anti-aliasing, and versatility. 
%
Towards that, we show discrepancy and spectral comparison in Fig.~\ref{fig:discrepancy} and~\ref{fig:spectral}, both demonstrating that $\xi$-sequences are a good representative sample of the space of digital dyadic sequences.
Note that all digital 
dyadic sequences exhibit some structure when looking at a single instance and the perceived structure changes with $m$ (see Fig.~\ref{fig:teaser}).
\begin{figure*}
    \centering
  \setlength{\unit}{(\textwidth - 4\gap)/5}
  {\scriptsize
    \begin{tabular*}{1\textwidth}{@{}c@{\extracolsep{\fill}}c@{\extracolsep{\fill}}c@{\extracolsep{\fill}}c@{\extracolsep{\fill}}c@{\extracolsep{\fill}}c@{}}
        \includegraphics[width=1\unit]{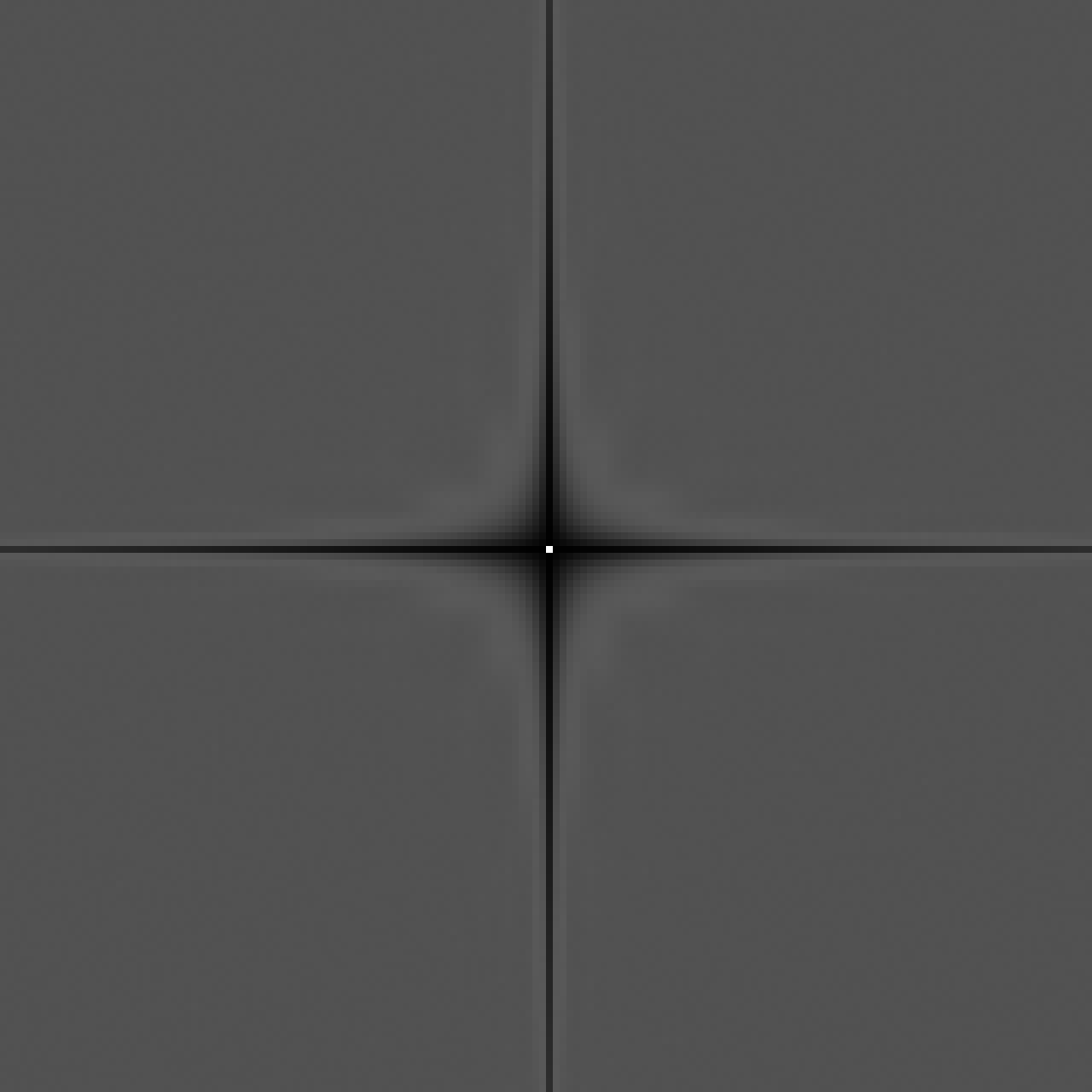}&%
        \includegraphics[width=1\unit]{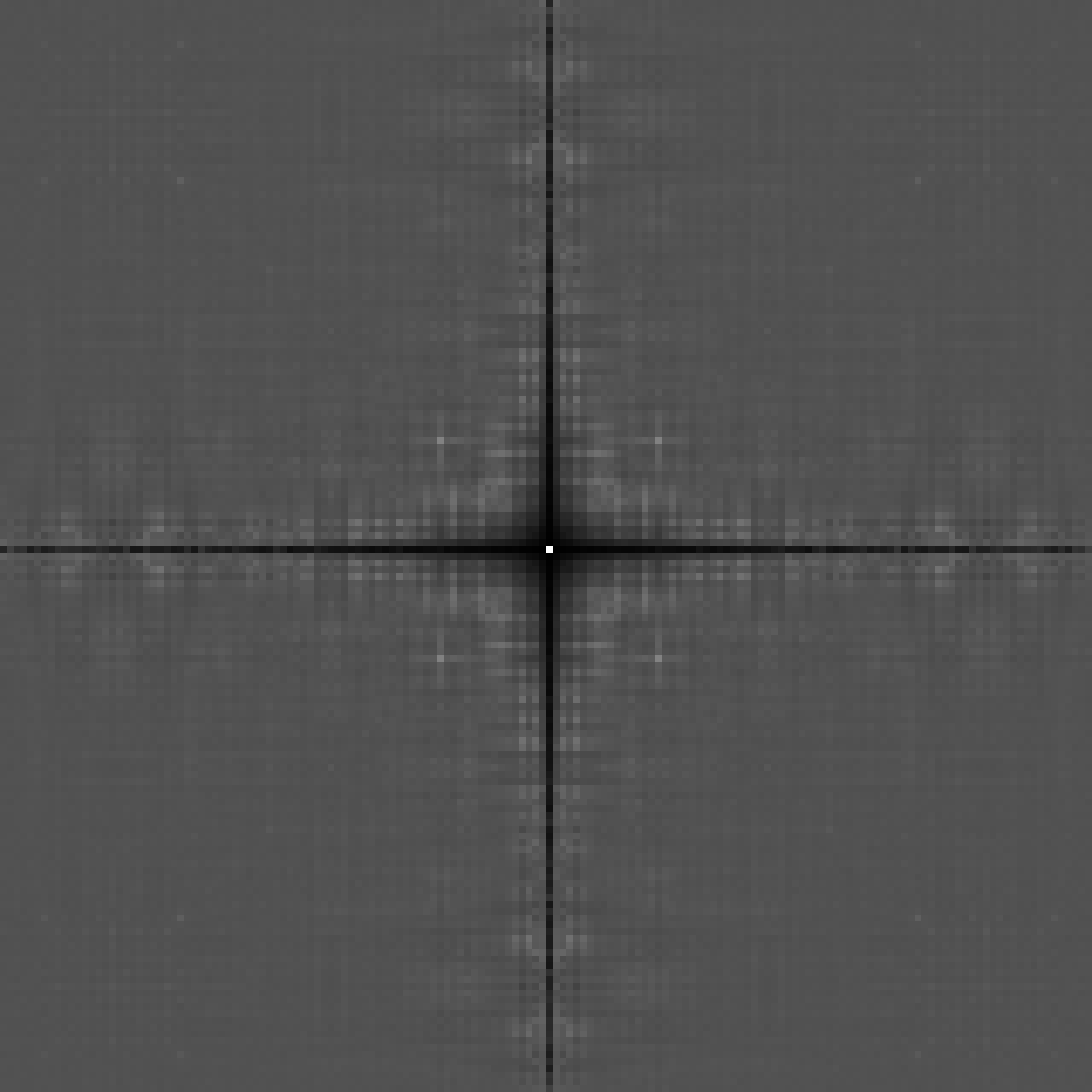}&%
        \includegraphics[width=1\unit]{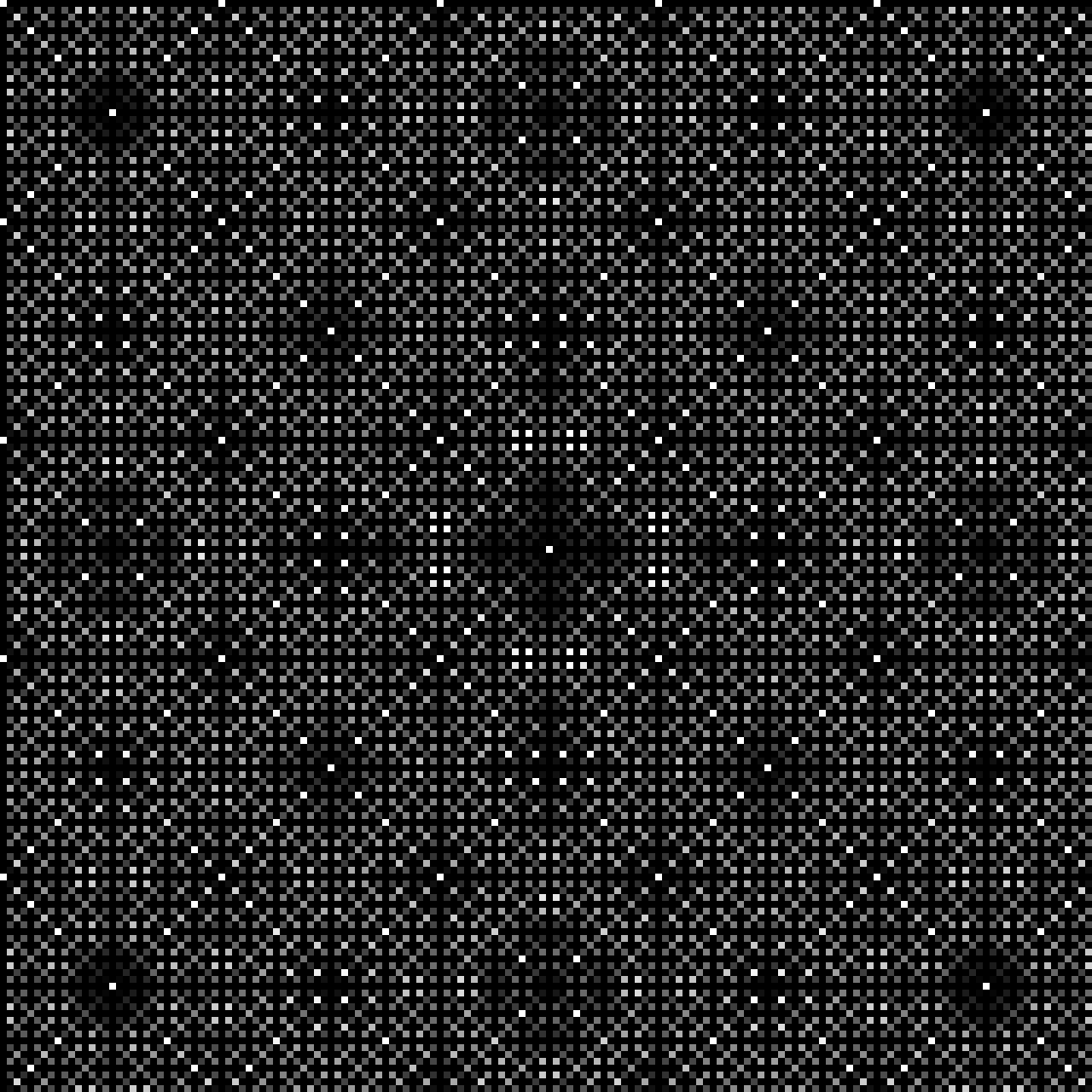}&%
        \includegraphics[width=1\unit]{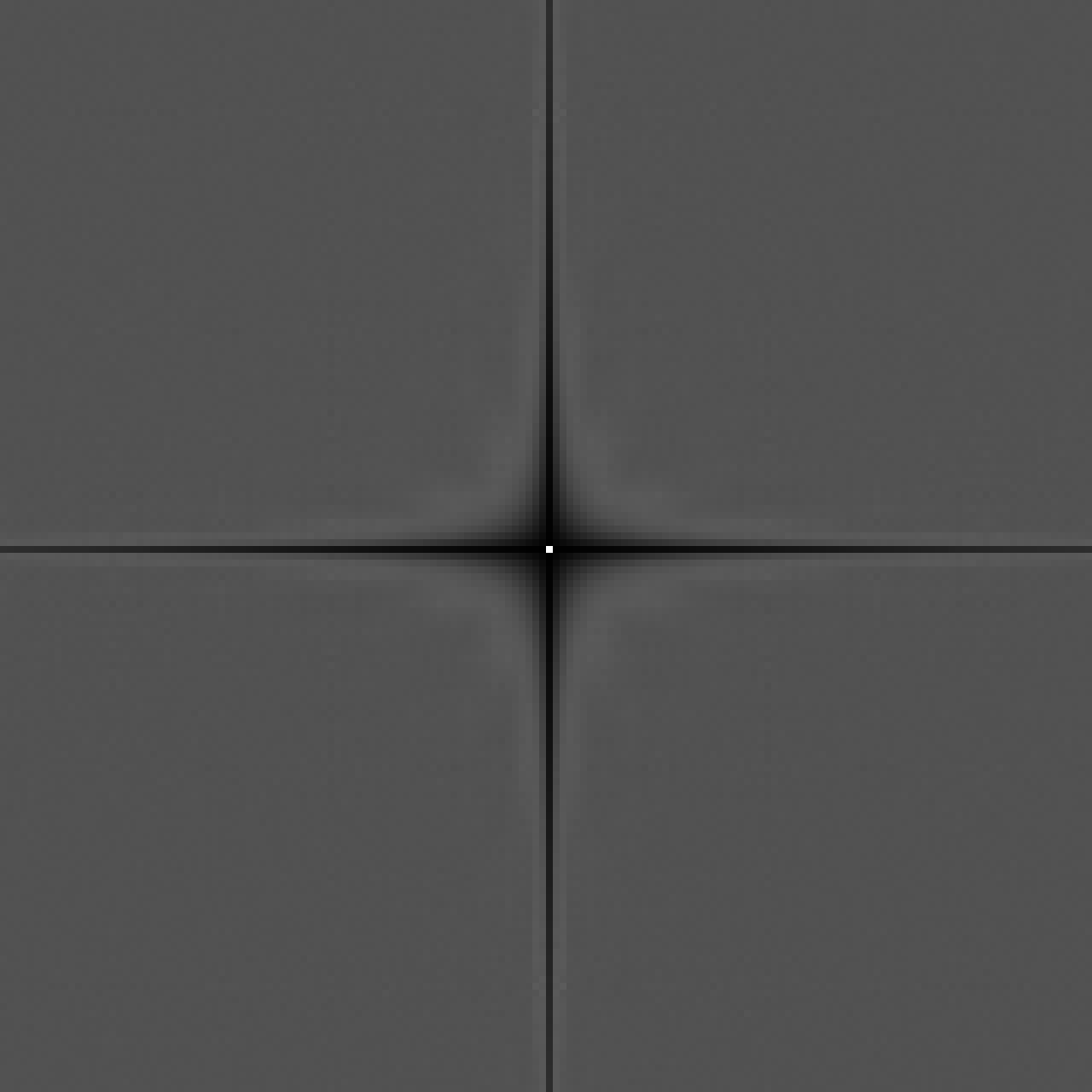}&%
        \includegraphics[width=1\unit]{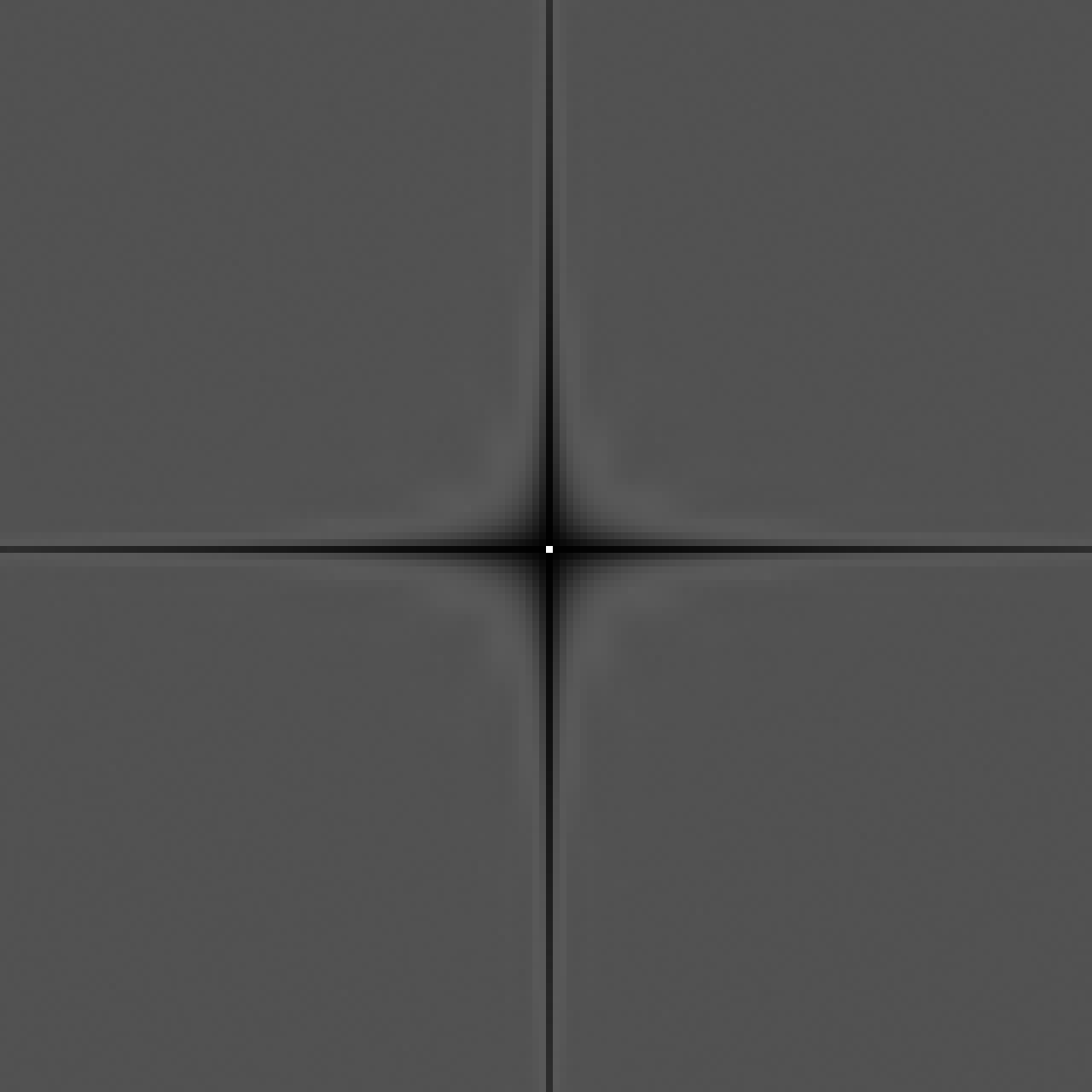}\\%
        \includegraphics[width=1\unit]{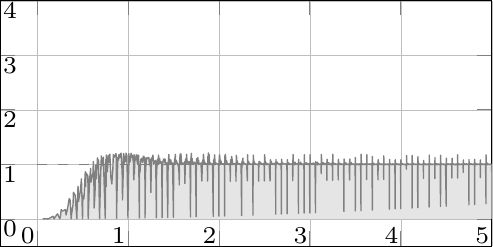}&%
        \includegraphics[width=1\unit]{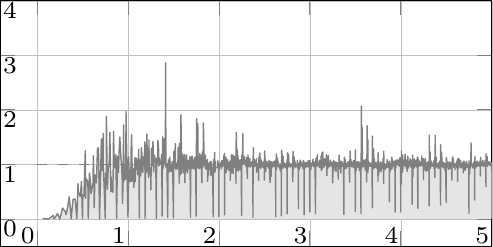}&%
        \includegraphics[width=1\unit]{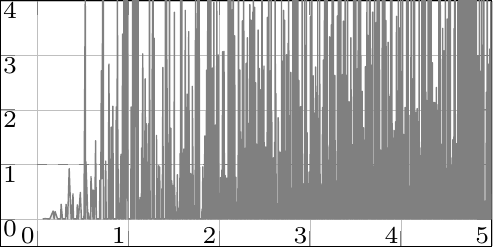}&%
        \includegraphics[width=1\unit]{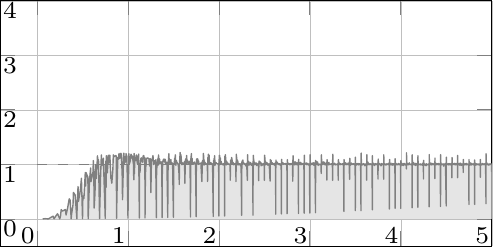}&%
        \includegraphics[width=1\unit]{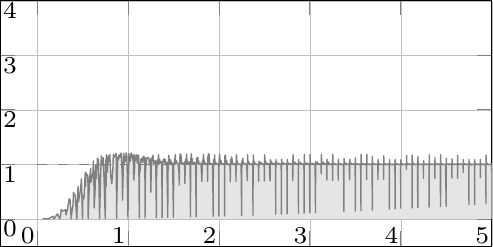}\\%
        \includegraphics[width=1\unit]{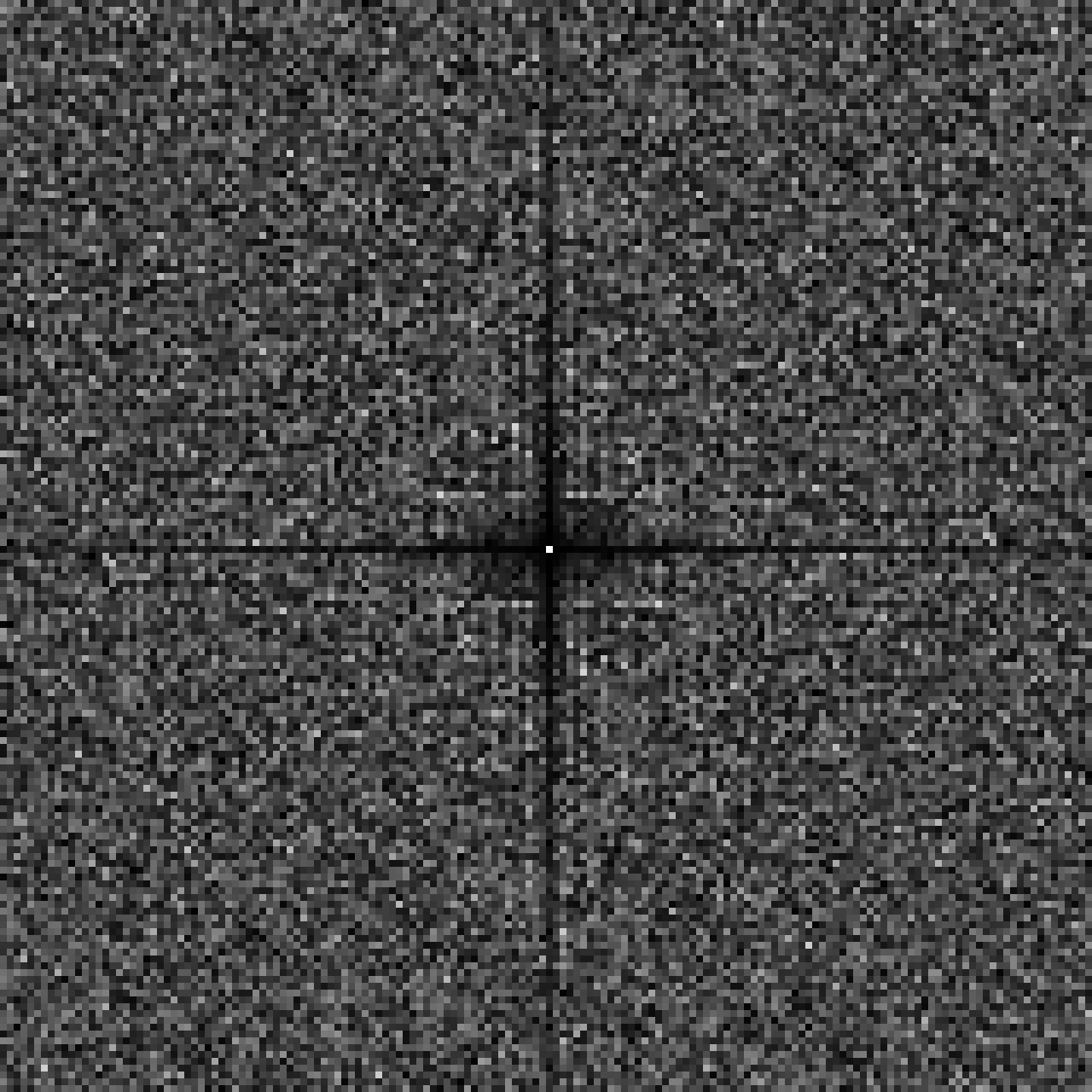}&%
        \includegraphics[width=1\unit]{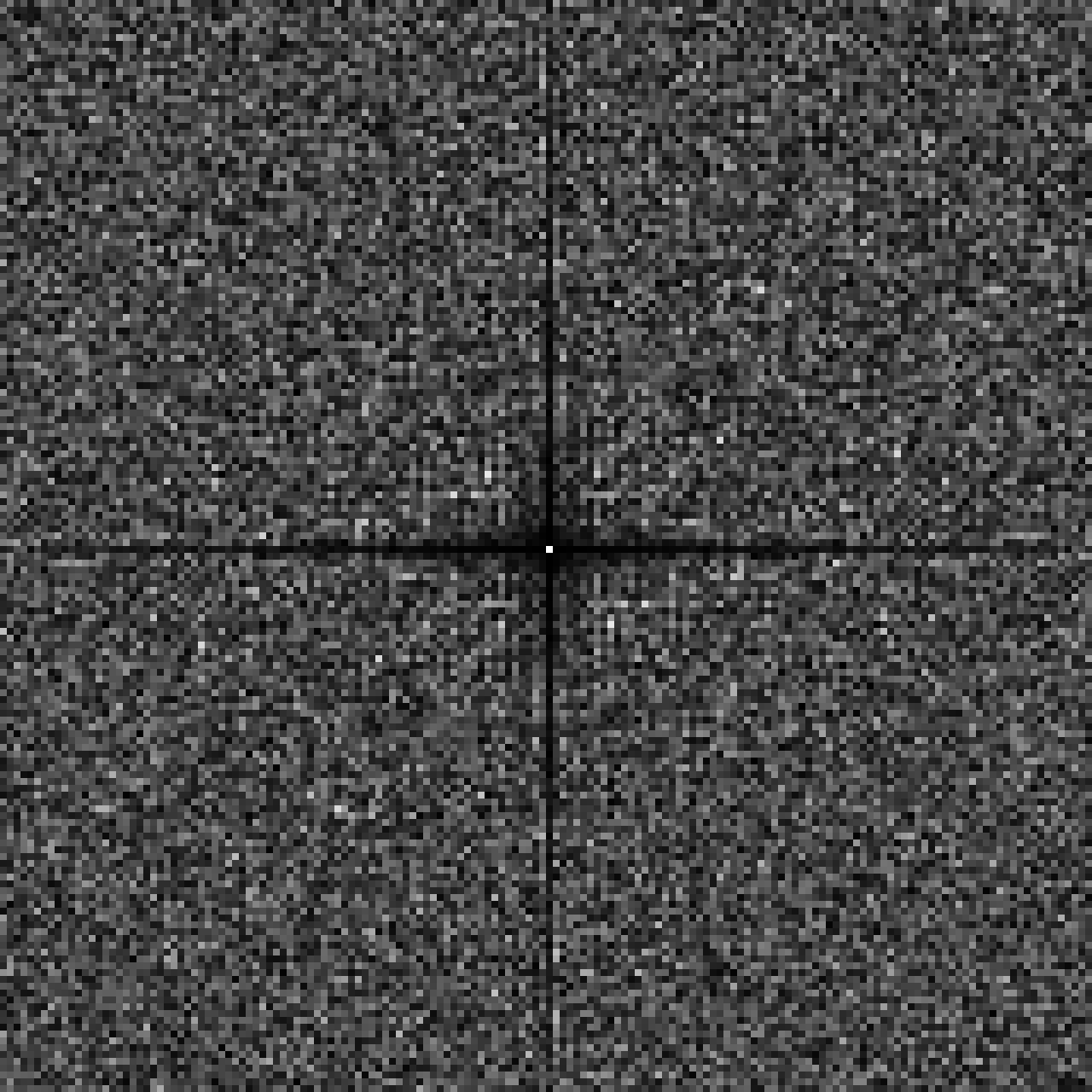}&%
        \includegraphics[width=1\unit]{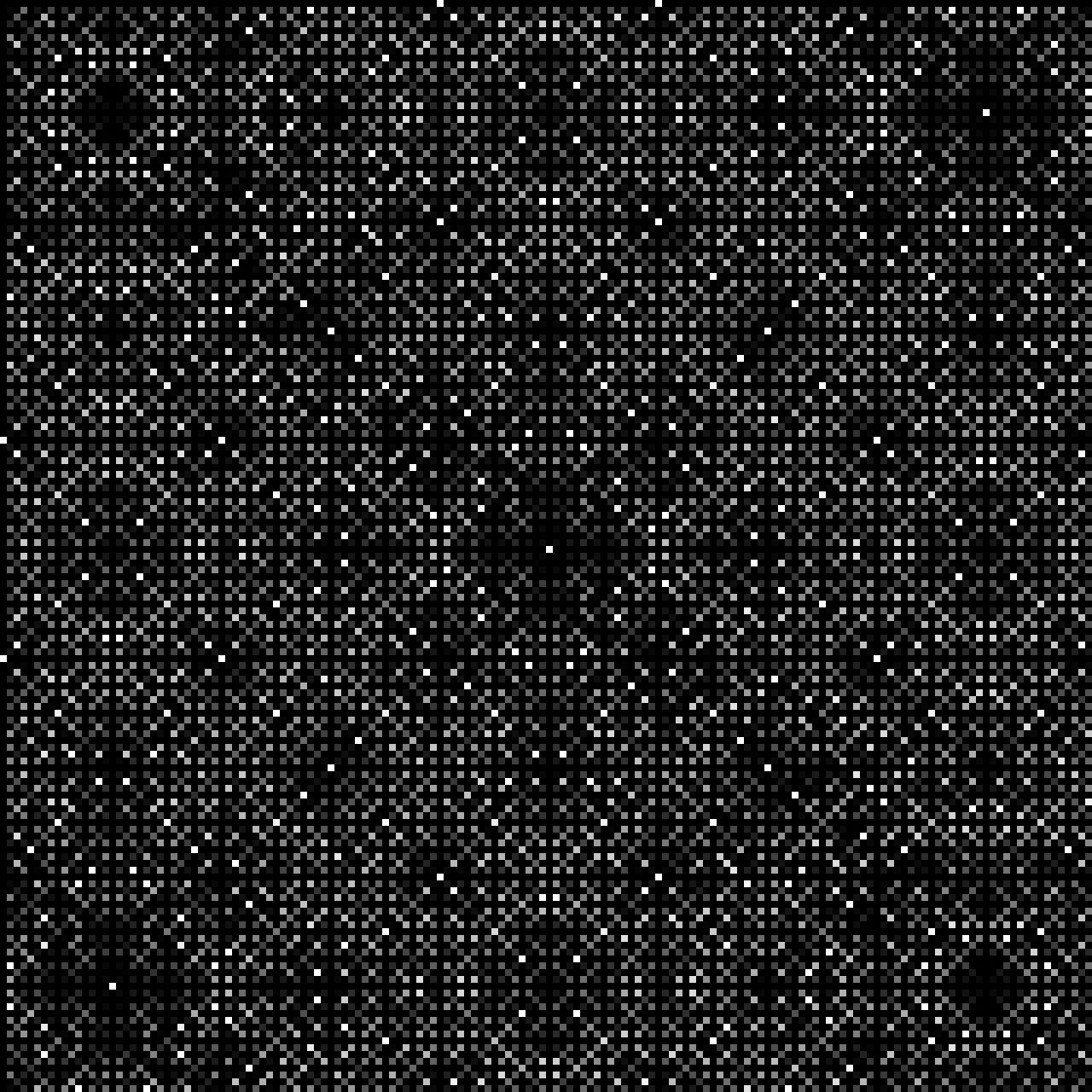}&%
        \includegraphics[width=1\unit]{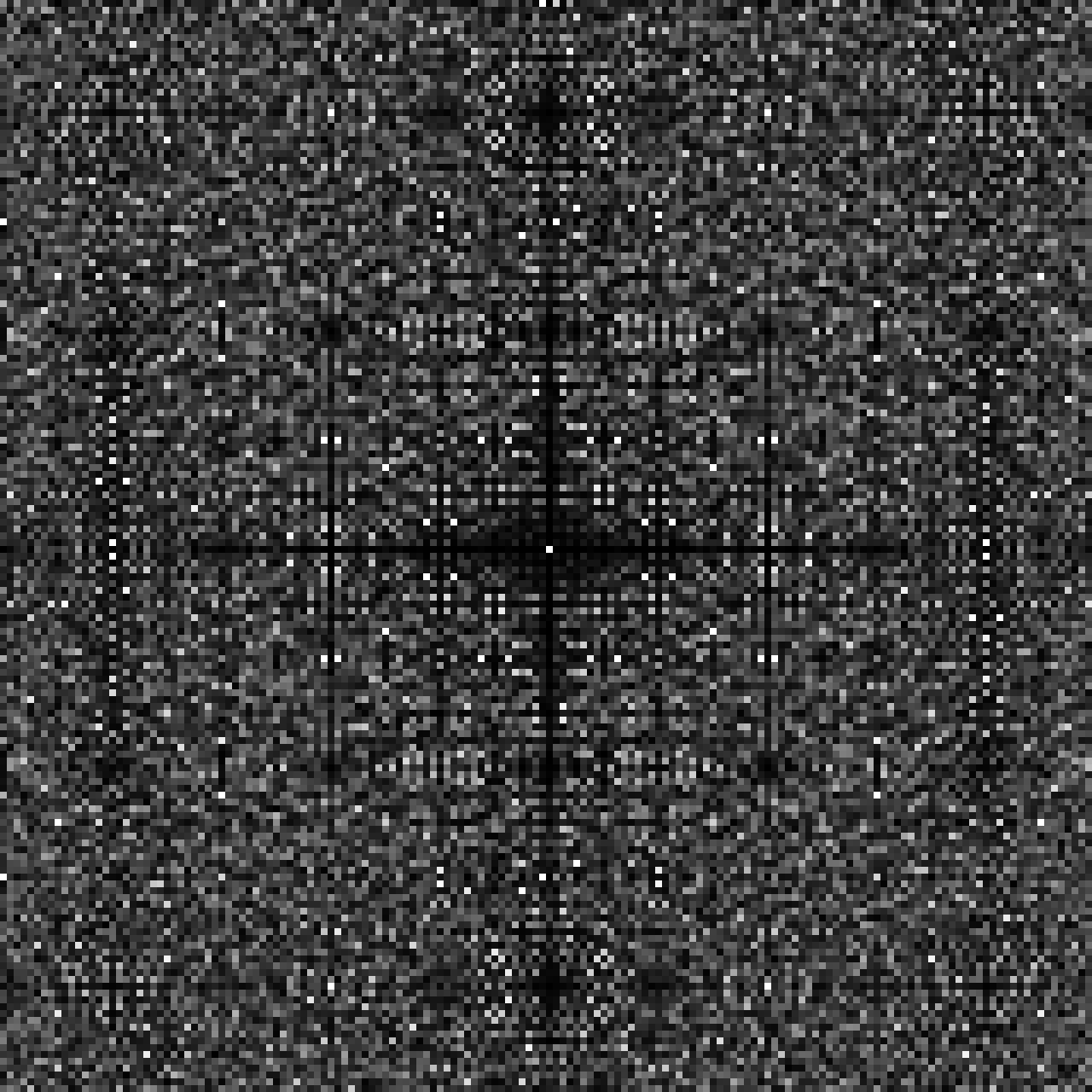}&%
        \includegraphics[width=1\unit]{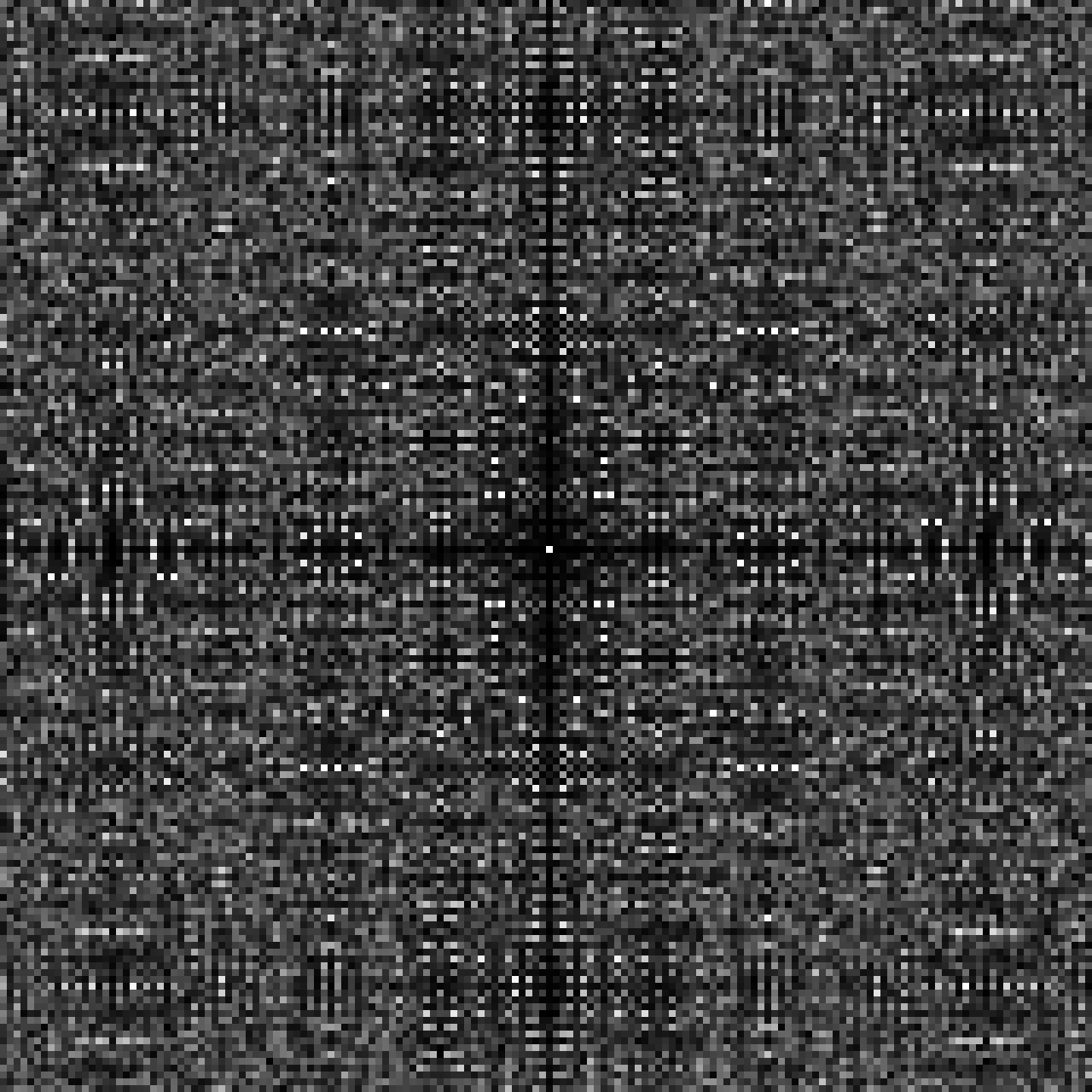}\\%
        \includegraphics[width=1\unit]{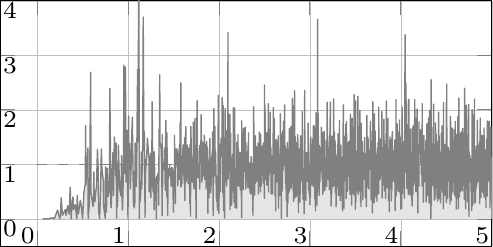}&%
        \includegraphics[width=1\unit]{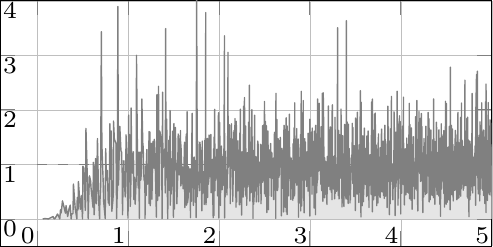}&%
        \includegraphics[width=1\unit]{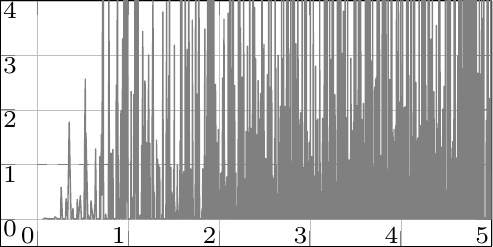}&%
        \includegraphics[width=1\unit]{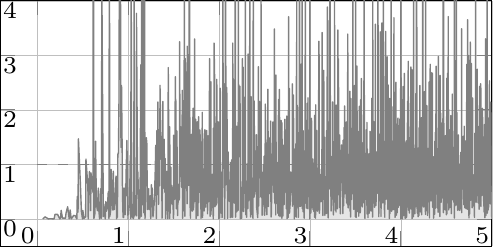}&%
        \includegraphics[width=1\unit]{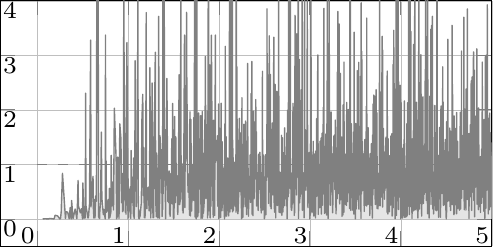}\\[2mm]%
        (a) Owen & (b) Burley & (c) XOR-Scrambled Sobol & (d) Random Digital & (e) $\xi$
    \end{tabular*}}
    \vspace{-2mm}
    \caption{%
        \label{fig:spectral}
        Two-dimensional and radial average spectral profiles of different dyadic sequences, showing (top) the frequency power spectrum of the point process, obtained by averaging 16k periodograms of 256 points, and (bottom) a periodogram of a randomly selected realization. We use radial profiles advocated by~\cite{Ahmed2022Gaussian}, where we average values for every distinct radius.
        The set for $\xi$-sequences includes all distinct sequences (up to XOR-scrambling) at 8-bit resolution.
        A randomly selected digital or $\xi$-sequence has a smoother periodogram than the XOR-Scrambled Sobol sequence but not Burley’s implementation of Owen-scrambling. However, the average of sufficiently many periodograms for $\xi$-sequence is smoother than both, and visually as smooth as Owen's scrambling.
    }
\end{figure*}


\section{Discussion and Conclusion}


In this paper, we presented new insights on digital dyadic sequences. \remove{}Our work provides a  comprehensive understanding of the design space and how to navigate it. 

We believe our work will lay the foundation for future work in systematically exploring the design space of higher-dimensional sequences and their impact on rendering. The importance of 2D sequences is that they often serve as building blocks for synthesizing high-dimensional sequences. Throughout the paper, we already discussed multiple implications of our work on rendering applications, and we are hopeful that the proposed $\xi$-sequences and Gray sequences will find their way into existing rendering frameworks.

\begin{acks}
  We are grateful to F.~Pillichshammer for bringing earlier proofs of Theorems~\ref{th-dyadic-pairs} and~\ref{th-progressive-pairs} (which we conceived independently) to our attention.
  Thanks to Mohanad Ahmed for his insightful discussions.
\end{acks}

\bibliographystyle{ACM-Reference-Format}
\bibliography{ld.bib}

\appendix

\section{Proof of Theorem~\ref{th-dyadic-pairs} (on the dyadic pairs)}
\label{app:proof of dyadic pairs}

We present \redit{the known proof} in detail to provide insight into the more complicated argument in the next appendix. We use two natural properties of dyadic pairs stated as lemmas below. \newremove{}

\begin{lemma} \label{l-5}
If a pair of matrices $(C_x,C_y)$ is dyadic and $M$ is any invertible matrix, then $(C_xM,C_yM)$ is also dyadic.
\end{lemma}

\begin{proof} Recall that the right multiplication by an invertible binary matrix is equivalent to a sequence of \emph{elementary transformations}, each being a swap of two columns or adding a column to another one. The transformation $(C_x,C_y)\mapsto(C_xM,C_yM)$ is equivalent to applying those simultaneously to the columns of $C_x$ and $C_y$. This leads to elementary transformations of any hybrid matrix~\eqref{eq:hybrid matrix} formed by the first $m-r$ rows of $C_x$ and the first $r$ rows of $C_y$. Since elementary transformations preserve the determinant, the hybrid matrix remains invertible, hence the pair remains dyadic. 
\end{proof}

Applying this lemma for $M=C_x^{-1}$, we can make the first matrix in the pair the identity matrix. In the latter case, we have the following 
characterization of dyadic pairs 
\redit{\cite[Lemma~3.1]{kajiura2018}}. 

\begin{lemma} \label{l-dyadic-minors}
A pair of matrices $( {I},C_y)$ is dyadic if and only if all leading principal minors of $C_yJ$ are nonzero.
\end{lemma}

\begin{proof} Let $C_y=(b_{ij})$. The determinants of all the possible hybrid matrices (formed by $r$ rows of $ {I}$ and $m-r$ rows of $C_y$)
\begin{multline*}
\left|
\begin{matrix}
  1         & \dots& 0        & 0          & \dots & 0          \\
  \vdots    &\ddots& \vdots   & \vdots     &\ddots & \vdots     \\
  0         & \dots& 1        & 0          &  \dots& 0          \\
  b_{11}    & \dots& b_{1r}   & b_{1,r+1}  & \dots & b_{1,m}    \\
  \vdots    &\ddots&\vdots    & \vdots     &\ddots & \vdots     \\
  b_{m-r,1} & \dots& b_{m-r,r}& b_{m-r,r+1}& \dots & b_{m-r,m}
\end{matrix}
\right|
=\\=
\left|
\begin{matrix}
   b_{1,r+1}  & \dots & b_{1,m}  \\
   \vdots     &\ddots & \vdots   \\
   b_{m-r,r+1}& \dots & b_{m-r,m}\\
\end{matrix}
\right|
\end{multline*}
are exactly the leading principal minors of $C_yJ$, for $r\ne m$. Here we applied the Laplace expansion along the first $r$ rows and used that $C_y {J}$ is obtained from $C_y$ by reversing the order of the columns.
\end{proof}


\begin{proof}[Proof of Theorem~\ref{th-dyadic-pairs}] By Lemma~\ref{l-5}, $(C_x,C_y)$ is dyadic if and only if  $( {I},C_yC_x^{-1})$ is dyadic. By Lemma~\ref{l-dyadic-minors}, the latter is dyadic if and only if $C_yC_x^{-1} {J}$ is progressive. This means that $C_yC_x^{-1} {J}=LU$ for some upper unitriangular matrix $U$ and lower unitriangular matrix $L$. The latter is equivalent to $C_y=LU {J}C_x$ with $C_x$ invertible.
\end{proof}



\begin{algorithm}[tb]
    \caption{
        \bluevar{LU-decomposition \cite{Cormen-etal}}
    }
    \label{alg:lu}
    \KwIn{
         \bluevar{an $m\times m$ progressive matrix $A=(a_{ij})$;}\\
    }
    \KwOut{\bluevar{lower and upper unitriangular matrices $L=(l_{ij})$ and $U=(u_{ij})$ such that $A=LU$;}
    }
    \bluevar{for $k\leftarrow1$ to $m$\\
    \qquad do for $i\leftarrow k$ to $m$\\
    \qquad\qquad\quad do   $l_{ik}\leftarrow a_{ik}$\\
    \qquad\qquad\qquad\hspace{0.5pt} $u_{ki}\leftarrow a_{ki}$ \\
    \qquad\quad\hspace{0.5pt} for $i\leftarrow k+1$ to $m$\\
    \qquad\qquad\quad do for $j\leftarrow k+1$ to $m$\\
    \qquad\qquad\qquad\quad do $a_{ij}\leftarrow a_{ij}-l_{ik}u_{kj}$\\
    Return $L$ and $U$.}
\end{algorithm}

\section{Proof of Theorem~\ref{th-progressive-pairs} (on the progressive pairs)}
\label{app:proof of progressive pairs}

\redr{We give a new short elementary proof of Theorem~\ref{th-progressive-pairs}. In contrast to the one by \cite{HoferSuzuki}, it does not involve 3D nets and works entirely in 2D. It is based on} 
properties of progressive pairs stated as lemmas below. 
\newremove{} Again, we start with the invariance under certain transformations, discovered by \cite{Faure02Another}.

\begin{lemma} \label{p-1}
If a pair of matrices $(C_x,C_y)$ is progressive,
$U$ is an upper unitriangular matrix, and $L_x,L_y$ are lower unitriangular matrices, then $(L_xC_xU,L_yC_yU)$ is also progressive.
\end{lemma}

\begin{proof} We use the same idea as in the proof of Lemma~\ref{l-5}.
The right multiplication by an upper unitriangular matrix is equivalent to a sequence of elementary transformations, each being adding a column to another one located to the right from it. The transformation $(C_x,C_y)\mapsto(C_xU,C_yU)$ is equivalent to applying those simultaneously to the columns of $C_x$ and $C_y$. Since a column is always added to a one to the right, this leads to elementary transformations of any hybrid matrix formed by the top-left corner $r\times k$ submatrix of $C_x$ and the top-left corner $(k-r)\times k$ submatrix of $C_y$. Since elementary transformations preserve the determinant, the hybrid matrix remains invertible, hence the pair remains progressive. 

Similarly, the left multiplication by a lower unitriangular matrix is equivalent to a sequence of elementary transformations, each being adding a row to another one located below it. The transformation $(C_x,C_y)\mapsto(L_xC_x,L_yC_y)$ now means \emph{different} elementary transformations of $C_x$ and $C_y$. But since a row is always added to a one below, this still leads to elementary transformations of any hybrid matrix.
Hence the pair remains progressive. 
\end{proof}

Applying this lemma, we can make the first matrix in the pair the identity matrix and the second one an upper unitriangular matrix. In the latter case, we have the following characterization of progressive pairs in terms of minors. A \emph{shifted leading minor} 
is a minor formed by the first $r$ rows and some $r$ consecutive columns for some $r$.

\begin{lemma} \label{l-4}
A pair of matrices $( {I},C_y)$ is progressive
if and only if all shifted leading minors of $C_y$ are nonzero.
\end{lemma}

\begin{proof} This is similar to Lemma~\ref{l-dyadic-minors}, but now the numbers $k$ and $r$ of rows taken from the matrices $I$ and $C_y$ need not sum up to $m$. Let $C_y=(b_{ij})$. The determinants of all the possible hybrid matrices
\begin{multline*}
\left|
\begin{matrix}
  1      & \dots& 0     & 0        & \dots & 0        \\
  \vdots &\ddots& \vdots& \vdots   &\ddots & \vdots   \\
  0      & \dots& 1     & 0        &  \dots& 0        \\
  b_{11} & \dots& b_{1k}& b_{1,k+1}& \dots & b_{1,k+r}\\
  \vdots &\ddots&\vdots & \vdots   &\ddots & \vdots   \\
  b_{r1} & \dots& b_{rk}& b_{r,k+1}& \dots & b_{r,k+r}
\end{matrix}
\right|
=
\left|
\begin{matrix}
   b_{1,k+1}& \dots & b_{1,k+r}\\
   \vdots   &\ddots & \vdots   \\
   b_{r,k+1}& \dots & b_{r,k+r}\\
\end{matrix}
\right|
\end{multline*}
are exactly the shifted leading minors of $C_y$.
\end{proof}

We now show that the condition that all shifted leading minors are non-zero is very restrictive \redit{\cite[Proposition~4]{hofer2010}.} This \redit{allows us to simplify the original proof of Theorem~\ref{th-progressive-pairs}}.

\begin{lemma} \label{p-2} 
There is a unique upper unitriangular matrix $U$ such that the pair $( {I},U)$ is progressive.
\end{lemma}

\begin{proof}[Proof of Lemma~\ref{p-2}]
The proof is by induction on the matrix size $m$. The base $m=1$ is obvious.

To prove the inductive step, use the characterization of progressive pairs  $( {I},U)$ from Lemma~\ref{l-4}.
Assume there is a unique upper unitriangular $(m-1)\times (m-1)$ matrix $U=(u_{ij})$ with all shifted leading minors nonzero (this is the inductive hypothesis). It suffices to prove that there is a unique way to add the entries $u_{k,m}$ for $k=1,\dots,m$ so that the resulting new shifted leading minors are also nonzero (the other added entries $u_{m,1}=\dots=u_{m,m-1}=0$). 

We do it by induction on $k$. The base $k=1$ holds because $u_{1,m}$ must be $1$ as a shifted leading minor of size $1\times 1$. Assume that $u_{1,m},\dots,u_{k,m}$ have already been determined. Consider the shifted leading minor $\Delta$ formed by the rows $1,\dots,k+1$ and columns $m-k,\dots,m$. Take the Laplace expansion along the last column:
$$
\Delta=
\Delta_1 u_{1,m}+\dots+\Delta_{k+1} u_{k+1,m},
$$
where the minor $\Delta_j$ is obtained from $\Delta$ by removing the $j$-th row and the last column.
Here
$
\Delta=\Delta_{k+1}=1
$
as shifted leading minors. Thus $$u_{k+1,m}=\Delta_1 u_{1,m}+\dots+\Delta_k u_{k,m}+1.$$ 
We have expressed $u_{k+1,m}$ in terms of the entries $u_{1,m},\dots,u_{k,m}$ and the entries of the top-left $(m-1)\times(m-1)$ submatrix. By the inductive hypothesis, all those entries are uniquely determined. Then $u_{k+1,m}$ is also 
determined, and the lemma follows by induction.
\end{proof}

This argument does not construct the unique matrix~$U$ explicitly. But we actually already know the answer.

\begin{corollary} \label{c-3} The Pascal matrix $ {P}$ is the unique upper unitriangular matrix such that the pair $( {I}, {P})$ is progressive.
\end{corollary}

\begin{proof}
The uniqueness has just been proved in Lemma~\ref{p-2}. The pair $( {I}, {P})$ is progressive by Theorem~\ref{th-digital dyadic sequence} and Table~\ref{tab:methods_overview_part2} but let us give a direct proof here. Consider the integer matrix $p_{ij}=\binom{j-1}{i-1}$. Applying \cite[Lemma~9]{GESSEL1985300} repeatedly, we conclude that each shifted leading minor of the matrix $(p_{ij})$ equals $1$. 
Hence by Lemma~\ref{l-4} the pair $( {I}, {P})$ is progressive. 
\end{proof}


\begin{proof}[Proof of Theorem~\ref{th-progressive-pairs}] 
Assume that $(C_x,C_y)$ is a progressive pair. Then 
both $C_x$ and $C_y$ are progressive matrices. Hence $C_x=L_xU_x$ and $C_y=L_yU_y$ for some upper unitriangular matrices $U_x,U_y$ and lower unitriangular matrices $L_x,L_y$.
By Lemma~\ref{p-1}, the pair 
$(L_x^{-1}C_xU_{x}^{-1},L_y^{-1}C_yU_{x}^{-1})=( {I},U_{y}U_{x}^{-1})$
is progressive. By Corollary~\ref{c-3}, $U_{y}U_{x}^{-1}= {P}$. We get  $C_x=L_xU_x$ and $C_y=L_y {P}U_x$, as required. 

Conversely,  each pair $(L_xU_x,L_y {P}U_x)$ is progressive by Corollary~\ref{c-3} and Lemma~\ref{p-1}.
\end{proof}


To count the number of possible characteristic matrices $ {C}=L_y {P}L_x^{-1}$, we need the following results.

\begin{lemma}\label{l-identities} 
We have $ {P}^2=( {P} {J})^3= {I}$ and
$ {P} {J} {P}= {J} {P} {J}$. 
\end{lemma}


\begin{proof} 
The identity $ {P}^2= {I}$ by \cite{Faure1982} is a compact form of the well-known identity \cite[(5.21) and (5.12)]{GKP}
$$
\sum_{j=1}^m\binom{j-1}{i-1}\binom{k-1}{j-1}=2^{k-i}\binom{k-1}{i-1}
\overset{\mod2}{=} {I}_{ik}.
$$
The identity $ {P} {J} {P}= {J} {P} {J}$ \redr{by \cite{kajiura2018} is 
a compact form of the} 
well-known identity \cite[(5.14) and (5.25)]{GKP} 
$$
\sum_{j=1}^m\binom{m-j}{i-1}\binom{k-1}{j-1}(-1)^{j-1}=\binom{m-k}{m-i}.
$$
The identity $( {P} {J})^3= {I}$ is a consequence of the previous two.
\end{proof}

\begin{lemma}\label{l-unique-decomposition}
For distinct pairs $(L_x,L_y)$ of lower unitriangular matrices, 
the matrices $L_y {P}L_x^{-1}$ are pairwise distinct.
\end{lemma}

\begin{proof}
The matrix  ${C}=L_y {P}L_x^{-1}$ uniquely determines $L_x$ and $L_y$ because
$$
 {C} {J}=L_y {P}L_x^{-1} {J}
=\underbrace{L_y {J} {P} {J}}_{L}\cdot
 \underbrace{ {P} {J}L_x^{-1} {J}}_{U}
$$
is an $LU$-decomposition of $ {C} {J}$, which is unique. Here we used the identity $ {P}= {J} {P} {J} {P} {J}$ (see Lemma~\ref{l-identities}) and
observed that for any upper unitriangular matrix $U'$ the matrix $ {J}U' {J}$ is lower unitriangular. 
\end{proof}





\begin{corollary}\label{cor-lp}
If $m$ is a power of $2$, then under notation~\eqref{eq-C_LP} and~\eqref{eq-LP} we have $PJL^{-1}_{\bluevar{\mathrm{LP}}}=\bluevar{U_{\bluevar{\mathrm{LP}}}J}$.
\end{corollary}

\begin{proof}
First, analogously to the proof of Lemma~\ref{l-identities}, we get $L^{-1}_{\bluevar{\mathrm{LP}}}=L_{\bluevar{\mathrm{LP}}}$ by the identity
\cite[(5.21) and (5.12)]{GKP} 
$$
\sum_{j=2}^{m}\binom{i-2}{j-2} \binom{j-2}{k-2}=2^{i-k} \binom{i-2}{k-2}\overset{\mod2}{=}I_{ik}.
$$
Second, $PJL_{\bluevar{\mathrm{LP}}}=\bluevar{U_{\bluevar{\mathrm{LP}}}J}$ by the identity \cite[(5.26)]{GKP} 
$$
\sum_{j=2}^{m}\binom{m-j}{i-1} \binom{j-2}{k-2}=\binom{m-1}{i+k-2}\overset{\mod2}{=}\bluevar{(U_{\bluevar{\mathrm{LP}}})_{i,m-k+1}},
$$
where the latter holds since $m$ is a power of $2$.
So $PJL^{-1}_{\bluevar{\mathrm{LP}}}=\bluevar{U_{\bluevar{\mathrm{LP}}}J}$. 
\end{proof}

\end{document}